\documentclass[twoside,11pt]{article}
\usepackage{graphicx} 
\usepackage{amsmath, amssymb, bm,bbm,amsthm}
\usepackage[toc,page]{appendix}
\usepackage{setspace}
\usepackage{color, xcolor}
\usepackage{natbib}
\usepackage{algorithm}
\usepackage{algpseudocode}
 \usepackage{multirow}
\usepackage{subfigure}
\allowdisplaybreaks[4]
\makeatletter
\newenvironment{breakablealgorithm}
  {
   \begin{center}
     \refstepcounter{algorithm}
     \hrule height.8pt depth0pt \kern2pt
     \renewcommand{\caption}[2][\relax]{
       {\raggedright\textbf{\ALG@name~\thealgorithm} ##2\par}%
       \ifx\relax##1\relax 
         \addcontentsline{loa}{algorithm}{\protect\numberline{\thealgorithm}##2}%
       \else 
         \addcontentsline{loa}{algorithm}{\protect\numberline{\thealgorithm}##1}%
       \fi
       \kern2pt\hrule\kern2pt
     }
  }{
     \kern4pt\hrule\relax
   \end{center}
  }
\makeatother
\usepackage{enumitem}
\usepackage{algpseudocode}
\usepackage{tabularx}
\makeatletter
\def\@biblabel#1{\hspace*{-\labelsep}}
\makeatother

\newcommand{\ve}{\mathrm{Vec}}
\newcommand{\var}{\mathrm{Var}}
\newcommand{\cov}{\mathrm{Cov}}
\newcommand{\1}{\mathbbm{1}}
\newcommand{\tr}{\mathrm{tr}}

\DeclareMathOperator*{\argmin}{arg\,min}
\usepackage{jmlr2e}

\definecolor{codegreen}{rgb}{0,0.6,0}
\definecolor{codegray}{rgb}{0.5,0.5,0.5}
\definecolor{codepurple}{rgb}{0.58,0,0.82}
\definecolor{backcolour}{rgb}{0.95,0.95,0.92}
\usepackage{listings}
\lstdefinestyle{mystyle}{
	backgroundcolor=\color{backcolour},   
	commentstyle=\color{codegreen},
	keywordstyle=\color{magenta},
	numberstyle=\tiny\color{codegray},
	stringstyle=\color{codepurple},
	basicstyle=\ttfamily\footnotesize,
	breakatwhitespace=false,         
	breaklines=true,                 
	captionpos=b,                    
	keepspaces=true,                 
	numbers=left,                    
	numbersep=5pt,                  
	showspaces=false,                
	showstringspaces=false,
	showtabs=false,                  
	tabsize=2
}

\usepackage{lastpage}
\jmlrheading{??}{????}{1-\pageref{LastPage}}{????; Revised ????}{????}{??-????}{Dehao Dai and Yunyi Zhang}
\ShortHeadings{On Statistical Inference for High--Dimensional Binary Time--Series }{Dehao Dai and Yunyi Zhang}
\firstpageno{1}

\begin{document}
	\title{On Statistical Inference for High-Dimensional Binary Time Series}
		
\author{\name Dehao Dai \email ddai@ucsd.edu \\
		\addr Department of Mathematics \\
		University of California San Diego\\
		La Jolla, CA 92037, USA
		\AND
\name Yunyi Zhang \email zhangyunyi@cuhk.edu.cn\\
		\addr 
        School of Data Science\\
        The Chinese University of Hong Kong, Shenzhen\\
		Shenzhen, Guangdong 518172, China\\
		}
	
 \editor{}
	\maketitle

	\begin{abstract}
          The analysis of non-real-valued data, such as binary time series, has attracted great interest in recent years. This manuscript proposes a post-selection estimator for estimating the coefficient matrices of a high-dimensional generalized binary vector autoregressive process and establishes a Gaussian approximation theorem for the proposed estimator. Furthermore, it introduces a second-order wild bootstrap algorithm to enable statistical inference on the coefficient matrices. Numerical studies and empirical applications demonstrate the good finite-sample performance of the proposed method.

	\end{abstract}


\section{Introduction}
Time series data frequently appear in practice, making their analysis a common and essential task in many real-life applications, like those introduced in   
\cite{MR2172368}, \cite{politis2019time}, \cite{MR4154894},  \cite{MR4270034}, \cite{MR4356579}, \cite{10.1371/journal.pbio.3002758},  \cite{MR4698614},  among others.  After practitioners collect a time series dataset, various models, such as ARMA or GARCH processes, are often resorted to summarize the conditional mean or volatility information of the data. Furthermore, after fitting a suitable model, statisticians may perform statistical inference on the underlying data generating process, like those in Section 15.4 of \cite{MR2172368}, or perform mean or probabilistic forecasting, like those in \cite{ZHANG2025112611}.

Despite wide applications and advantages of the existing models, an issue with these models is that they are mainly designed for time series with real value support---making them unsuitable for time series data that exhibit special characteristics. For example, if the time series dataset has binary support, fitting an autoregressive model does not guarantee that maintain such characteristics. However, constraints on the support of practical time series datasets often exist, which may arise from the physical meaning of the data.  For example, \cite{keenan1982time} mentioned that binary-value time series can arise when statisticians consider the dynamic occurrence of events or changes in states; \cite{MR3773399} introduced a count time series model for hurricane data; other research like \cite{MR2957867}, \cite{fokianos2009poisson}, \cite{MR2981620}, and the reference therein, also raised similar issues.

This manuscript leverages the work of \cite{jentsch2022generalized} to introduce a high-dimensional generalized binary vector autoregressive (gbVAR) process for high-dimensional binary-valued vector time series data. After incorporating the sparsity assumption of the parameter matrices, it introduces a post-selection estimator to estimate the parameter matrices of the model and derives the associated Gaussian approximation for the proposed estimator. The complex spatial and temporal dependence of the vector time series data causes the proposed estimator to exhibit complex and hard-to-estimate variances and covariances, making it unrealistic to perform statistical inference through manual calculations.  Concerning this, we propose a bootstrap algorithm, named the second-order wild bootstrap,  to assist statisticians in performing statistical inference on the parameter matrices. This manuscript also theoretically justifies the validity of the proposed bootstrap algorithm.

It is common for modern time series to simultaneously exhibit both high dimensionality and constraints on their support. Given the prevalence of non-real-valued-supported data---such as binary vector time series---in various application domains, including \cite{kharin2018statistical,darolles2019bivariate,franchi2024theory}, our work is expected to be of interest to practitioners in these domains.

This paper is organized as follows: Section \ref{section.related_literature} makes a thorough review of the analysis of non-real-value time series data.  Section \ref{sec 2} introduces gbVAR Models and their stochastic probability for low-dimensional structures. Section \ref{sec 3} proposes the transition matrix estimation procedure and bootstrap inference for binary time series. Section \ref{sec 4} presents asymptotic analysis based on the functional dependence measure and Gaussian approximation. Section \ref{sec 5} provides numerical experiments, and Section \ref{sec 6} summarizes discussions and conclusions. Technical proofs are deferred to the online supplement \cite{Supplement}.

\textbf{Frequently used notations: }
In this part, we introduce some notation. For a random variable $X$ and $q\geq 1$, define $\|X\|_q = (\mathbb{E}(|X|^q))^{1/q}$. For a vector $u = (u_1,\ldots, u_d)^\top \in \mathbb{R}^d$, define $|u|_0 = \sum_{j=1}^d \1\{u_j \neq 0\}$, $|u|_1 = \sum_{j=1}^d |u_j|$, $|u|_2 = \sqrt{\sum_{j=1}^d u_j^2}$ and $|u|_\infty = \max_{j}|u_j|$. For a matrix $A = (a_{ij})_{i,j=1}^d \in \mathbb{R}^{d \times d}$, define $|A|_1 = \max_{1\leq j\leq d}\sum_{i=1}^d |a_{ij}|$, $|A|_\infty = \max_{1\leq i\leq d}\sum_{j=1}^d |a_{ij}|$, element-wise maximum norm$|A|_{\max} = \max_{i,j}|a_{ij}|$, and the spectral norm $|A|_2 = \sqrt{\lambda_{\max}(A^\top A)}$ where $\lambda_{\max}$ means the largest eigenvalue. Write the $d\times d$ identity matrix as $I_d$. For two sequences of positive numbers $\{a_n\}$ and $\{b_n\}$, define $a_n \lesssim b_n$ if there exists some constant $C>0$ such that $a_n/b_n \leq C$ as $n \rightarrow \infty$, and define $a_n \asymp b_n$ if $a_n \lesssim b_n$ and $b_n \lesssim a_n$; also define $a_n =O(b_n)$ if there exists a positive constant $C>0$ such that $|a_n|\leq C|b_n|$ for all $n$ and $a_n = o(b_n) $ if $\lim_{n\rightarrow \infty} a_n/b_n =0$. For two random variable sequences $X_n$ and $Y_n$, define $X_n = O_\mathbb{P}(Y_n)$ if for any $0<\varepsilon<1$, there exists a constant $C_\varepsilon$ such that $\sup_{n} \mathbb{P}(|X_n |\geq C_\varepsilon |Y_n|)\leq \varepsilon $; and $X_n = o_\mathbb{P}(Y_n)$ if $X_n/Y_n \stackrel{P}{\rightarrow} 0$. We also use $c, C, \ldots$ to denote positive constants whose values may vary from context to context.

\section{Related literature}
\label{section.related_literature}

\textbf{Analysis of non-real-value-supported time series.} There have been numerous attempts in the literature to model time series whose support is not the real axis. To mention a few, the work of \cite{mckenzie1985some, al1987first} introduced the integer-value autoregressive processes, which replaced the conventional scalar multiplication with the so-called ``thinning operators,''  to accommodate integer-value time series. \cite{mckenzie1988some, al1988integer} later extended this approach to develop integer-value moving average processes.  More recent studies on integer-value autoregressive processes include \cite{nastic2016random, zhang2023new, ristic2016binomial, huh2017monitoring,chen2020two},  which introduced random coefficients, new innovation structures, and models incorporating conditional heteroskedasticity. In addition to the aforementioned work, we refer readers to \cite{bourguignon2017inar, weiss2021stationary, fokianos2022statistical, chen2022new, piancastelli2023flexible}, for further exploration of this topic.

Other attempts leveraged generalized linear models, such as the Logistic regression model, to accommodate integer-value time series. Examples include  \cite{liang1989class, MR1965687} that developed logistic regression models for multivariate binary time series, \cite{fokianos2009poisson, christou2014quasi, ahmad2016poisson} introduced the quasi-likelihood inference for generalized count autoregressive models, while \cite{fokianos2004partial} proposed a partial likelihood inference method for modeling the generalized time series.  \cite{zhu2017network} introduced a neural autoregressive model, which was later extended to count-valued time series by \cite{armillotta2024count}. Other nonlinear models for count time series have been studied in \cite{fokianos2011log, fokianos2012count, christou2014quasi, wang2014self, dunsmuir2015generalized, davis2016theory}. More recent and applied works can be found in the literature, such as \cite{hall2018learning, amillotta2022generalized, jia2023latent, pumi2024unit}.

\textbf{Analysis of high-dimensional time series.} Modern-era data often exhibit complex structures. In particular, due to the advancement in data collection techniques, the dimension of the observations can be comparable to, or exceed, the number of observations. The literature has extensively explored high-dimensional real-valued time series, as seen in works by \cite{jentsch2015covariance, zhang2017gaussian, MR4206676, krampe2023structural, MR4718536}, among others. Compared to the rich literature in high-dimensional real-valued time series, to the best of the authors' knowledge, relatively little research has been conducted on high-dimensional time series with non-real number support. Recent examples include \cite{fokianos2020multivariate, guo2023consistency, duker2023high}, and the references therein.

\textbf{Applications to network data.} One potential application of our study lies in the analysis of dynamic networks. Practitioners often assign binary random variables to the edges of network data when performing statistical inference.  Classical approaches assume independence among these random variables, thereby ignoring the potential spatial and temporal dependence across edges. Recent studies have tended to solve such a problem. For example, \cite{jiang2023autoregressive} introduced a first-order autoregressive dynamic network process to introduce temporal dependence, \cite{hanneke2010, krivitsky2014, Leifeld2018} developed temporal exponential family random-graph models, \cite{suveges2023, chang2024autoregressive} respectively proposed autoregressive network models based on logistic regressions and Markov chains. We also refer readers to  \cite{mantziou2023gnar, zhu2025autoregressive}, and references therein.

\section{gbVAR Model and its statistical properties}
\label{sec 2}
This section leverages the work of \cite{jentsch2022generalized} to introduce the gbVAR model for high-dimensional binary-valued vector time series. In addition, it explores several statistical properties of this model. Given a $d\times (p+1)d$ parameter matrix for $1\leq q\leq p$,
\begin{equation}
\label{eq: param P}
    \mathcal{P} :=[\mathcal{A}^{(1)}, \mathcal{A}^{(2)}, \ldots, \mathcal{A}^{(p)}, \mathcal{B}], \text{ where }\mathcal{A}^{(q)} = (\alpha_{kl}^{(q)})_{kl = 1, \ldots, d}\ \text{ and } \mathcal{B} = diag(\beta_1,\cdots, \beta_d)
\end{equation}
is a diagonal matrix with positive definite $\mathcal{B}$, following \cite{jentsch2022generalized}, we define the ``counterpart matrix'' $\mathcal{P}_{|\cdot|}$ as follows: 
\begin{equation*}
    \begin{aligned}
\mathcal{P}_{|\cdot|}=[\mathcal{A}_{|\cdot|}^{(1)}, \mathcal{A}_{|\cdot|}^{(2)}, \ldots, \mathcal{A}_{|\cdot|}^{(p)}, \mathcal{B}],\text{ where }\mathcal{A}_{|\cdot|}^{(q)} = (|\alpha_{kl}^{(q)}|)_{kl = 1, \ldots, d}
    \end{aligned}
\end{equation*}
where $|\alpha_{kl}^{(q)}|$ represents the absolute value of the matrix element $\alpha_{kl}^{(q)}.$ Suppose the following constraints for the  parameters $\alpha_{kl}^{(q)}$ and $\beta_k$ in \eqref{eq: param P}:
\begin{equation}
    \sum_{i =1}^p \sum_{l=1}^d |\alpha^{(i)}_{kl}|+\beta_k = 1.
    \label{eq.constraint_counterpart}
\end{equation} 
With this constraint,  all elements in $\mathcal{P}_{|\cdot|}$ are positive.

Define $P_t\in\mathbb{R}^{d \times d(p+1)}, t\in\mathbb{Z}$ as mutually independent random matrices such that each row of the matrix $P_{t, k \cdot},$   $k = 1, \ldots, d,$   follows multinomial distribution with size of trials 1 and event probabilities $\mathcal{P}_{|\cdot|,k\cdot},$ specifically 
$$
P_{t, k\cdot}= (P_{t,k1},\cdots, P_{t,kd(p+1)}) := \left[a_{t,k\cdot}^{(1)},\ldots, a_{t,k\cdot}^{(p)} , b_{t, k\cdot}\right] \sim Mult(1;\mathcal{P}_{|\cdot|, k\cdot})
$$
with the probability mass function 
\begin{equation*}
\begin{aligned}
    \mathbb{P}(P_{t,k\cdot} = (c_1,\cdots, c_{d(p+1)})) = \beta_k^{c_{kd(p+1)}}\times \prod_{s = 1}^p\prod_{q = 1}^d \vert\alpha^{(s)}_{kq}\vert^{c_{(s - 1)\times d + q}},
\end{aligned}
\end{equation*}
where $c_{j} \in \{0,1\},$ $1\leq j\leq d(p+1).$  We further assume that $P_{t, k_1 \cdot}$ is independent of $P_{t, k_2 \cdot}$ for $k_1\neq k_2.$ We define the matrices $A_t^{(i)}, B_t\in\mathbb{R}^{d\times d}$ such that
$$
A_t^{(i)} = (a^{(i)}_{t,kl})_{k,l= 1, \ldots, d}\quad\text{and}\quad B_t = (b_{t,kl})_{k,l=1, \ldots, d} = diag(b_{t, 11}, \ldots, b_{t, dd}),
$$
with these notations, we have $P_t = \left[A_t^{(1)}, \ldots, A_t^{(p)}, B_t\right].$

Innovations should have special marginal distributions to maintain the binary support of the vector time series. We define a sequence of independent and identically distributed (i.i.d.) innovations $e_t\in\mathbb{R}^d, t \in \mathbb{Z}$ such that each element $e_{t,i}\in\{0,1\}, i = 1,\cdots, d,$ and its marginal distribution is a Bernoulli distribution with parameter $\mu_{e,i} \in (0,1),$ which is of the form 
\begin{equation}
        \mathbb{P}\left(e_{t,i} = 1\right) = \mu_{e,i}\quad \text{and} \quad\mathbb{P}\left(e_{t,i} = 0\right) = 1 - \mu_{e,i},
        \label{eq.binary_e}
\end{equation}
where each $\mu_{e,i} \in (0,1) $ are parameters of the model. From \eqref{eq.binary_e}, we have 
$
\mathbb{E}\left[e_{t,i}\right] =\mu_{e,i}.
$
Furthermore, define the variance-covariance matrix of $e_t$ as $\mathbf{\Sigma}_e= (\sigma_{e, {ij}})_{i,j = 1, \ldots, d},$ we have 
$$
\sigma_{e, ii}=  \mathrm{Var}\left(e_{t,i}\right) = \mu_{e, i}(1-\mu_{e,i}).
$$
In the manuscript, we do not constrain the joint distribution of $e_{t,1},\cdots, e_{t,d},$ so $e_{t,i}$ may have non-zero correlations $\sigma_{e, ij}$ with $e_{t,j}$ for some $i\neq j.$ However, we require $\mu_{e,i}\in (0,1)$ so that $\mathbf{\Sigma}_e$ is positive-definite. We assume that $(e_t)_{t \in \mathbb{Z}}$ are independent of $P_t = \left[A_t^{(1)}, \ldots, A_t^{(p)}, B_t\right].$

With these notations, we are able to define the gbVAR($p$) processes as in Definition \ref{def.gbVAR}.

\begin{definition}[gbVAR($p$) processes]
The time series $(X_t)_{t \in \mathbb{Z}}$ satisfying the following iteration  
\begin{equation}
\label{eq: gbvarp}
    X_t = \sum_{i = 1}^p \left(A_t^{(+, i)} X_{t-i}+A_t^{(-,i)} \1_d\right) +B_t e_t, \quad t \in \mathbb{Z},
\end{equation}
where the parameter matrices $A_t^{(+, i)} = (a_{t,kl}^{(+,i)})_{k,l=1,\ldots, d},$ $A_t^{(-,i)} = (a_{t,kl}^{(-,i)})_{k,l=1,\ldots, d}$ satisfy 
\begin{equation*}
    \begin{aligned}
        a_{t,kl}^{(+,i)}&:= a_{t,kl}^{(i)}\left(\1\{\alpha_{kl}^{(i)}\geq 0\}-\1\{\alpha_{kl}^{(i)}<0\}\right),\\
        a_{t,kl}^{(-,i)}&:= a_{t,kl}^{(i)}\1 \{\alpha_{kl}^{(i)}<0\},
    \end{aligned}
\end{equation*}
and the $X_t$ is defined as a generalized binary vector Autoregressive time series of order $p$ (gbVAR$(p)$). In \eqref{eq: gbvarp}, $\1_d\in\mathbb{R}^d$ denotes the $d$-dimensional vector with all elements 1.
\label{def.gbVAR}
\end{definition}

We take $p = 1,$ where  
$$
X_t = A_t^{(+)}X_{t-1}+A_t^{(-)}\1_d+B_t e_t =
\left[
\begin{matrix}
    A_t^{(+)} &  A_t^{(-)} & B_t
\end{matrix}
\right]
\left[
\begin{matrix}
    X_{t-1}^\top & \1_d^\top & e_t^\top
\end{matrix}
\right]^\top
$$
as an illustrative example to explain why \eqref{eq: gbvarp} leads to a binary process. Since each row of the matrix $[A_t,B_t]$ follows a  Multinomial distribution with size $1$ and probability $[\mathcal{A}_{|\cdot|},\mathcal{B}],$ each row of the composited matrix 
$\left[A_t\ B_t\right]$ only contains one ``1'', and remainders are ``0''. For each row $k,$ if $a_{t,kl}^{(+)} = 1$ and $a_{t,kl}^{(-)} =0$ with $\alpha_{kl}^{(i)}\geq 0,$ then 
$$
X_{t,k} = a_{t,kl}X_{t-1,l} + B_{t,kk}e_{t,k}, 
$$
which has binary support if $X_{t-1,l}\in\{0,1\}.$ Otherwise, if $a_{t,kl}^{(+)} = -1$ and $a_{t,kl}^{(-)} =1$  with $\alpha_{kl}^{(i)}< 0,$ then 
$$
X_{t,k} = a_{t,kl}(1 - X_{t-1,l}) + B_{t,kk}e_{t,k},
$$
which also has binary support if $X_{t-1,l}\in\{0,1\}.$

\begin{remark}
    The reason for setting up the constraint \eqref{eq.constraint_counterpart} is to ensure that each row in the counterpart matrix serving as a valid probability vector of a multinominal distribution, namely
   $$ P_{t,kj}\in\{0,1\}\ \text{and }\sum_{j = 1}^{d(p+1)}P_{t,kj} = 1.$$
\end{remark}

\begin{example}[3-variate gbVAR(1) model] 
\label{ex: 3-variate}
Let $(X_t)_{t \in \mathbb{Z}}$ be from a 3-dimensional gbVAR(1) model
\begin{equation*}
    X_t = A_t^{(+)}X_{t-1}+A_t^{(-)}\1_d +B_te_t,
\end{equation*}
with parameter matrix $\mathcal{P}=[\mathcal{A}, \mathcal{B}]$ and corresponding matrix $\mathcal{P}_{|\cdot|} = [\mathcal{A}_{|\cdot|}, \mathcal{B}]$, where
\begin{equation*}
\begin{aligned}
\mathcal{A} &= \begin{pmatrix}
      0.15    &-0.25  & 0.49\\
-0.19  &0.27  & 0.28\\
 0.17 &-0.37  &0.21
    \end{pmatrix},  \mathcal{A}_{|\cdot|} = \begin{pmatrix}
      |0.15|    &|-0.25|  & |0.49|\\
|-0.19|  &|0.27|  & |0.28|\\
 |0.17| &|-0.37|  &|0.21|
    \end{pmatrix},\\ 
    \text{ and }
    \mathcal{B} &= diag(0.11, 0.26,0.23).
\end{aligned}
\end{equation*}
Suppose that the innovation $(e_t)_{t \in \mathbb{Z}}$ consists of three independent Bernoulli random processes $(e_{t,i})_{t \in \mathbb{Z}}$ with $\mu_{e} = (0.48, 0.52,0.47)^\top$. By extending the original positive parameter matrix $\mathcal{A}$ to negative entries, the gbVAR(1) model is more flexible to describe different scenarios. 
\end{example}

\begin{example}[gbVAR(1) random graph]
In this example, we use a gbVAR(1) process $X_t, t\in\mathbf{Z},$ to model the dynamics (presence or absence) of edges in a sequence of undirected networks with $n = 10$ fixed nodes. The dimension of $X_t$ here is 45, representing $n(n-1)/2$ possible edges for these  networks. The parameter matrix $\mathcal{P} = [\mathcal{A}, \mathcal{B}]$ in the model is 
\begin{align*}
    \mathcal{A}_{i,i-1} = \mathcal{A}_{i,i+1} = 0.3\quad \text{and}\quad \mathcal{B} = diag(0.7, 0.4, \ldots, 0.4, 0.7),
\end{align*}
and the innovation $(e_t)_{t\in \mathbb{Z}}$ follows from independent Bernoulli distribution with $\mu_e = 0.5\1_{d}$.
The resulting network is displayed in Figure \ref{fig: example gbVAR}.
    \begin{figure}[htbp]
        \centering
    \includegraphics[width = \textwidth]{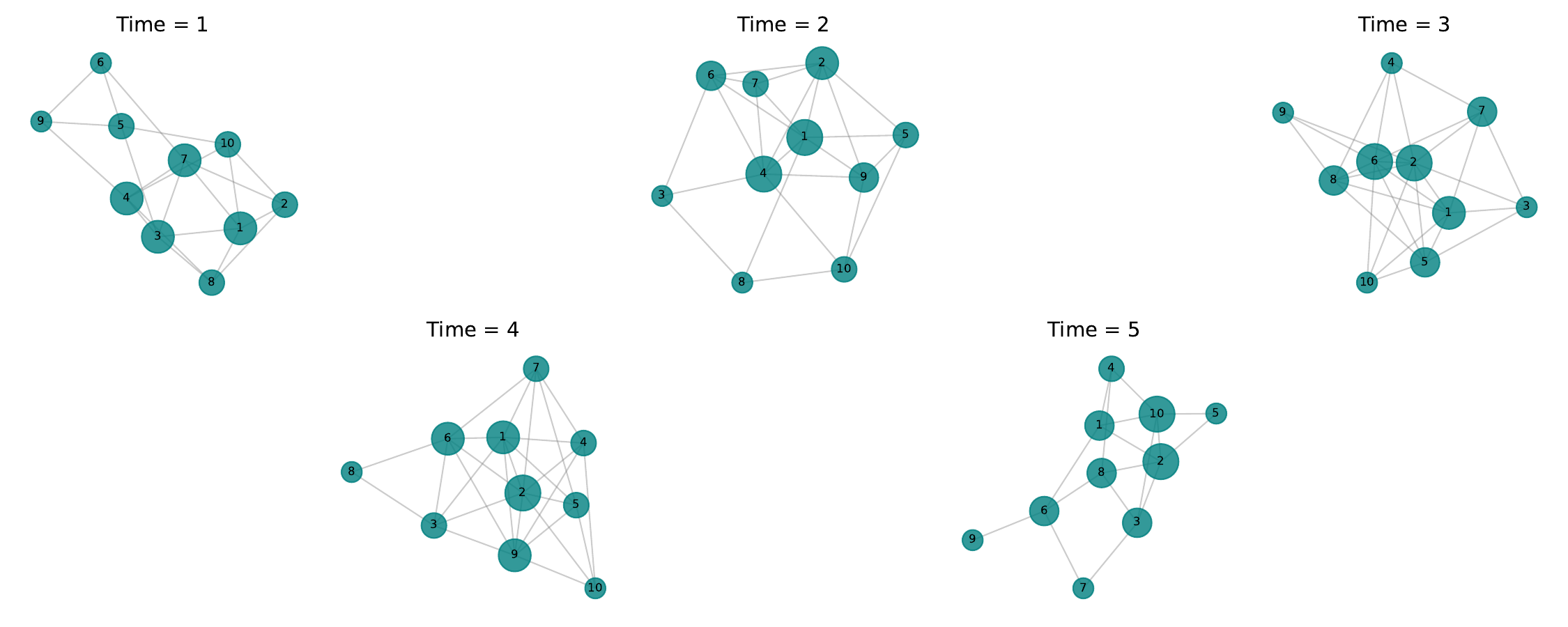}
        \caption{Realization of gbVAR(1) process $(X_t)_{t=1, \ldots, 5}$ with 10 nodes: the size of a node is proportional to its degree after discretization.}
        \label{fig: example gbVAR}
    \end{figure}
\end{example}

The remainder part of this section focuses on statistical properties of a lag-1 gbVAR model. While \cite{jentsch2019generalized, jentsch2022generalized} addressed the estimation of autoregressive coefficient matrices for binary vector autoregressive processes, they did not consider cases where the dimension $d$ is moderately or very large, and did not perform statistical inference. Our work complements their studies by extending the analysis to high-dimensional settings and incorporating sparse structures in the statistical inference. Additionally, we establish simultaneous confidence intervals for the entire autoregressive coefficient matrices.

\subsection{Stochastic properties of gbVAR(1) process}
Our discussion starts with the expectation of $A_t^{(+)}$ and $A_t^{(-)}$. By construction, $\mathbb{E}\left[A_t\right] =\mathcal{A}_{|\cdot|}$ and $\mathbb{E}\left[B_t\right] =\mathcal{B}$. Since $\mathcal{A}_{|\cdot|} = \mathcal{A} +2\mathcal{A}^{(-)} $, the expectations of $A_t^{(+)} $ and $A_t^{(-)}$ become 
$$
\mathbb{E}\left[A_t^{(+)}\right] = \Big[\alpha_{kl}\Big]_{k,l=1, \ldots, d}:=\mathcal{A} \text{ and } \mathbb{E}\left[A_t^{(-)}\right] = \Big[|\alpha_{kl}|\mathbbm{1}\{\alpha_{kl}<0\}\Big]_{k,l=1, \ldots, d}:=\mathcal{A}^{(-)}.
$$
Proposition \ref{thm: MA rep} below  presents a result similar to the moving-averaging representation commonly seen in real-valued vector autoregressive models, as introduced in \cite{MR2839251}. The MA representation of \eqref{eq: gbvarp} can be established under the following condition:
\begin{equation}
\label{eq: condition stationary}
    det(I_d - \mathcal{A}_{|\cdot|}z) \neq 0 \quad \forall z \in \{z \in \mathbb{C}: |z|\leq 1\},
\end{equation}
indicating that $\mathcal{A}_{|\cdot|}$ does not have unit root.

\begin{proposition}[Moving-average representation of gbVAR(1) process ]
\label{thm: MA rep}
Let $(X_t)_{t\in \mathbb{Z}}$ be a $d$-dimensional gbVAR$(1)$ process satisfying \eqref{eq: gbvarp} for all $t\in \mathbb{Z}$ and suppose condition \eqref{eq: condition stationary} holds true. Then the gbVAR(1) model has a gbVMA$(\infty)$-type representation
\begin{equation}
\label{eq: gbVMA}
    X_t = (X_{t1},\ldots, X_{td})^\top = \sum_{i = 0}^\infty \zeta_{t,i-1}\eta_{t-i}, \quad t \in \mathbb{Z},
\end{equation}
where
\begin{equation}
\label{eq: zeta, eta}
    \zeta_{t,i} = \prod_{j=0}^iA_{t-j}^{(+)}, \quad \eta_{t-i}:=A_{t-i}^{(-)}\1_d+B_{t-i}e_{t-i}.
\end{equation}
The convergence result holds true under $L_1$ norm.
\end{proposition}

\begin{remark}
The case of $p > 1$ in equation \eqref{eq: gbvarp} can be addressed by  stacking the data and forming a $pd\times 1$ VAR(1) process $Z_t = 
[
\begin{matrix}
    X_t^{\top} & X_{t-1}^\top & \cdots & X_{t -p + 1}
\end{matrix}
]^\top$. By leveraging Proposition \ref{thm: MA rep}, the gbVMA$(\infty)$-type representation can be established, provided that all roots of the characteristic matrix polynomial lie outside the unit circle:
\begin{equation*}
    det(I_d - \mathcal{A}_{|\cdot|}^{(1)}z - \cdots - \mathcal{A}_{|\cdot|}^{(p)}z^p) \neq 0 \quad \forall z \in \{z \in \mathbb{C}: |z|\leq 1\},
\end{equation*}
and the gbVMA$(\infty)$-type representation is 
$$X_t = J\tilde{X}_t = J\left(\sum_{i=0}^\infty \tilde{\zeta}_{t, i-1}\tilde{\eta}_{t-i}\right), \quad t \in \mathbb{Z},$$
where $J:= [\1_d, 0_{d\times dp}]$.
\end{remark}
A common way to build up model identification procedure  in vector time series analysis literature involves the Yule-Walker equation, as introduced in \cite{brockwell1991time}. Proposition \ref{thm.Yule_Walker} leverages similar idea to identify $\mathcal{A}$.

\begin{proposition}[Yule-Walker equations for gbVAR($1$) models] 
    \label{thm.Yule_Walker}
    Let $(X_t)_{t \in \mathbb{Z}}$ be a $d$-dimensional gbVAR$(1)$ process that satisfies \eqref{eq: condition stationary}. Then for all $k \in \mathbb{N}^+$, we have $\bm{\Sigma}^{(k)} = \mathcal{A}\bm{\Sigma}^{(k-1)}$ leading to Yule-Walker equation
\begin{equation}
\label{eq: Yule-Walker}
\begin{aligned}
    \bm{\Sigma}^{(1)} =\mathcal{A}\bm{\Sigma}^{(0)}, \text{ where } \bm{\Sigma}^{(k)} = \mathbb{E}\left[(X_{t+k} - \mu)(X_{t}-\mu)^\top\right], k = 0,1 
\end{aligned} 
\end{equation}

\end{proposition}

\begin{remark}
Notice that the stationary mean of the process is 
$$
\mu := \mathbb{E}\left[X\right] = \left(I-\mathcal{A}\right)^{-1}\left(\mathcal{A}^{(-)}\mathbbm{1}_d +\mathcal{B} \mu_e\right)\neq 0,
$$
which can be estimated by the sample mean $\widehat{\mu}:=\Bar{X} = 1/n \sum_{t=1}^n X_t$. Furthermore, the parameter estimation procedure can be derived by solving \eqref{eq: Yule-Walker} after replacing the population autocovariance matrix $\bm{\Sigma}^{(k)}$ with the sample autocovariance matrix 
\begin{equation}
\label{eq: sample-cov}
\widehat{\bm{\Sigma}}^{(k)} = \frac{1}{n} \sum_{t=1}^{n-k} (X_{t+k} - \Bar{X})(X_{t} - \Bar{X})^\top, \text{ for } k = 0,1.
\end{equation}
\end{remark}

\section{Estimation and inference of gbVAR process}
\label{sec 3}
This section focuses on estimation and inference of the coefficient matrix $\mathcal{A}$  of a gbVAR(1) model that satisfies 
\eqref{eq: gbvarp} with the time lag $p = 1.$ In addition, it establishes a bootstrap algorithm to assist statistical inference through computer simulations.

This paper mainly considers the time lag $p = 1$ situation to decrease the notational complexity and makes the presentation clear. For more general $p \geq 1$ setups, we demonstrate in the online supplement that statisticians may convert a $d$ dimensional gbVAR($p$) model to a new $pd$ dimensional  gbVAR(1) model by rearranging the parameter matrices. This rearrangement allows our estimation and inference methods to be applied in the general cases.

\subsection{Estimating Transition Matrix in VAR Models }

The key difficulty of the estimation lies in the high-dimensionality of data, where the dimension $d$ may exceed the number of observations $n$. In setups where $d$ is substantially smaller than $n$, according to the Yule--Walker equation in Proposition \ref{thm.Yule_Walker}, practitioners may replace the population covariance matrices $\bm{\Sigma}^{(k)}, k = 0,1$ with the sample covariance matrices $\widehat{\bm{\Sigma}}^{(k)},$ and obtain the estimator by inverting $\widehat{\bm{\Sigma}}^{(0)}.$  However, when $d > n,$ the sample covariance matrix $\widehat{\bm{\Sigma}}^{(0)}$ is not invertible, and the solution of the Yule-Walker equations is not well-defined. The literature has offered methods to estimate precision matrices for high-dimensional data, including the works of \cite{cai2011constrained, 10.1111/ectj.12061, 29a36494-561a-3a70-b30a-b00da391c084}, among others. However, such methods often impose sparse structures on the precision matrices of data, which are not compatible with our settings, where the parameter matrix $\mathcal{A}$ is already assumed to be sparse.

This work adopts a two-stage approach for estimation summarized in Algorithm \ref{algo: post-selection} as follows.
\begin{breakablealgorithm}
\caption{Post-selection estimation on $\mathcal{A}$}\label{algo: post-selection}
\begin{algorithmic}[1]
\Require{Observations $X_t \in \mathbb{R}^d$, $t =1, \ldots, n$ from \eqref{eq: gbvarp}. Lasso parameter $\lambda$, the threshold $b_d$.}
\State Calculate the sample covariance matrix $\widehat{\bm{\Sigma}}^{(k)}$ for $k = 0, 1$ by \eqref{eq: sample-cov}.
\State Leverage Lasso to obtain a consistent estimator for $\bm{\alpha}_{i,.}$, each row of the parameter matrix $\mathcal{A},$ as follows: 
\begin{equation*}
\label{eq: lasso}
    \begin{aligned}
        \widehat{\bm{\alpha}}_{i.}(\lambda) = \argmin_{\bm{\omega}\in \mathbb{R}^d}\left\{\frac{1}{2d}|\widehat{\bm{\Sigma}}_{i.}^{(1)\top}- \widehat{\bm{\Sigma}}^{(0)} \bm{\omega}|_2^2 +\lambda |\bm{\omega}|_1\right\},
    \end{aligned}
\end{equation*}
where $\lambda>0$ is a tuning hyper-parameter.
\State Choose the elements in $\widehat{\mathcal{A}}$ with relatively large absolute values greater than $b_d$, formally defined as the following index set:
\begin{equation*}
\label{eq: set S}
    \widehat{S}_i = \{j \in \{1,2,\ldots, d\}: |\widehat{\bm{\alpha}}_{ij}|>b_d\}.
\end{equation*}
\State Fit a least-squares estimation on the chosen elements:
\begin{equation}
\label{eq: least-square}
\tilde{\bm{\alpha}}_{i.} = \argmin_{\bm{\omega}\in\{ \bm{\omega}_j , j\in \widehat{S}_i\}} |\widehat{\bm{\Sigma}}_{i.}^{(1)\top}- \widehat{\bm{\Sigma}}^{(0)} \bm{\omega}|_2^2,\quad\text{where}\quad i = 1,\cdots,d.
\end{equation}
\Statex \textbf{Output:} The post-selection estimator $\tilde{\mathcal{A}} = [\tilde{\bm{\alpha}}_{1.}^\top, \ldots, \tilde{\bm{\alpha}}_{d.}^\top]^\top$. 
\end{algorithmic}
\end{breakablealgorithm}

Compared to more sophisticated linear regression algorithms, least-squares estimation has the advantage of close-form formulas, making  theoretical analysis easier. We demonstrate below that, despite least-squares is applied to the indice sets $\widehat{S}_i$ for different $i,$  after introducing a partial inverse operator $\mathcal{F}.(\cdot),$ the least-squares estimators $\tilde{\bm{\alpha}}_{i.}$ still admit simple close-form formulas \eqref{eq: post estimator}.

\begin{definition}[Partial inverse operator $\mathcal{F}.(\cdot)$]
    Suppose a matrix $\mathbf{A} \in \mathbb{R}^{d \times d}$ and an index set $S = \{1 \leq k_1 < k_2<\cdots < k_v \leq d\}$, do the following operations:
    \begin{enumerate}
        \item Choose the $k_s$, $s = 1, \cdots, v$-th row and column of $\mathbf{A}$ and form a new matrix $\mathbf{A}_S$.
        \item Calculate the Moore-Penrose pseudo-inverse matrix of $\mathbf{A}_S$, and denote the matrix by $\mathbf{Z}$.
        \item Define the matrix $\mathcal{F}_S(\mathbf{A})$ as follows:
        \begin{equation}
            \mathcal{F}_S(\mathbf{A}) = \mathbf{C} \in \mathbb{R}^{d \times d}, \text{ where } \mathbf{C}_{ij} = \begin{cases}
            \mathbf{Z}_{st} & \text{ if } i = k_s, \text{ and } j = k_t,\\
            0 & \text{ otherwise. }
            \end{cases}
        \end{equation}
    \end{enumerate}
\end{definition}
With the help of the partial inverse operator, the post-selection estimator in \eqref{eq: least-square} has the following close-form expression 
\begin{equation}
\label{eq: post estimator}
\tilde{\bm{\alpha}}_{i.} = \widehat{\bm{\Sigma}}_{i.}^{(1)}\widehat{\bm{\Sigma}}^{(0)} \mathcal{F}_{\widehat{S}_{i}}\left(\widehat{\bm{\Sigma}}^{(0)\top} \widehat{\bm{\Sigma}}^{(0)} \right), i = 1,\ldots, d.
\end{equation}
The validity of the post-selection estimator relies on the sparsity of $\mathcal{A}.$ Even though the dimension of the observations can be large, for each $i,$ the size of the index sets $\widehat{S}_{i}$ is not large compared to the sample size due to the sparse assumption, which makes the least-squares estimators well-defined. Notably, by replacing $\mathcal{F}_{\widehat{S}_{i}}(\widehat{\bm{\Sigma}}^{(0)\top} \widehat{\bm{\Sigma}}^{(0)})$ with the inverse matrix $(\widehat{\bm{\Sigma}}^{(0)\top} \widehat{\bm{\Sigma}}^{(0)}  )^{-1}$, the estimator becomes the usual least-squares estimator. However, in the high-dimensional setup, the number of parameters may be larger than the sample size, and $\widehat{\bm{\Sigma}}^{(0)\top} \widehat{\bm{\Sigma}}^{(0)}$ is not invertible. Our work therefore can be recognized as a generalization of least-squares estimator to high-dimensional data setups.

\begin{remark}
\label{rem: partial operator}
If $\mathbf{A}$ is symmetric and there exists a positive constant $c>0$ such that $\lambda_{\min}(\mathbf{A}) > c,$ then $\mathcal{F}_S(\mathbf{A})$ is also symmetric and it satisfies $|\mathcal{F}_S(\mathbf{A})|_1 \leq \frac{|S|}{c}.$ Besides, if a vector $\bm{\alpha} \in \mathbb{R}^{1\times d}$ satisfies $S = \{j : \bm{\alpha}_j \neq 0\},$ then we have 
\begin{equation}
    \bm{\alpha}\mathbf{A}\mathcal{F}_S(\mathbf{A}) = \bm{\alpha}.
    \label{eq.recover_consistency}
\end{equation}
This property motivates the name ``partial inverse operator'' as $\mathcal{F}_S(\mathbf{A})$ behaves like an inverse matrix of $\mathbf{A}$ when applied to the sparse row vector $\bm{\alpha}.$ It will be proved in the Lemma \ref{lem: partial}.
\end{remark}

\subsection{Second-order wild bootstrap}
Bootstrap method has emerged as a powerful and versatile tool for estimating distribution information of sample statistics, especially when identifying a pivotal quantity is difficult. \cite{efron1979bootstrap, wu1986jackknife, stine1985bootstrap}, among others, introduced various bootstrap algorithms for independent data, while later works, including \cite{politis1990circular, politis1994stationary}, claimed the effectiveness of bootstrap algorithms in time series analysis. Examples of recent studies on bootstrap algorithms and their extensions include \citep{buhlmann197Sieve, jentsch2015covariance} for linear processes, \citep{meyer2020extending} for spectral statistics, \citep{kreiss2015bootstrapping} for locally stationary time series. Among them, \cite{shao2010dependent} introduced the dependent wild bootstrap for stationary time series, which can be extended to  different setups, even non-stationary time series, as demonstrated in \citep{zhang2022ridge, zhang2023simultaneous, zhang2023statistical}.

gbVAR processes demonstrate complicated spatial and temporal correlations. Compared to the vector autoregressive process, the parameter matrices in a gbVAR process are random variables instead of numbers,  so the covariance matrix of the post-selection estimator depends not only on the innovations’ covariance matrix but also on the covariance structure induced by the random coefficients. These two sources of covariances inevitably lead to more complex and harder to calculate variance and covariance structure in the estimators---according to \citep{krampe2023structural}, even when the innovations are i.i.d, the autoregressive coefficient estimator of gbVAR processes still has complicated covariance matrices.  Our work leverages the second-order wild bootstrap introduced in \cite{zhang2023statistical} to develop a bootstrap algorithm suitable for analyzing estimators of the gbVAR time series.

The implementation of the second-order wild bootstrap requires a kernel function that satisfies Definition \ref{definition: Kernel function} below.

\begin{definition}[Kernel function]
\label{definition: Kernel function}
    Suppose function $K(\cdot) : \mathbb{R} \rightarrow [0, \infty)$ is symmetric, continuously differentiable, $K(0) = 1, \int_{\mathbb{R}} K(x)dx < \infty$, and K(x) is decreasing on $[0, \infty)$. Besides, define the Fourier transform $\mathcal{F}K(x) = \int_{\mathbb{R}} K(t) \exp(-2\pi i  tx)dt$, here $i = \sqrt{-1}$. Assume $\mathcal{F}K(x) \geq 0$ for all $x \in \mathbb{R}$ and $\int_{\mathbb{R}}\mathcal{F} K(x)dx<\infty$. We call $K$ the kernel function.
\end{definition}

\begin{breakablealgorithm}
\caption{Second-order wild bootstrap}
\label{algo: swb}
\begin{algorithmic}[1]
\Require{Observations $X_t \in \mathbb{R}^d$, $t =1, \ldots, n$ from \eqref{eq: gbvarp}. Lasso parameter $\lambda$, the threshold $b_d$, nominal coverage probability $1-\alpha$, the kernel bandwidth $h_n$, a kernel function $K(\cdot)$, number of bootstrap replicates $B$.}
\State Derive the Lasso estimator $\widehat{\mathcal{A}}$ as in \eqref{eq: lasso}, the set $\widehat{S}_i$, the post-selection estimator $\tilde{\mathcal{A}}$ in \eqref{eq: post estimator}. 

\State Calculate the lag 1 second-order residual
    $$
\widehat{\bm{\Theta}}_t = (X_{t+1} - \bar{X})(X_t -\bar{X})^\top - \tilde{\mathcal{A}}(X_t - \bar{X})(X_t -\bar{X})^\top, \quad t = 1,\ldots, n-1. 
    $$
\State Generate joint normal random variables $e_t$, $t=1, \ldots, n$ such that $\mathbb{E}e_t = 0$ and $\mathbb{E}e_{t_1} e_{t_2} = K\left(\frac{t_1 -t_2}{h_n}\right)$ given the bandwidth $h_n$.

\State Calculate the resample covariance matrix
    $$
\widehat{\bm{\Sigma}}^{(1)*} = \widehat{\bm{\Sigma}}^{(1)} + \frac{1}{n}\sum_{t = 1}^{n-1}\widehat{\bm{\Theta}}_t e_t.
    $$
\State Derive the bootstrapped estimator $\tilde{\mathcal{A}}^*$ and the bootstrapped estimation root as follows:
$$
\tilde{\bm{\alpha}}_{i\cdot}^*= \widehat{\bm{\Sigma}}_{i\cdot}^{(1)*}\widehat{\bm{\Sigma}}^{(0)} \mathcal{F}_{\widehat{S}_i}(\widehat{\bm{\Sigma}}^{(0)\top}\widehat{\bm{\Sigma}}^{(0)}), \quad i = 1,\ldots, d,
$$
and
$$
\bm{\Delta}_{ij}^* = \sqrt{n}(\tilde{\bm{\alpha}}_{ij}^*- \tilde{\bm{\alpha}}_{ij})
$$
here $j \in \widehat{S}_i$, $i = 1, \ldots , d$. Then, calculate
$$
\psi_b^*=\max_{i,j = 1, \ldots, d }|\bm{\Delta}_{ij}^*|.
$$

\State Repeat step for $b = 1, \ldots, B$, and calculate the $1 -\alpha$ sample quantile, i.e., sort $\psi_b^*$ into $\psi_{(1)}^* \leq \psi_{(2)}^* \leq \cdots \leq \psi_{(B)}^*$, then
$$
c^*_{1-\alpha} = \psi^*_{(k)} \text{  s.t.  } k = \min\left\{s = 1, \ldots, B : \frac{s}{B} \geq 1 - \alpha\right\}.
$$

\State (Constructing simultaneous confidence region) The $(1-\alpha) \times 100\%$ simultaneous confidence region for $\mathcal{A}$ is given by the inequality set 
$$
\left\{\mathcal{A}= (\bm{\alpha}_{1\cdot}^\top, \ldots, \bm{\alpha}_{d\cdot}^\top)^\top: \max_{i,j = 1, \ldots, d}\sqrt{n}|\tilde{\bm{\alpha}}_{ij} -\bm{\alpha}_{ij}|\leq c^*_{1-\alpha} \right\}
$$

\State (Performing hypothesis test) Reject the null hypothesis if 
$$
\max_{i,j = 1, \ldots, d}\sqrt{n} |\tilde{\bm{\alpha}}_{ij} -\bm{\alpha}_{ij}|>c_{1-\alpha}^*
$$
\end{algorithmic}
\end{breakablealgorithm}

\begin{remark}

This second-order wild bootstrap introduces a lag-1 ``second-order'' residual component \(\widehat\Theta_t\) in Step 2 and employs a kernel/ bandwidth smoothing in Steps 3 and 4 to capture serial dependence and heteroskedasticity. In practice, the choice of bandwidth \(h_n\) and kernel \(K(\cdot)\) represents the smoothness of the estimated dependence structure with $h_n \asymp\sqrt{n}/(\log^2(d\vee n)\log^9(n))$ discussed in theorem \ref{thm: swb} later.
\end{remark}

\section{Asymptotic analysis of the estimator}
\label{sec 4} 
Analysis of time series data typically relies on short-range dependent conditions, such as the mixing conditions in \cite{MR1312160}. However, the strong mixing conditions are not well-suited for integer-value time series---as illustrated in \citet{wu2005nonlinear}, even a simple Bernoulli shift process can fail to satisfy the mixing condition. \citet{wu2005nonlinear} introduced another condition to describe short range dependence, named the ``physical dependence measure,'' which were later used in the work of \cite{Rho03072015}, \cite{zhang2017gaussian}, \cite{zhang2018asymptotic}, \cite{li2023inference}, among others.  This section demonstrates that the  data generating process  \eqref{eq: gbvarp} obeys the physical dependence condition and establishes the asymptotic properties of the post-selection estimator. Let $\bm{\varepsilon}_t,$ $t\in\mathbb{Z}$  be i.i.d. random variables, define  $\mathcal{F}_{-\infty}^t = (\ldots,\bm{\varepsilon_{t-1}}, \bm{\varepsilon_t}).$ We suppose  $X_t\in\mathbb{R}^d, t\in\mathbb{Z}$ be a stationary process that satisfies equation \eqref{eq: gbvarp}. Furthermore, we assume that each $X_t$ is a function of $\mathcal{F}_{-\infty}^t$ obeying the following form 
\begin{equation}
\label{eq: causal}
    X_t: = \bm{h}(\mathcal{F}_{-\infty}^t) = \bm{h}(\ldots,\bm{\varepsilon}_{t-1}, \bm{\varepsilon}_t),
\end{equation}
where $\bm{h} = (h_1,\ldots, h_d)^\top$ is a $\mathbb{R}^d$-valued measurable function. 

\begin{remark}
  The process \eqref{eq: gbVMA} and equation \eqref{eq: causal} are not mutually exclusive. 
  For example, define  
  $$
    \bm{\varepsilon}_t: = (\ve(A_t^{(+)})^\top, \ve(A_t^{(-)})^\top, \ve(B_t)^\top, e_t^\top)^\top.
$$
  Suppose $\bm{\varepsilon}_t \in \mathbb{R}^{3d^2+d}$, $t\in \mathbb{Z}$ are i.i.d random vectors, according to Theorem \ref{thm: MA rep}, the process \eqref{eq: gbVMA} has the causal form \eqref{eq: causal}, which is already a linear function of $\bm{\varepsilon}_q, q =1,\cdots, t$ that satisfies \eqref{eq: causal}. Our work introduces the representation \eqref{eq: causal} to include more general situations, such as $\bm{\varepsilon}_t$ are general white noises instead of independent random vectors.
\end{remark}

Following  we define
$$    
X_{t,\{0\}}: = \bm{h}(\mathcal{F}_{-\infty}^{t,\{0\}}): = \bm{h}(\ldots,\bm{\varepsilon}_{-1},\bm{\varepsilon}_{0}',\bm{\varepsilon}_{1},\ldots, \bm{\varepsilon}_t)
$$
where $\bm{\varepsilon}_{k}'$ is an i.i.d copy of $\bm{\varepsilon}_{k}$ with copied entry ${A_k^{(+)}}', {A_k^{(-)}}', B'_k,$ and $ e'_k$. Now we can define the functional dependence measure for each component process $(X_{.s}), 1\leq s\leq d$: If $\|X_{ts}\|_q<\infty$ for some $q\geq 1$, define 
$$
    \delta_{t,q,s} = \|X_{ts} - X_{ts,\{0\}}\|_q = \|h_s(\ldots,\bm{\varepsilon}_{t-1}, \bm{\varepsilon}_t) - h_s(\ldots,\bm{\varepsilon}_{-1},\bm{\varepsilon}_{0}',\bm{\varepsilon}_{1},\ldots, \bm{\varepsilon}_t)\|_q,
$$
which measures temporal dependence at lag $t$. Since each component $X_{ts}$ is dependent on the $d$-variate vectors $X_{t-1},X_{t-2}, \ldots$, $\delta_{t,q,s}$ also concerns the cross-sectional dependence. To account for the dependence in the process $X_{.s}$ we define $q$-th dependence adjusted moment (DAM) of the process 
\begin{equation}
\label{eq: GMC_q}
    \begin{aligned}
      \|X_{.s} \|_q: = \sup_{r\geq 0}\rho^{-r}\sum_{t = r}^\infty \|X_{ts} - X_{ts}^{(0)} \|_q <\infty
    \end{aligned}
\end{equation}
and 
the uniform DAM of the process as 
$$\|X_.\|_q = \max_{1\leq s\leq d}\|X_{.s} \|_q.$$

 We provide an example of high-dimensional time series below, for which we can bound the dependence measure and uniform DAM, a key step to applying the theorems to be stated.
\begin{example}[High-dimensional linear process]
    Let $\xi_{ij}, i, j \in \mathbb{Z}$, be i.i.d. random variables with mean 0, variance 1 and having finite $\tau$-th moment with $\tau>2$; let $A_0,A_1,\ldots,$ be $d \times d$ matrices with real entries such that $\sum_{j=0}^\infty \tr(A_j A_j^\top)<\infty$. Write $\varepsilon_i = ( \xi_{i1},\ldots,  \xi_{ip})^\top.$ Then by Kolmogorov’s three series theorem, the $d$-dimensional linear process
$$
\bm{x}_i  =  \sum_{l =0}^\infty A_l \varepsilon_{i-l}
$$
is well-defined. The above process is a special case with a linear functional $\bm{h}(\cdot)$. Let $A_{l,j}$ be the $j$th row of $A_l$ . Then by Burkholder’s inequality, $$\delta_{i,q,j}: = \|A_{i,j}\varepsilon_0\|_\tau \leq (\tau-1)^{1/2}|A_{i,j}|_2\|\xi_{00}\|_\tau.$$ If there exist $0< \rho  < 1$ and $ K > 0 $ such that $\max_{j\leq p}|A_{i,j}|_2\leq K\rho^i$ hold for all $i \geq 0$, then we have 
$$
\|\bm{x}_.\|_\tau = \max_{1\leq j\leq p} \sup_{m\geq 0}\rho^{-m}\sum_{i = m}^\infty \delta_{i,q,j} \leq \frac{K(\tau-1)^{1/2}\|\xi_{00}\|_\tau}{1-\rho},
$$where the constant $K$ only depends on $\tau$.
\end{example}

The first proposition involves \eqref{eq: gbVMA} satisfying geometric moment contraction (GMC(q)). And the framework \eqref{eq: causal} is suitable for two classical tools for dealing with dependent sequences, martingale approximation and $m$-dependence approximation. 
\begin{proposition}[GMC($q$) property of gbVAR(1) process]
\label{lem: gmc(q)}
Suppose the observations $X_t, t\in\mathbb{Z}$ stem from a gbVAR(1) process of the form \eqref{eq: gbvarp}, and there exists  a constant $0<\rho<1$ such that 
$$
    \max_{1 \leq k\leq d} \left(\sum_{l =1}^d \mathcal{A}_{|\cdot |}^{kl}\right)^{1/q}\leq \rho  <1.
$$
Then $X_t$ satisfies the GMC(q) condition.
\end{proposition}


\begin{remark}
The dependence parameter $\rho$ in equation \eqref{eq: GMC_q} is restricted to $(0,1)$ to ensure the weak stationarity of DAM. If $\rho > 1,$ the variance of process $X_t$ may diverge to infinity and measures, such as  \cite{phillips1988testing} and \cite{dejong1992power} proposed tests for detecting the presence of a unit root in general time series models. \cite{zhang2014bootstrapping} considered the special case which imposes the stronger GMC condition $\max_{1\leq s\leq d} \sum_{t=m}^\infty \delta_{t,q,s} \leq C\rho^m$ with some constant and $\rho \in(0,1).$

    In \cite{zhang2017gaussian}, the condition can be extended to exponential-type tail, such as sub-Gaussian (exponential) innovation. We can also define an adapted Orlicz norm if there exists a positive constant $\nu>0$ such that $$\|X_{\cdot}\|_{\psi_\nu}:=\sup_{q\geq 2} q^{-\nu}\|X_\cdot\|_q < \infty.$$
\end{remark}
\subsection{Bernstein-type Inequality under Dependence}
In this section, we introduce a univariate Bernstein-type inequality for the process in \eqref{eq: causal} with $d =1$, the assumption of boundedness and finite second DAM. The well-known Bernstein inequality (\cite{bernstein1946theory}) provides an exponential concentration result for sums of uniformly bounded independent random variables. Let $X_1, \ldots,X_n$ be independent random variables with $\mathbb{E}X_i =0$ and $\sigma_i^2 = \var(X_i)<\infty$, and $|X_i|\leq M$, for all $i$. Denote $S_n = \sum_{i=1}^n X_i$. Then, for any $x >0$, we have
\begin{equation}
    \label{eq: bernstein}
    \mathbb{P}(X_n\geq x) \leq \exp\left(-\frac{x^2}{2\sum_{i=1}^n \sigma_i^2 +2Mx/3}\right).
\end{equation} 
It is well known that the Bernstein-type inequality also holds if $X_i$ has finite exponential moments.

Now, we consider the Bernstein-type inequality for both dependent data and random variables with finite geometric moments. The exponential inequality \eqref{eq: exp univ} is characterized by the DAM $\|X_.\|_2$, the uniform bound $M$, and the dependence parameter $\rho$ which quantifies the constants $C_1$ and $C_2$ in the inequality.

\begin{proposition}[Theorem 2.1 in \cite{zhang2021robust}]
\label{thm: robust}
Let $\left(X_t\right)$ be the process in one-dimensional \eqref{eq: causal} and let $S_n=\sum_{t=1}^n X_t - \mathbb{E}X_t$. Assume $\left|X_t\right| \leq M$ for all $t$, and $\|X.\|_2<\infty$ for some $\rho \in (0,1)$. Also assume $n \geq 4 \vee\left(\log \left(\rho^{-1}\right) / 2\right)$. For any $x>0$, we have the Bernstein-type inequality: 
\begin{equation}
\label{eq: exp univ}
\mathbb{P}\left(S_n \geq x\right) \leq \exp \left\{-\frac{x^2}{4 C_1(n\|X.\|_2^2+M^2)+2 C_2 M(\log n)^2 x}\right\},
\end{equation}
where $C_1=2 \max \{(e^4-5) / 4,[\rho(1-\rho) \log (\rho^{-1})]^{-1}\} \cdot(8 \vee \log (\rho^{-1}))^2, C_2=$ $\max \{(c \log 2)^{-1},[1 \vee(\log (\rho^{-1}) / 8)]\}$ with $c=[\log (\rho^{-1}) / 8] \wedge \sqrt{(\log 2) \log (\rho^{-1}) / 4}$.
\end{proposition}

\begin{remark}
    If $\|X_{\cdot}\|_2 = O(1)$, compared with the classical Bernstein inequality \eqref{eq: bernstein} for independent processes, the result provides an additional $(\log n)^2$ order in the sub-exponential-type tail. Theorem 6 in \cite{adamczak2008tail} provides a slightly sharper inequality involving only an additional $\log n$ order:
    $$\mathbb{P}\left(S_n \geq x\right) \leq C \exp \left\{-\frac{1}{C} \min \left(\frac{x^2}{n \nu^2}, \frac{x}{\log n}\right)\right\},$$
    where $S_n = \sum_{i=1}^n X_i$, $X_i = \sum_{i=1}^n f (Y_i)$, $(Y_i)$ is a Markov chain satisfying some minorization condition, $f$ is a bounded function, and $\nu^2 = \lim_{n\rightarrow\infty} \var(S_n/\sqrt{n})$. The result is as sharp as Theorem 2 in \cite{merlevede2009bernstein} up to a multiplicative constant in the exponential function:
    $$
\mathbb{P}\left(S_n \geq x\right) \leq \exp \left\{-\frac{C x^2}{n \nu^2+M^2+M(\log n)^2 x}\right\},
$$
where $(X_i)$ is a strong mixing process with mean zero, bounded by $M$.
\end{remark}
Denote $\widehat{\mu} = \frac{1}{n}\sum_{i=1}^n X_i$, and we introduce the main assumptions required in studying the properties of the mean estimator $\widehat{\mu} = ( \widehat{\mu}_1, \ldots, \widehat{\mu}_d)^\top$ which are displayed in Appendix.
\begin{itemize}
    \item [A1.] $n \geq 4 \vee\left(\log \left(\rho^{-1}\right) / 2\right)$ and $d \geq 3$.
    \item [A2.] $\sigma_2 := \max_{1\leq j\leq d}\sqrt{\var (X_{ij})}< \infty$.
    \item [A3.] $n, d\to\infty$, and $\log n\sqrt{\frac{\log d}{n}}\to 0$
\end{itemize}

\subsection{Consistency and Gaussian approximation}
This section establishes theoretical justifications, including consistency, Gaussian approximation, and the validity of the bootstrap algorithm, for the proposed estimators. We define the index set representing the support of the parameters $\bm{\alpha}_{ij}$ as follows:
$$
S_i = \{j \in \{1,2, \ldots, d\}: \bm{\alpha}_{ij} \neq 0\}.
$$
Before presenting the theoretical results, we introduce the frequently used assumptions in this section.

\noindent\textbf{Assumptions: }
\begin{itemize}
    \item [B1.] $\max_{i=1,\ldots ,d} |S_i| = O(1)$ and $\max_{i=1,\ldots ,d} |S_i| > 0$.
    \item [B2.] There exits a positive constant $c>0$ such that the minimum eigenvalue of $\bm{\Sigma}^{(0)} = \text{Cov}(X_t, X_t)$ is greater than $c$.
    \item [B3.]There exits a positive constant $c>0$ such that the minimum eigenvalue of $\sum_{\ell=0}^\infty\text{Cov}(X_t, X_{t+\ell})$ is greater than $c$.
\end{itemize}

Theorem \ref{thm: alpha} establishes the model selection consistency of the Lasso estimator. Based on this result, the partial inverse operator can accurately recover the sparse structure of the parameters $\bm{\alpha}_{ij}$ according to \eqref{eq.recover_consistency}, thereby justifying the use of the post-selection least-squares estimators.

\begin{theorem}
\label{thm: alpha}
    Suppose $(X_t)$ satisfies \eqref{eq: causal} form and conditions in Proposition \ref{lem: gmc(q)}. Under conditions A1 to A3, B1 and B2, $\lambda \asymp (\log n)\sqrt{\frac{\log d}{n}}$, we have
    \begin{equation}
        \max_{i = 1,\ldots,d} |\hat{\bm{\alpha}}_{i.} - \bm{\alpha}_{i.}|_1 = O_\mathbb{P}(\lambda) \text{ and } \max_{i = 1,\ldots,d} |\hat{\bm{\alpha}}_{i.} - \bm{\alpha}_{i.}|_2 = O_\mathbb{P}(\lambda) 
    \end{equation}
where $\lambda$ is the tuning parameter defined in \eqref{eq: lasso}. In particular, we have 
     \begin{equation}
     \label{eq: index set}
      \mathbb{P}\left(\bigcup_{i = 1}^d \{\widehat{S}_{i} \neq S_i\}\right) =o(1)  
     \end{equation}
\end{theorem}

Lemma \ref{thm: GA} discusses  the Gaussian approximation theorem for the sample mean of dependent random variables. Gaussian approximation, as presented in \citet{chernozhukov2013gaussian, zhang2017gaussian, chernozhuokov2022improved, chang2024central} among others is a useful tool for analyzing high-dimensional dependent data. 

\begin{lemma}
    \label{thm: GA}
Let $(X_t)$ be the process in the form \eqref{eq: causal} with $\mu =\mathbb{E}X_t$ and satisfies \eqref{eq: GMC_q}. Under conditions B1 to B3 and $\|X_{\cdot}\|_2 <\infty $, then
   \begin{equation}
       \label{eq: GA}
       \sup_{x \in \mathbb{R}}|\mathbb{P}(\sqrt{n}|\bar{X} - \mu |_\infty \leq x)-\mathbb{P}(|Z|_\infty \leq x)| = o(1),
   \end{equation}
   where $Z \in \mathbb{R}^d$ has a joint normal distribution with $\mathbb{E}Z=0$ and $\cov(Z_i, Z_j) = \cov(\bar{X}_i, \bar{X}_j)$.
\end{lemma}
Similar to Lemma \ref{thm: GA}, the result in Theorem \ref{thm: GA post selection} presents the convergence rate and the asymptotic distribution of the post-selection estimator, based on the Gaussian approximation for $\tilde{\mathcal{A}}$ under the regular conditions.

\begin{theorem}
\label{thm: GA post selection}
   Let conditions B1 to B3 hold, the post-selection estimator satisfies
    \begin{equation}
    \label{eq: GA post selection}
        \begin{aligned}
            \sup_{x \in \mathbb{R}}\left|\mathbb{P}\left(\sqrt{n}|\tilde{\mathcal{A}} - \mathcal{A}|_{\max} \leq x\right) - \mathbb{P}\left(\max_{i = 1, \ldots, d, j \in S_i}|Z_{ij}|\leq x\right)\right|=o(1)
        \end{aligned}
    \end{equation}
    where $Z_{ij}$ is joint normal random variables with $\mathbb{E} Z_{ij} = 0$ and 
    \begin{equation}
    \label{eq: cov}
        \cov(Z_{i_1, j_1}, Z_{i_2, j_2}) = \bm{\Sigma}^{(i_1, i_2)}_{j_1 j_2}
    \end{equation}
    which satisfies there exists a constant $c > 0$ such that for sufficiently large $n$ , $\bm{\Sigma}_{j,j}^{(i,i)}>c$ for all $i = 1,\ldots,d$ and $j \in S_i$.
\end{theorem}

\begin{remark}
   Theorem \ref{thm: GA post selection} explains why the Gaussian approximation for the post-selection estimator needs the condition $\max_{i = 1, \ldots, d}|S_i|>0$ in Assumption B1. Without this condition, it implies that $S_i$ can be empty and $Z_{ij}$ will be a constant rather than a random variable (or Gaussian). For the joint normal random variables $Z_{ij}$, it satisfies $$\max_{i=1,\ldots ,d,j \in S_i} |Z_{ij}| = O_\mathbb{P} (\sqrt{\log d}).$$ Therefore, Theorem \ref{thm: GA post selection} also implies that
\begin{equation}
\label{eq: convergence rate for post selection}
    \max_{i=1,\ldots ,d,j \in S_i} |\tilde{\mathcal{A}} - \mathcal{A}|_{\max} =O_{\mathbb{P}}\left(\sqrt{\frac{\log d}{n}}\right).
\end{equation}
   
\end{remark}

Finally, we establish the consistency of the proposed bootstrap algorithm (Algorithm \ref{algo: swb}). According to Theorem 1.2.1 of \cite{MR1707286}, the validity of Algorithm \ref{algo: swb} is ensured if the cumulative distribution function of the bootstrap estimator in the bootstrap world converges to that of the original estimator. That is, the following condition holds true:
\begin{equation}
    \sup_{x \in \mathbb{R}}\left|\mathbb{P}^*\left(\max_{i,j = 1, \ldots, d} |\mathbf{\Delta}_{ij}^*| \leq x \right) - \mathbb{P}\left(\max_{i = 1, \ldots, d, j\in S_i}|Z_{ij}|\leq x\right)\right| =o(1),
    \label{eq.bootstrap_consistency}
\end{equation}
where $\bm{\Delta}_{ij}^*$ is defined in Algorithm \ref{algo: swb}, $Z_{ij}$ is defined in Theorem \ref{thm: GA post selection}, and $\mathbb{P}^*(\cdot) = \mathbb{P}\left(\cdot \mid\{X_t\}_{t=1}^n\right)$ is "the
 probability in the bootstrap world," i.e., the conditional probability conditioning on the observations $X_t$.

We define the conditional expectation $\mathbb{E}^{*} \left[{\cdot}\right] = \mathbb{E} \left[{\cdot}  \mid X_t , t = 1, \ldots, n\right]$, i.e., conditional on all observed data. In the literature, such expectation is always referred to as ``the expectation in the bootstrap world.'' Suppose $\widehat{S}_i = S_i$, which is guaranteed for large sample size from Theorem \ref{thm: alpha}, with a given the bandwidth $h_n$, the conditional covariances of the bootstrapped estimators $\bm{\Delta}_{ij}^*$  in Algorithm \ref{algo: swb} is
\begin{equation}
    \begin{aligned}
        \mathbb{E}^*\left[\bm{\Delta}_{i_1,j_1}^* \bm{\Delta}_{i_2,j_2}^*\right] = \frac{1}{n}\sum_{t_1,t_2=1}^{n-1}\sum_{k_1,l_1=1}^d \sum_{k_2,l_2=1}^d &\widehat{\bm{\Theta}}_{t_1, i_1, k_1} \widehat{\bm{\Theta}}_{t_2, i_2, k_2}\widehat{\bm{\Sigma}}_{k_1,l_1}^{(0)} \widehat{\bm{\Sigma}}_{k_2,l_2}^{(0)} \\
       & \mathcal{F}_{S_i}(\widehat{\bm{\Sigma}}^{(0)}\widehat{\bm{\Sigma}}^{(0)})_{l_1,j_1}\mathcal{F}_{S_i}(\widehat{\bm{\Sigma}}^{(0)}\widehat{\bm{\Sigma}}^{(0)})_{l_2,j_2} K\left(\frac{t_1 -t_2}{h_n}\right)
    \end{aligned}
    \label{eq.cov_bootstrap}
\end{equation}
if $j_1 \in S_{i_1}$ and $j_2 \in S_{i_2}$ and 0 otherwise. Therefore, Algorithm \ref{algo: swb} actually generates joint normal random variables with specific covariance structures \eqref{eq.cov_bootstrap}. Since Theorem \ref{thm: GA post selection} shows that the estimation root converges in distribution to the maximum of joint normal random variables with specific covariances, the fact that the covariances in \eqref{eq.cov_bootstrap} approximates that in Theorem \ref{thm: GA post selection} is sufficient to verify the consistency of Algorithm \ref{algo: swb}.

Theorem \ref{thm: swb} provides the consistency of the covariance estimator to validate the bootstrap algorithm. With the same assumptions as in Theorem \ref{thm: swb}, the consistency of the bootstrap algorithm is also validated in Corollary \ref{cor: bootstrap}. 

\begin{theorem}
\label{thm: swb}
    Suppose Assumptions B1 to B3 hold. In addition, suppose that the function $K(\cdot)$ in Algorithm \ref{algo: swb} is a kernel function satisfying Definition \ref{definition: Kernel function} and $h_n$ satisfies $h_n \rightarrow \infty$, $h_n \log^2(d \vee n) \log^{9}(n)/\sqrt{n}\rightarrow 0$.
Then,
\begin{equation}
    \label{eq: swb}
    \max_{i_1, i_2, j_1, j_2 = 1,\ldots, d} \left|\mathbb{E}^*\left[\bm{\Delta}_{i_1,j_1}^*\bm{\Delta}_{i_2,j_2}^*\right]- \bm{\Sigma}_{j_1,j_2}^{(i_1,i_2)}\right\vert =O_\mathbb{P}\left(h_n^{-1}+\frac{h_n \log^2(d\vee n)}{\sqrt{n}}\right),
\end{equation}
where $\bm{\Sigma}^{(i_1, i_2)}$ is defined in \eqref{eq: kernel + cov}, $\bm{\Delta}_{ij}^*$ is defined in Algorithm \ref{algo: swb} and $h_n$ represents the bandwidth in Lemma \ref{lem: kernel + cov}.
\end{theorem}

  \begin{corollary}
    \label{cor: bootstrap}
 Suppose Assumptions B1 to B3 and assumptions in Theorem \ref{thm: swb} hold. Then equation \eqref{eq.bootstrap_consistency} holds true.
    \end{corollary}

    \begin{remark}
Assumption B3 provides the existence of finite long-run covariances. Under this Assumption, suppose $\ell > 0$, then
 \begin{equation}
 \label{eq: cov bound}
        |\cov ({X}_{\ell,i}, {X}_{0,j} )| = |\sum_{k=0}^\infty \mathbb{E}[\mathcal{P}_{-k}(X_{\ell, i})\mathcal{P}_{-k}(X_{0,j}) ] |\leq \sum_{k=0}^\infty \delta_{k,2,i} \delta_{k+\ell,2,j}
 \end{equation}
 so the covariances of random variables in Definition \ref{definition: Kernel function} also satisfy GMC condition in Proposition \ref{lem: gmc(q)}. 
\end{remark}

\section{Numerical Studies}
\label{sec 5}
This section provides numerical experiments and real data studies to evaluate the finite-sample performance of the proposed algorithms in various scenarios. 
\subsection{Simulation Study}
 \textbf{Data Generating Process. }We take the numerical setup of the sample size $n = 1500,$ and the dimension $d = 80$. For the true transition
matrix $\mathcal{A} = (\bm{\alpha}_{ij} )$, we consider the following designs.
\begin{itemize}
    \item (DGP1) Define the Tridiagonal matrix without diagonal elements
$$
\bm{\alpha}_{ij} = \begin{cases}
0.3  & \text{ if } j = i+1 \text{ or }i -1, \\
 0 & \text{ otherwise },
\end{cases}
\text{ and }
\mathcal{B} = diag(0.7, 0.4, \ldots, 0.4, 0.7).
$$

\item (DGP2) Define the X-shaped matrix
\begin{align*}
\bm{\alpha}_{ij} = \bm{\alpha}_{i,d-j} = 0.3 \text{ and }
\mathcal{B}=diag(0.4,\ldots,0.4,0.7, 0.4, \ldots, 0.4).
\end{align*}
\item (DGP3) Define the anti-tridiagonal matrix
\begin{align*}
\bm{\alpha}_{ij} = \begin{cases}
0.3  & \text{ if } j = d - i+1 \text{ or }d - i -1, \\
 0 & \text{ otherwise },
\end{cases}
\text{ and }
\mathcal{B} = diag(0.7, 0.4, \ldots, 0.4, 0.7).
\end{align*}
\end{itemize}
We define the innovation $(e_t)_{t \in \mathbb{Z}}$ in DGP1, DGP2 and DGP3 as $d$-dimensional Bernoulli random processes with $\mu_{e, i} = 1/2$ for $i = 1, \ldots, d$.  Since the model needs to estimate $d^2$ parameters, the chosen value $80$ of $d$ already reflect high-dimensional setups.

 {\noindent \textbf{Hyper-parameter selection. }} Several hyper-parameters are required to be selected in our approach, including the Lasso-type tuning parameter $\lambda$, the threshold $b_d,$ and the bandwidth $h_n$ of Algorithm \ref{algo: swb}. While cross-validation is a common approach for hyper-parameter selection,  randomly partitioning the data into training and testing sets destroys the data dependence structures in our setup, which, according to \cite{burman1994cross}, resists the cross-validation algorithms from selecting suitable hyper-parameters.  Therefore, we propose to use a train-test split instead of commonly used cross-validations to tune parameters $\lambda$ and $b_d,$ which are demonstrated as follows: We fix an split integer $n_\mathtt{train} < n$, and use $X_1, \ldots, X_{n_\mathtt{train}}$ to fit the parameter $\widehat{\mathcal{A}}$. Based on the $n_\mathtt{test}$ data $X_{n_\mathtt{train}+1} , \ldots, X_n$, we calculate the testing $L_1$ or $L_2$ norm error
$$
\tau_1 = \left\vert\widehat{\bm{\Sigma}}^{(1)}_{\mathtt{test}} - \widehat{\mathcal{A}} \widehat{\bm{\Sigma}}^{(0)}_{\mathtt{test}}  \right\vert_1, \qquad \text{ and } \qquad \tau_2 = \left\vert\widehat{\bm{\Sigma}}^{(1)}_{\mathtt{test}} - \widehat{\mathcal{A}} \widehat{\bm{\Sigma}}^{(0)}_{\mathtt{test}}  \right\vert_2,
$$
where these two sample autocovariance matrices $\widehat{\bm{\Sigma}}^{(0)}_{\mathtt{test}}$ and $\widehat{\bm{\Sigma}}^{(1)}_{\mathtt{test}}$ are based on the test data. We choose $\lambda$ and $b_d$ which minimize $\tau$. In our work, we choose $n_{\mathtt{train}} = \lfloor \frac{3}{4}n\rfloor$, where $\lfloor x \rfloor$ denotes the largest integer smaller than or equal to $x$.

{\noindent \textbf{Results.}} The numerical results are presented in Table \ref{table: experiment_param} and \ref{table: performance}. We evaluate the performance of the estimator using the row-wise one (R1) norm $|\cdot|_{R1}$, the row-wise two (R2) norm $|\cdot|_{R2}$, and the model misspecification $\kappa$, defined as follows:
\begin{align*}
    \left\vert\widehat{\mathcal{A}} - \mathcal{A}\right\vert_{R1} = \max_{i = 1, \ldots,d}\sum_{j =1}^d\left\vert\widehat{\bm{\alpha}}_{ij} - \bm{\alpha}_{ij}\right\vert, \quad \left\vert\widehat{\mathcal{A}} - \mathcal{A}\right\vert_{R2} = \max_{i = 1, \ldots,d} \sqrt{\sum_{j = 1}^ d\left(\widehat{\bm{\alpha}}_{ij} - \bm{\alpha}_{ij}\right)^2},
\end{align*}
and $\kappa(\widehat{\mathcal{A}} , \mathcal{A})$ is the number of elements in the set $\mathcal{K}$ based on two estimators $\widehat{\mathcal{A}}$ and $\mathcal{A}$, where
\begin{align*}
    \mathcal{K}(\widehat{\mathcal{A}} , \mathcal{A}) =\{(i,j): \widehat{\bm{\alpha}}_{ij} = 0 \cap \bm{\alpha}_{ij} \neq 0\} \cup \{(i,j): \widehat{\bm{\alpha}}_{ij} \neq 0 \cap \bm{\alpha}_{ij} =0\},
\end{align*}
here $\widehat{\mathcal{A}}$ is an estimator of $\mathcal{A}$. The reason for using the row-wise norm is that our approach in estimation procedures for the vector autoregressive model is established by the Yule-Walker equation and row-wise linear regression. The row-wise norms reflect the efficiency of these models and display the different estimation performances when fitting the sparse model. The model misspecification evaluates the difference between the estimated and true position of non-zero entries in matrix parameter $\mathcal{A}$. If the difference is small enough, the estimator can cover the pattern of sparse matrix $\mathcal{A}$.

\begin{table}[!h]
\centering
\caption{Experiment parameters selected based on methods presented in Section 5. The sample size is chosen to be 1500.}
\begin{tabular}{llllll}
\hline
\multicolumn{1}{c}{Data} & \multicolumn{1}{c}{Method} & \multicolumn{1}{c}{$d$} & \multicolumn{1}{c}{$\lambda$} & \multicolumn{1}{c}{$b_d$} & \multicolumn{1}{c}{$h_n$} \\ 
\hline

DGP1 & Post-selection  & 80 & $6.14\times 10^{-6}$ & 0.131 &  2.333\\
     & Lasso           &    &  $3.17\times 10^{-5}$&  &  \\
     & Han \& Liu      &    & 0.121 &  &  \\
     & Threshold Lasso &    & $1.06\times 10^{-5}$ & 0.145 &  \\
     & Banded Lasso    &    & 0.135 &  &  \\ 
     \hline
DGP2 & Post-selection  & 80 & $6.14\times 10^{-6}$ & 0.105 & 3.368 \\
     & Lasso           &    & $4.09\times 10^{-5}$ &  &  \\
     & Han \& Liu      &    &  0.126 &  &  \\
     & Threshold Lasso &    & $4.09\times 10^{-5}$ & 0.131&  \\
     & Banded Lasso    &    & 0.196 &  &  \\ 
     \hline
    DGP3 & Post-selection  & 80 & $6.14\times 10^{-6}$ & 0.095 & 2.421 \\
     & Lasso           &    & $3.31\times 10^{-5}$ &  &  \\
     & Han \& Liu      &    &  0.103 &  &  \\
     & Threshold Lasso &    & $7.58\times 10^{-6}$ & 0.174&  \\
     & Banded Lasso    &    & 0.288 &  &  \\ 
     \hline
\end{tabular}
\label{table: experiment_param}
\end{table}

\begin{table}[htbp]
\centering
\caption{The performance of different estimation procedures and bootstrap algorithms. $|\hat{\mathcal{A}} - \mathcal{A}|_{R1}$, $|\hat{\mathcal{A}} - \mathcal{A}|_{R2}$, $\kappa$ respectively represent the row-wise one norm, the row-wise two norm, and the model misspecification. These three metrics are derived from the average of 300 simulations.}
\begin{tabular}{lllllcll}
\hline
\multicolumn{1}{c}{\multirow{2}{*}{Data}} &
  \multicolumn{1}{c}{\multirow{2}{*}{$d$}} &
  \multicolumn{1}{c}{\multirow{2}{*}{Method}} &
  \multicolumn{1}{c}{\multirow{2}{*}{$|\hat{\mathcal{A}} - \mathcal{A}|_{R1}$}} &
  \multicolumn{1}{c}{\multirow{2}{*}{$|\hat{\mathcal{A}} - \mathcal{A}|_{R2}$}} &
  \multirow{2}{*}{$\kappa$} &
  \multicolumn{2}{c}{Algorithm 1} \\
\multicolumn{1}{c}{} &
  \multicolumn{1}{c}{} &
  \multicolumn{1}{c}{} &
  \multicolumn{1}{c}{} &
  \multicolumn{1}{c}{} &
   &
  \multicolumn{1}{c}{Coverage} &
  \multicolumn{1}{c}{Length} \\  
  \hline
DGP1 & 80 & Post-selection  & 0.107 & 0.080 & 0.0 &  92\% &  0.227\\
     &    & Lasso           & 0.280 & 0.113 & 163.71 &  &  \\
     &    & Han \& Liu      & 0.490 & 0.147 & 6055.76 &  &  \\
     &    & Threshold Lasso & 0.109 & 0.084 & 0.0 &  &  \\
     &    & Banded Lasso    & 0.149 & 0.116 & 6.8 &  &  \\
     \hline
DGP2 & 80 & Post-selection  & 0.128 & 0.094 & 0.0 &  93\% &  0.181\\
     &    & Lasso           & 0.211 & 0.110 & 162.87 &  &  \\
     &    & Han \& Liu      & 0.517 & 0.149 & 6134.07 &  &  \\
     &    & Threshold Lasso & 0.135 & 0.149 & 0.0 &  &  \\
     &    & Banded Lasso    & 0.402 & 0.317 & 83.53 &  &  \\
     \hline
  DGP3 & 80 & Post-selection  & 0.111 & 0.083 & 0.0 &  93\% &  0.176\\
     &    & Lasso           & 0.267 & 0.115 & 162.92 &  &  \\
     &    & Han \& Liu      & 0.714 & 0.176 & 5799.79 &  &  \\
     &    & Threshold Lasso & 0.108 & 0.084 & 0.08 &  &  \\
     &    & Banded Lasso    & 0.612 & 0.425 & 157.66 &  &  \\
     \hline   
\end{tabular}
\label{table: performance}
\end{table}

\begin{figure}[htbp]
  \centering
     \subfigure[Real parameter]{
         \centering
         \includegraphics[width = 0.8\textwidth]{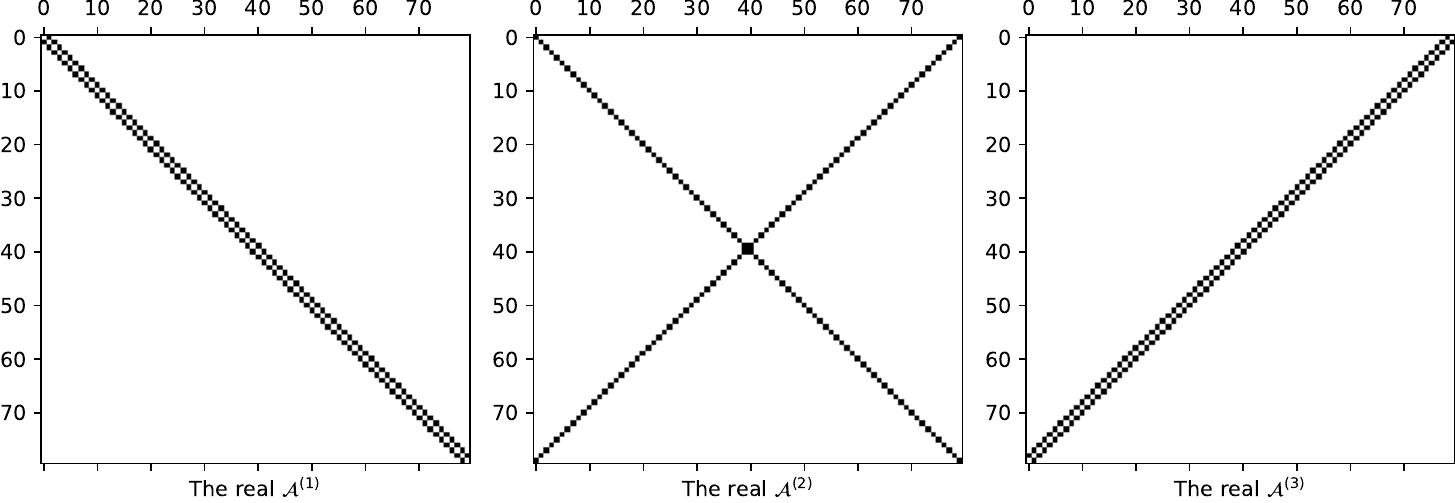}
         \label{figure.heat_1}
     }
     \subfigure[Row-wise Lasso estimator]{
         \centering
         \includegraphics[width = 0.8\textwidth]{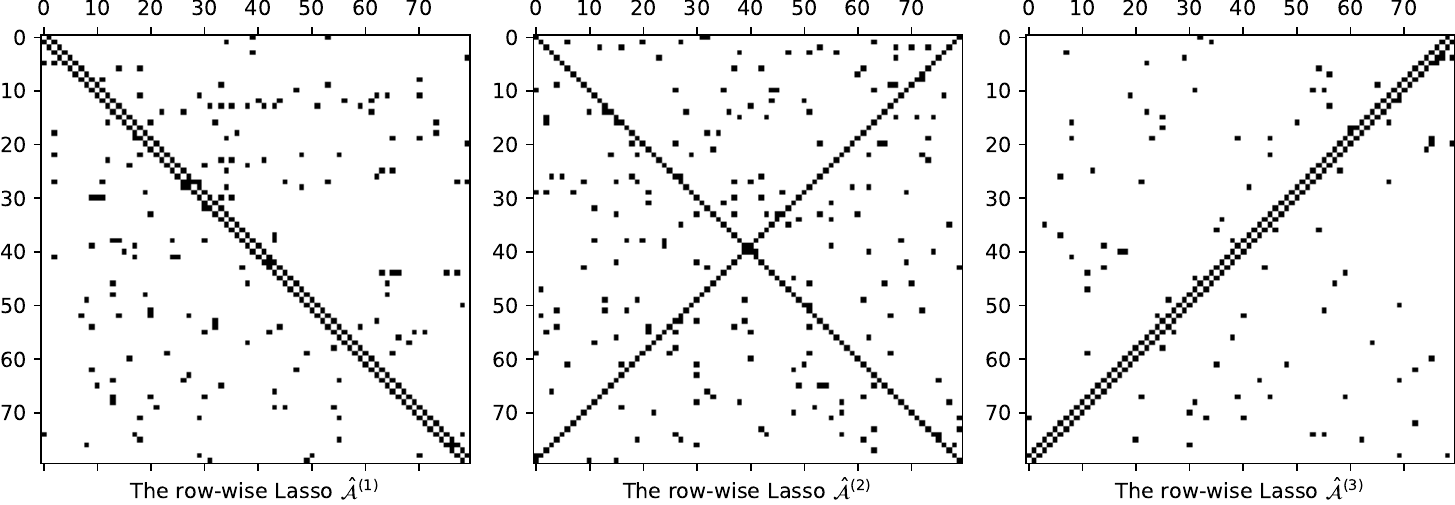}
         \label{figure.heat_2}
     }
     \subfigure[Post-selection estimator]{
         \centering
         \includegraphics[width = 0.8\textwidth]{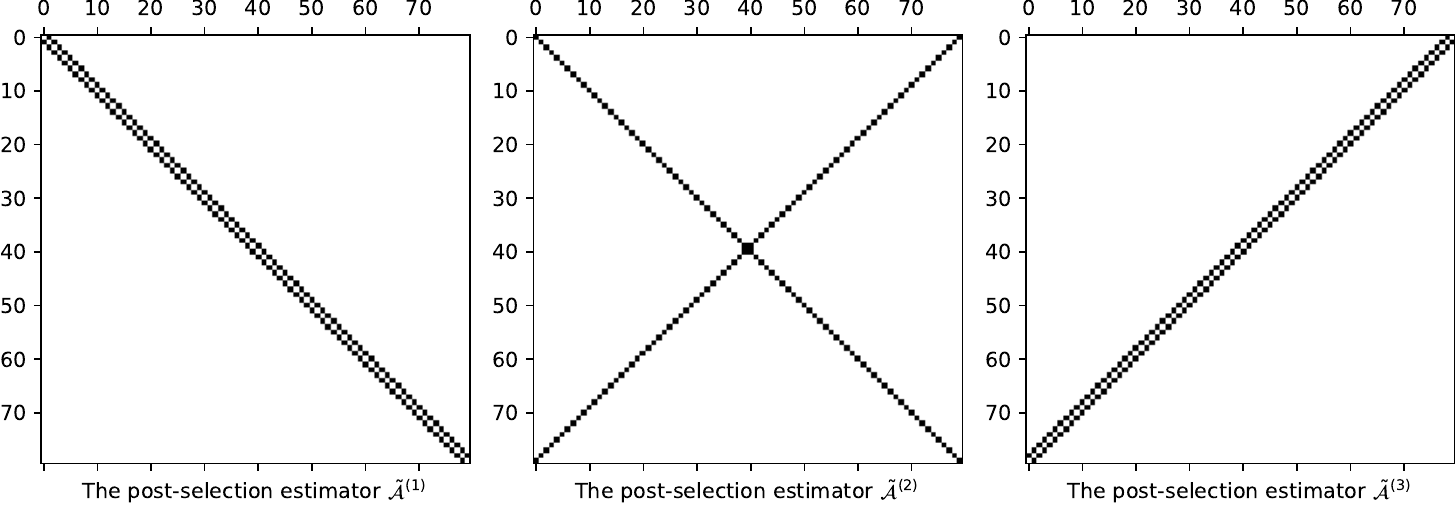}
         \label{figure.heat_3}
     }
     \caption{Figure \ref{figure.heat_1}, \ref{figure.heat_2}, and \ref{figure.heat_3} respectively depict
     the positions of the actual non-zero parameters and the corresponding estimated positions based on the row-wise algorithm and the post-selection estimator. In Figure \ref{figure.heat_1}, $\mathcal{A}^{(1)}$, $\mathcal{A}^{(2)}$ and $\mathcal{A}^{(3)}$ represent the real parameter matrix in DGP1, DGP2, and DGP3, respectively. These observations are generated using $n = 1500$, $d = 80$.}
     \label{figure.heat}
\end{figure}

Figure \ref{figure.heat} plots the actual locations of non-zero parameters in $\mathcal{A}$ along with those estimated by the proposed post-selection and the Lasso method. For the Lasso method, we consider the estimated parameter to be $0$ if its absolute value is smaller than $b_d$ in Table \ref{table: experiment_param} depending on the DGPs. Figure \ref{figure.heat} illustrates that adding a threshold in Lasso precisely recovers the underlying sparse pattern of the parameter matrix $\mathcal{A}.$ On the other hand,  the conventional row-wise Lasso method tends to identify many zero parameters as non-zero due to the fluctuations in the optimization algorithm. Similar conclusions can be drawn from Table \ref{table: performance}, where the model misspecifications $\kappa$ between Lasso estimators and the real matrix is larger than those based on the post-selection estimator. 

\begin{figure}[htbp]

  \centering
     \subfigure[DGP1 model under $n = 1500$, $d = 80$]{
         \centering
         \includegraphics[width = 0.8\textwidth]{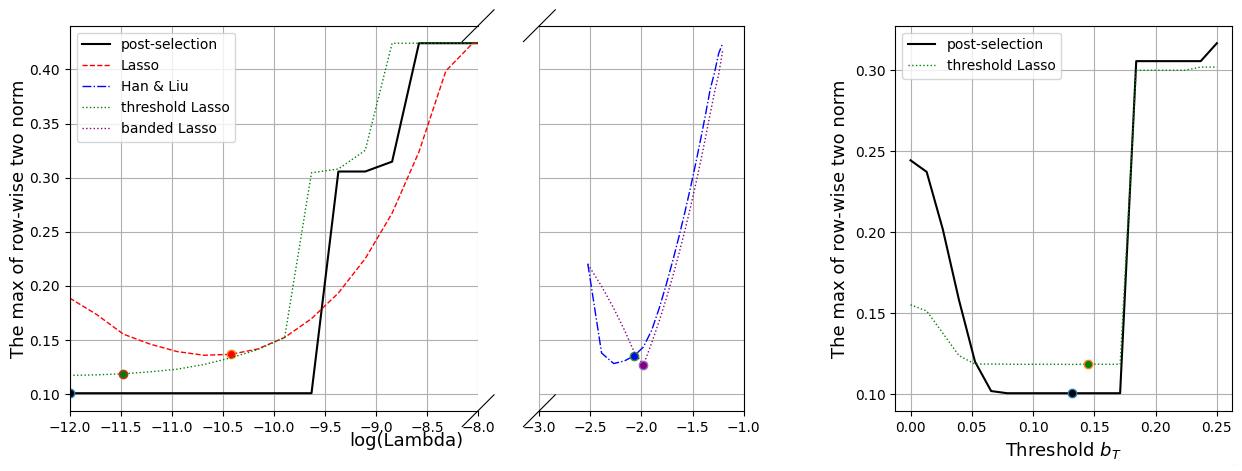}        
     }
     \subfigure[DGP2 model under $n = 1500$, $d = 80$]{
         \centering
         \includegraphics[width = 0.8\textwidth]{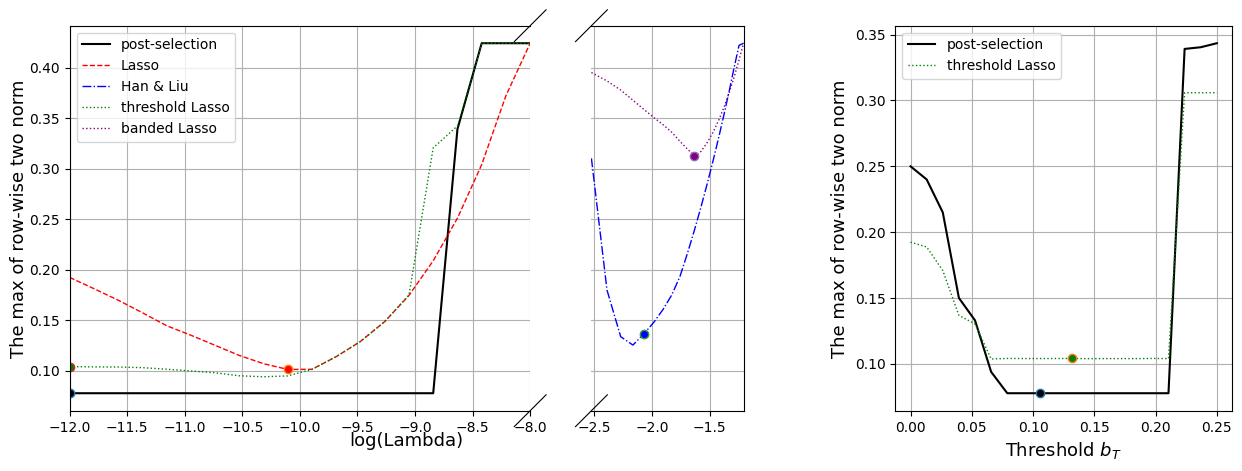} 
         }
      \subfigure[DGP3 model under $n = 1500$, $d = 80$]{
         \centering
         \includegraphics[width = 0.8\textwidth]{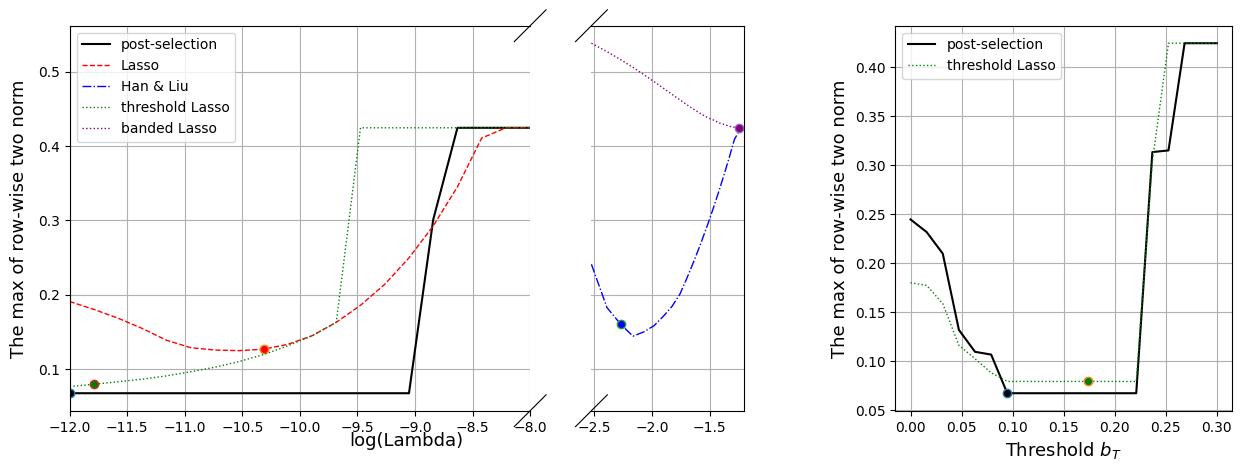} 
          
     }
     \caption{The performance of Lasso and post-selection algorithm concerning different choices of hyperparameters. Figure \ref{fig: performance}(a), (b), and (c) represent the row-wise two norms of estimation error with different $\ln$-scaled tuning parameters $\lambda$ and threshold $b_d$ in DGP1, DGP2, and DGP3 under $n = 1500$, $d = 80$. In the case of the post-selection algorithm, the left figure plots their performance with $b_d$ as the optimal parameter, while the right figure plots the performance with $\lambda$ as the optimal parameter. }
          \label{fig: performance}

\end{figure}

Figure \ref{fig: performance} evaluates the robustness of the proposed post-selection estimator with respect to different choices of
$\lambda$ and the threshold $b_d.$ The results indicate that the estimator is insensitive to the specific selection of these tuning parameters, achieving near-optimal performance over a relatively wide range of $\lambda$ and $b_d$ values.

\subsection{Real Data Study}

\subsubsection{Portfolio Management}
The aftermath of the 2008 financial crisis highlighted the importance of robust economic analysis and effective risk management. In recent years,  growing attention has been devoted to analyzing the stock prices of leading technology companies. Due to their dominant market positions, fluctuations in these stocks can exert a substantial influence on overall market dynamics and investor behaviors.  In our study, we collect data on the seven most valuable NASDAQ-listed stocks from \url{https://www.kaggle.com/datasets/kalilurrahman/nasdaq100-stock-price-data?rvi=1}. However, since the returns of these stocks remain nonstationary, we focus on examining the mutual relationships between advances (days when the closing price exceeds that of the previous day) and declines in their prices. Recording these states yields a vector-valued binary time series.

\begin{figure}[htbp]
    \centering
    \includegraphics[width=\linewidth]{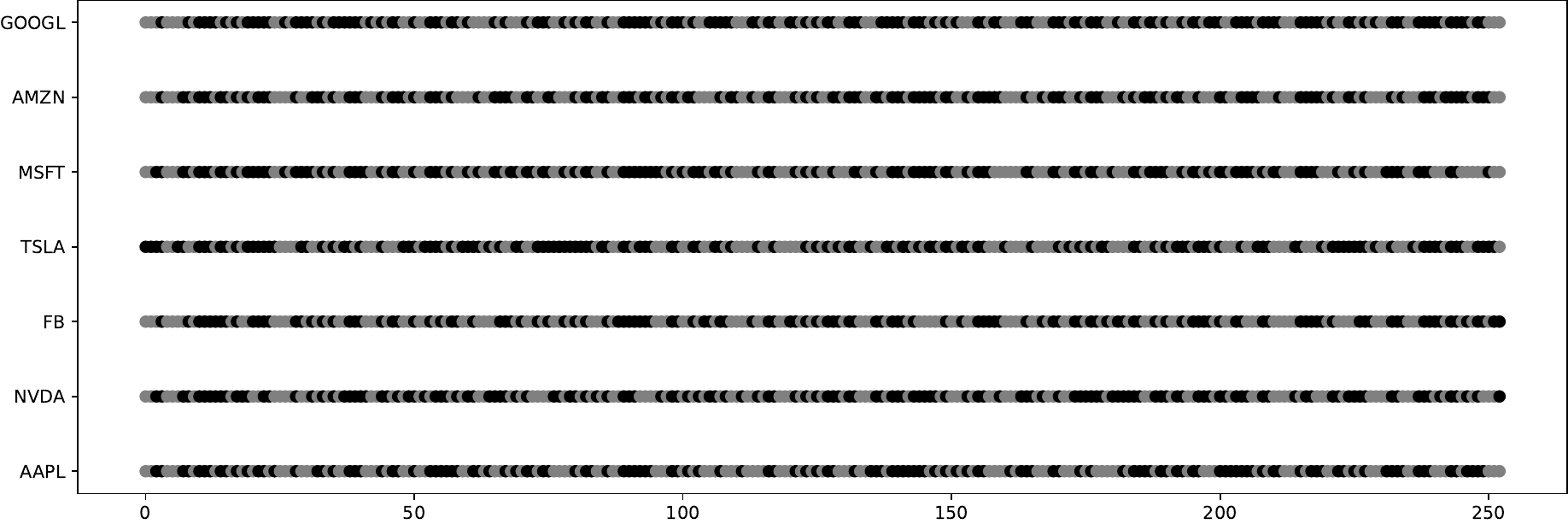}
    \caption{Weekday time series $X_{t,k}$, $k = 1, \ldots, 7$ indicate Advance (black dots) and Decline (no dots) of the closing price between the day $t$ and $t-1$ at seven stocks in NASDAQ: Apple (AAPL), NVIDIA (NVDA), Facebook (FB, currently META), Tesla (TSLA), Microsoft (MSFT), Amazon.com (AMZN), Alphabet (GOOGL) from September 10, 2020 to September 10, 2021. }
\label{fig: weekday time series}
\end{figure}

We observe an important feature: advances and declines in these stocks tend to occur in clusters, indicating both temporal and cross-sectional dependence. Moreover, all sequences exhibit frequent alternations between advances and declines rather than prolonged periods of decline, reflecting the stationarity of the observations. We also calculate the sample mean vector of observations, 
$$
\widehat{\mu}_X = (0.5134, 0.5494, 0.5020, 0.5375, 0.5296, 0.5257, 0.5771)',
$$
demonstrating that all stocks share similar advance rates. It is not surprising, as these seven stocks highlight almost 50\% of the NASDAQ. Even when suffering from the coronavirus pandemic and high unemployment rates, all seven stocks have surged in 2020-2021, with gains ranging from 16\% (Alphabet) to 279\% (Tesla). 

After examining the mean performance, we fit a gbVAR(1) model to study both the temporal and cross-sectional dependence between data. Following the work of \cite{jentsch2022generalized}, we fit a preliminary Yule–Walker estimation to assess the regularity of data. It yields the largest row-sum in the parameter matrix $\max_{k}\sum_{l=1}^K|\widehat{\bm{\alpha}}_{kl}| = 0.6712<1$. In addition, the largest absolute eigenvalue of $\widehat{\mathcal{A}}^\mathtt{OLS}$ is 0.1666, indicating that the underlying data-generating process of the stock prices is stationary and satisfies the assumptions required by our analysis.

We fit the data using the proposed post-selection estimator. In addition, we fit an ordinary least squares (OLS) estimator $\widehat{\mathcal{A}}^\mathtt{OLS}$ as a baseline for comparison, with the results demonstrated in Figure \ref{fig: global_econ}. In addition, Table \ref{tab: wild bootstrap} plots the hyper parameters and the widths of the 95\% confidence intervals generated by Algorithm \ref{algo: swb}. The selection of the hyper parameters is based on the aforementioned train-test split.

Compared to the baseline OLS estimator, our method shrinks small coefficients to zero. Furthermore, for a large proportion of the selected coefficients, the corresponding confidence intervals exclude zero, indicating their statistical significance. As shown in Figure \ref{fig: heatmap}(b), the state of Apple depends negatively on previous states of NVIDIA, and the state of Tesla also depends on both the states of Microsoft and Amazon the day before. This phenomenon can be explained by the ripple effect in the supply chain introduced in \cite{kravchenko2024responding}: the COVID pandemic in 2020-2021 triggered a global semiconductor shortage and a surge in demand for Tesla forced other companies such as Apple and Amazon to temper their expansion plans after the initial pandemic peak, which, in turn, affected their stock prices.

\begin{figure}[htbp]
    \centering
    \subfigure[OLS Estimator $\widehat{\mathcal{A}}^\mathtt{OLS}$]{
    \includegraphics[width=0.45\linewidth]{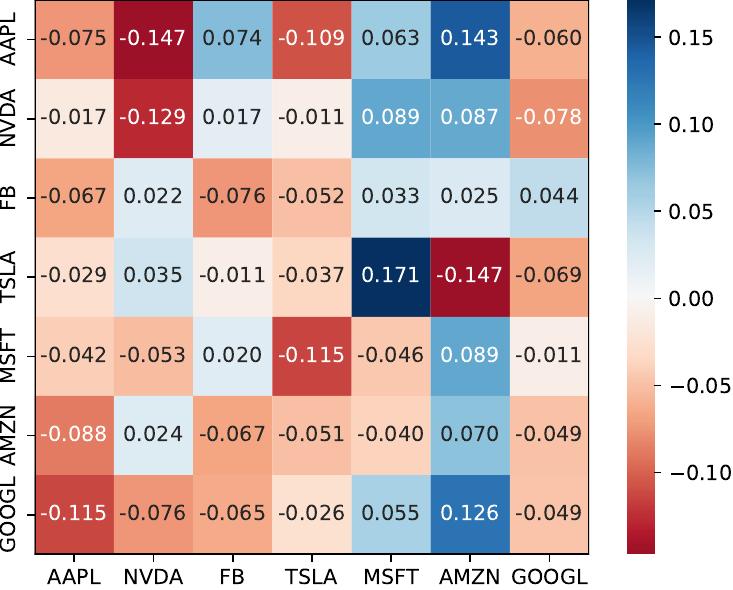}
}
\subfigure[Post-selection Estimator $\widehat{\mathcal{A}}^\mathtt{Post}$]{
    \includegraphics[width=0.45\linewidth]{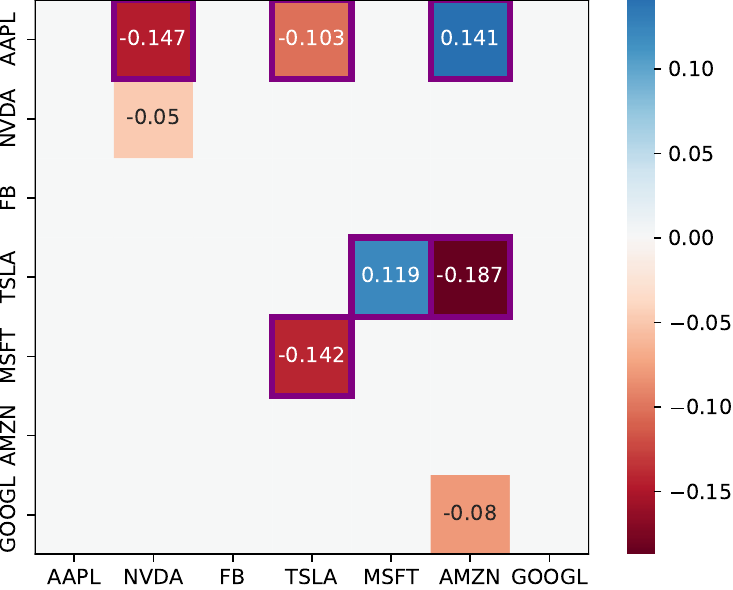}
}
\caption{Heatmaps of estimated coefficient matrices for fitted gbVAR(1) processes on the seven stocks data sample based on (a) OLS $\widehat{\mathcal{A}}^\mathtt{OLS}$ (b) Post-selection algorithm $\widehat{\mathcal{A}}^\mathtt{Post}$. Purple boxes indicate statistically significant coefficients, i.e., that -0.08 and -0.05 are not contained within their confidence intervals.}
\label{fig: heatmap}
\end{figure}

\begin{table}[!ht]
  \caption{Description and results of the daily stock data in the portfolio}
\label{tab: wild bootstrap}
    \centering
    \begin{tabular}{cccccc}
  \hline
 $d$ &
Sample size &
$\lambda$ &
$b_d$ &
$h_n$&
Length in Algorithm \ref{algo: swb}\\ 
  \hline 7 & 253 & $9.4855\times 10^{-5}$ & 0.07158 &2.74 & 0.1986\\
 \hline
    \end{tabular}
\end{table}

\subsubsection{Global trade data}
Another example concerns the annual international trade among $d = 7$ countries (China, the United States, India, Japan, Germany, France, and the United Kingdom) from the year of 1950 to 2014 ($n = 65$), with data obtained from \cite{barbieri2009trading, keshk2017correlates}. We define an edge between two countries as 1 if the imports of country $i$ from country $j$ either begins or increases by more than 10\%, and 0 otherwise. After collecting data, we fit a gbVAR(1) model with our proposed post-selection method, and demonstrate the result in Figure \ref{fig: global_econ}. 

Figure \ref{fig: global_econ} illustrates that the imports of China from Germany, France from China, India from France, and the United States from Germany retain the largest positive autoregressive coefficients on their own lags (ranging from 0.2 to 0.6) with strong persistence. The only positive cross‐regional spillover between Europe and the United States is along the import of the United States from Germany, with an estimated coefficient of about 0.3541. In Asia, the main bilateral trade flows between neighboring countries exhibit mutual positive spillovers, consistent with geographic proximity. The gbVAR model in this case features a highly sparse, modular network and represents the tightest internal connections and relatively point-to-point intercontinental connections. This modular network structure can be explained by the economic policy uncertainty within and through countries introduced by \cite{tam2018global}.

\begin{figure}[htbp]
    \centering
    \includegraphics[width=\textwidth]{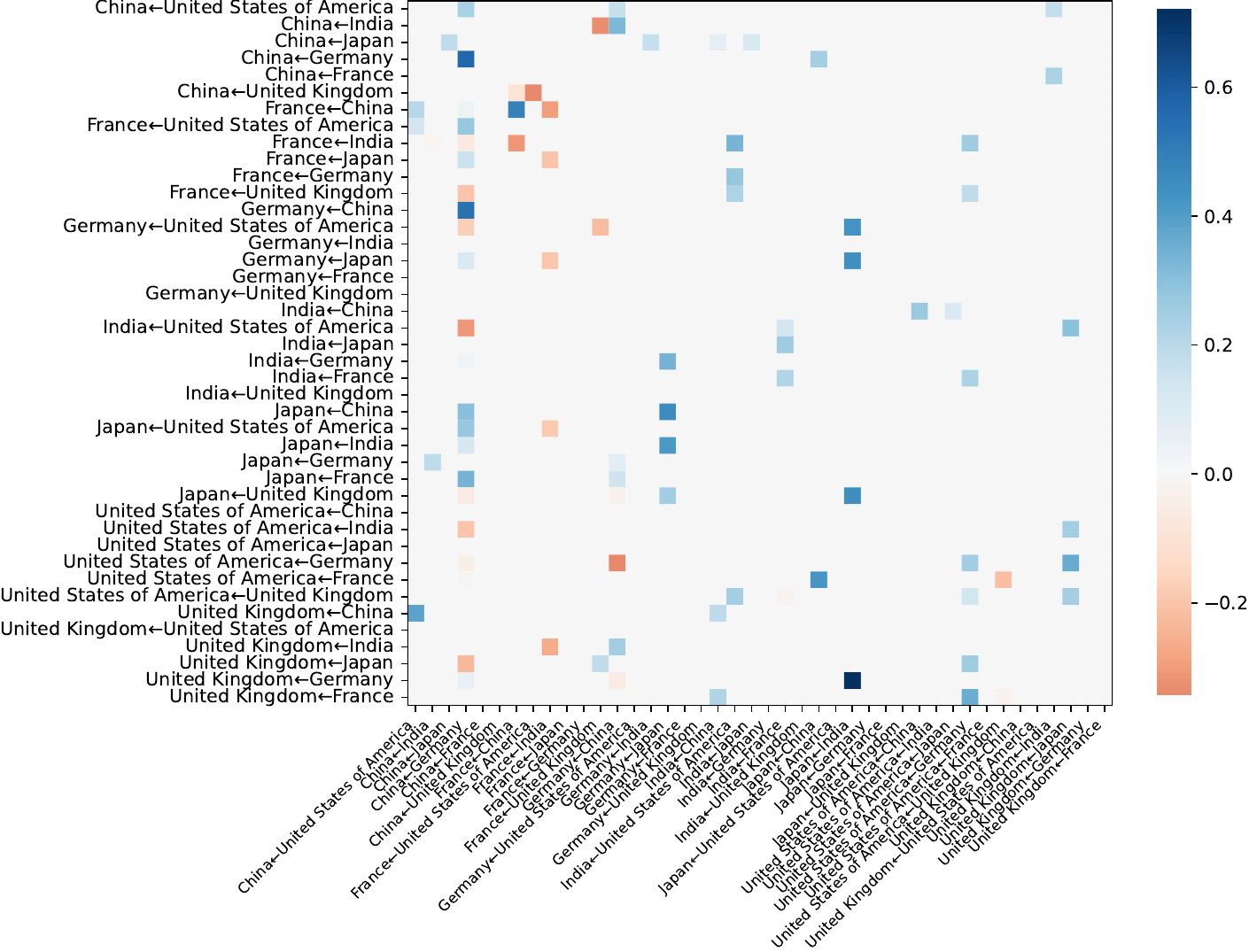}
    \caption{Estimated transition matrix of adjacency matrix for the trades among the 7 countries from 1950 to 2014. The values from -0.4 to 1 are color-coded from light red, light blue to blue.}
    \label{fig: global_econ}
\end{figure}

\section{Conclusion}
\label{sec 6}
This manuscript extends the generalized binary vector autoregressive (gbVAR) model of \cite{jentsch2022generalized} to a high-dimensional sparse binary time series setting, and  adapts the post-selection estimator of \cite{zhang2023statistical} for fitting the gbVAR model. To facilitate statistical inference, we propose a second-order wild bootstrap algorithm. Theoretically, we establish the asymptotic consistency and distributional properties of the post-selection estimator and prove the validity of the proposed bootstrap algorithm. Simulation and real-life data studies demonstrate that our method achieves good model selection and estimation performance relative to existing approaches, and has wide applications in modeling economic and trade data. 

\bibliography{main}
\newpage
\appendix

\section{Proofs in Section \ref{sec 2}}
We first recall the algebraic properties in Remark \ref{rem: partial operator} of the partial inverse operator in the following lemma.
\begin{lemma}
\label{lem: partial}
    The partial inverse operator $\mathcal{F}(\cdot)$ has the following properties:
    \begin{itemize}
        \item [(i)] If $\mathbf{A}\in \mathbb{R}^{d \times d}$ is symmetric and there exists a positive constant $c>0$ such that $\lambda_{\min}(\mathbf{A}) > c,$ then $\mathcal{F}_S(\mathbf{A})$ is also symmetric and it satisfies $|\mathcal{F}_S(\mathbf{A})|_1 \leq \frac{|S|}{c}.$ 
        \item [(ii)] Suppose $\mathbf{A}_{S}$ is invertible. If a vector $\bm{\alpha} \in \mathbb{R}^{1\times d}$ satisfies $S = \{j : \bm{\alpha}_j \neq 0\},$ then we have 
$$\bm{\alpha} \mathbf{A}\mathcal{F}_S(\mathbf{A}) = \bm{\alpha}.$$
    \end{itemize}
\end{lemma}
\begin{proof}
(i) For $\mathbf{A}$ is symmetric and positive definite, $\mathbf{A}_S$ is also symmetric and positive definite, and its inverse $\mathbf{A}_S^{-1}$ is symmetric. Therefore, $\mathcal{F}_S(\mathbf{A})$ is symmetric. For $|\mathbf{A}_S^{-1}|_2 \leq 1/c$,
$$
|\mathcal{F}_S(\mathbf{A})|_1 = |\mathbf{A}_S^{-1}|_1 \leq |S||\mathbf{A}_S^{-1}|_2 \leq \frac{|S|}{c}.
$$

(ii) For any $j = 1, \ldots, d$, if $j \not\in S$, then $\mathcal{F}_S(\mathbf{A})_{\cdot j}=0$, so the $j$th element of $\bm{\alpha} \mathbf{A}\mathcal{F}_S(\mathbf{A})$ is 0, which equals $\bm{\alpha}_j$. If $j = k_v$ for some $v= 1, \ldots, d$, then
$$
(\bm{\alpha} \mathbf{A}\mathcal{F}_S(\mathbf{A}))_j = \sum_{s=1}^{d}\bm{\alpha}_{k_s} \sum_{t=1}^d\mathbf{A}_{k_sk_t} \mathcal{F}_S(\mathbf{A})_{k_t j}= \bm{\alpha}_j.
$$
\end{proof}
\section{Proofs of main theorems in Section \ref{sec 4}}
In this section, we provide proofs of aforementioned propositions, theorems and useful lemmas.
\subsection{Proof of Proposition \ref{lem: gmc(q)}}
\begin{proof}
    By \eqref{eq: gbVMA} and \eqref{eq: zeta, eta}, we have
\begin{equation*}
   \begin{aligned}
       X_{ts} - X_{ts,\{0\}} &= \sum_{i=0}^\infty u_s^\top (\zeta_{t,i-1} \eta_{t-i} - \zeta_{t,i-1}' \eta_{t-i}')\\
       & = \sum_{i={t+1}}^\infty u_s^\top(\zeta_{t,i-1} - \zeta_{t,i-1}') \eta_{t-i} + u_s^\top\zeta_{t, t-1}(\eta_{0}-\eta_{0}'),
   \end{aligned} 
\end{equation*}
where $u_s$ is an $s$-th unit vector, 
\begin{equation*}
    \begin{aligned}
        \zeta_{t,i-1}- \zeta_{t,i-1}' &=  \prod_{j=0}^{t-1}A_{t-j}^{(+)}\left(A_{0}^{(+)}-{A_0^{(+)}}'\right) \prod_{j=t+1}^{i}A_{t-j}^{(+)},
    \end{aligned}
\end{equation*}
and
\begin{equation*}
    \begin{aligned}
        \eta_0 - \eta_0' &= A_0^{(-)}\1_d +B_0e_0 - ({A_0^{(-)}}'\1_d +B_0'e_0') \\
        & = (A_0^{(-)} - {A_0^{(-)}}')\1_d +(B_0 - B_0') e_0 + B_0'(e_0 - e_0')
    \end{aligned}
\end{equation*} 
Now, we derive the upper bound of $q$-moment of $X_{ts} - X_{ts,\{0\}}$. First, since
\begin{equation*}
    \begin{aligned}
        \||A_{t-j}^{(+)}|_\infty\|_q& = \left\|  \max_{1\leq k \leq d}\sum_{l=1}^d |A_{t-j}^{(+,kl)}|\right\|_q \leq \max_{1 \leq k \leq d} \left\| \sum_{l=1}^d |A_{t-j}^{(+,kl)}|\right\|_q \\
        & \leq \max_{1 \leq k \leq d} \left(\sum_{l=1}^d \mathbb{E}|A_{t-j}^{(+,kl)}|^m\right)^{1/q} \\
        & = \max_{1 \leq k \leq d} \left(\sum_{l=1}^d \mathcal{A}_{|\cdot|}^{kl}\right)^{1/q}:=\rho_A,
    \end{aligned}
\end{equation*}
we have
\begin{equation*}
    \begin{aligned}
        \||\zeta_{t,i-1}- \zeta_{t,i-1}'|_\infty\|_q &\leq \prod_{j =0, j \neq t-1}^i \||A_{t-j}^{(+)}|_\infty \|_q \||A_{0}^{(+)}-{A_0^{(+)}}'|_\infty \|_q\\
        & \leq  \rho_A^i \||A_{0}^{(+)}-{A_0^{(+)}}'|_\infty \|_q\\
        & := \rho_A^i \delta_{A,q} ,
    \end{aligned}
\end{equation*}
and
\begin{equation*}
    \begin{aligned}
        \||\zeta_{t,t-1}|_\infty \|_q & = \left\|\prod_{j=0}^{t}\left|A_{t-j}^{(+)} \right|_\infty\right\|_q  \leq  \left(\max_{1\leq k\leq d} \left(\sum_{l=1}^d \mathcal{A}_{|\cdot|}^{kl}\right)^{1/q}\right)^{t+1} = \rho_A^{t+1}.
    \end{aligned}
\end{equation*}
Similarly, we have
\begin{equation*}
    \begin{aligned}
        \||\eta_{t-i}|_\infty \|_q & = \||A_{t-i}^{(-)}\1_d +B_{t-i}e_{t-i} |_\infty\|_q \\
        &\leq \||A_{t-i}^{(-)}|_\infty \|_q+ \||B_{t-i}|_\infty\|_q + \||e_{t-i}|_\infty\|_q\\
        & := \rho_A+\rho_B+\rho_e.
    \end{aligned}
\end{equation*}
and 
\begin{equation*}
\begin{aligned}
     \||\eta_0 - \eta_0'|_\infty\|_q &\leq \||A_0^{(-)} - {A_0^{(-)}}'|_\infty\|_q+\||B_0 -B_0'|_\infty\|_q + \||e_0-e_0'|_\infty\|_q\\
     &:= \delta_{A,q} + \delta_{B,q}+\delta_{e,q}.
\end{aligned}
\end{equation*}
Therefore,
\begin{equation*}
\begin{aligned}
     \|X_{ts} - X_{ts,\{(0)\}}\|_q  \leq  &\sum_{i=t-1}^\infty  \||\zeta_{t,i-1}- \zeta_{t,i-1}'|_\infty\|_q \||\eta_{t-i}|_\infty\|_q +\||\zeta_{t,t-1}|_\infty\|_q \||\eta_{0} -\eta_0'|_\infty\|_q \\
      \leq &\sum_{i=t+1}^\infty \rho_A^{i}\delta_{A,q}(\rho_A+\rho_B+\rho_e) + 2\rho_A^{t+1} (\delta_{A,q} + \delta_{B,q}+\delta_{e,q})\\
       = & \frac{\rho_A^{t+1}}{1-\rho_A}\delta_{A,q}(\rho_A+\rho_B+\rho_e) + 2\rho_A^{t+1} (\delta_{A,q} + \delta_{B,q}+\delta_{e,q}).
\end{aligned}
\end{equation*}
As a result, if $\rho_A \in (0,1)$ and there exists a constant $$C = \frac{\rho_A \delta_{A,q}}{(1-\rho_A)^2} (\rho_A+\rho_B+\rho_e) + \frac{2\rho_A}{1-\rho_A}\rho_A(\delta_{A,q} + \delta_{B,q}+\delta_{e,q}), $$ we have
\begin{equation*}
    \begin{aligned}
      \|X_{.s} \|_q: = \sup_{r\geq 0}\rho_A^{-r}\sum_{t = r}^\infty \|X_{ts} - X_{ts}^{(0)} \|_q \leq C<\infty.
    \end{aligned}
\end{equation*}
\end{proof}

\subsection{Proof of Theorem \ref{thm: alpha}}
\begin{proof}
    By definition
    \begin{equation*}
        \begin{aligned}
            &\frac{1}{2d}|\hat{\bm{\Sigma}}_{i.}^{(1)} - \hat{\bm{\alpha}}_{i.}\hat{\bm{\Sigma}}^{(0)}|_2^2+\lambda |\hat{\bm{\alpha}}_{i.}|_1 \leq \frac{1}{2d}|\hat{\bm{\Sigma}}_{i.}^{(1)} - \bm{\alpha}_{i.}\hat{\bm{\Sigma}}^{(0)}|_2^2+\lambda |\bm{\alpha}_{i.}|_1\\
            &\frac{1}{2d}(\hat{\bm{\alpha}}_{i.} - \bm{\alpha}_{i.})\hat{\bm{\Sigma}}^{(0)} \hat{\bm{\Sigma}}^{(0)\top} (\hat{\bm{\alpha}}_{i.} - \bm{\alpha}_{i.})^\top \\
            \leq  & \lambda (|\bm{\alpha}_{i.}|_1  - |\hat{\bm{\alpha}}_{i.}|_1) + \frac{1}{d}(\hat{\bm{\Sigma}}_{i.}^{(1)} - \bm{\alpha}_{i.}\hat{\bm{\Sigma}}^{(0)})\hat{\bm{\Sigma}}^{(0)}(\hat{\bm{\alpha}}_{i.} -\bm{\alpha}_{i.}) \\
             \leq &\lambda \sum_{j \in \hat{S} _i}|\hat{\bm{\alpha}}_{ij} - \bm{\alpha}_{ij}| - \lambda \sum_{j \not\in \hat{S}_i}|\hat{\bm{\alpha}}_{ij}| + \frac{1}{d}|\hat{\bm{\Sigma}}_{i.}^{(1)} - \bm{\alpha}_{i.}\hat{\bm{\Sigma}}^{(0)}|_{\infty}|\hat{\bm{\Sigma}}^{(0)}|_{\max} |\hat{\bm{\alpha}}_{i.} -\bm{\alpha}_{i.}|_1
        \end{aligned}
    \end{equation*}
    From Assumption B1, we have 
    \begin{equation*}
    \begin{aligned}
        0\leq \frac{3\lambda}{2}\sum_{j \in \hat{S} }|\hat{\bm{\alpha}}_{ij} - \bm{\alpha}_{ij}| - \frac{\lambda}{2}\sum_{j \not\in \hat{S}_i}|\hat{\bm{\alpha}}_{ij}|\\
        \sum_{j \not\in \hat{S}_i}|\hat{\bm{\alpha}}_{ij}|\leq 3\sum_{j \in \hat{S}_i}|\hat{\bm{\alpha}}_{ij} - \bm{\alpha}_{ij}| \text{ for all } i = 1,2,\ldots, d
        \end{aligned}
    \end{equation*}
    Therefore, with probability tending to 1
    \begin{equation}
    \label{eq: wp1 }
        \begin{aligned}
            |\hat{\bm{\alpha}}_{i.} - \bm{\alpha}_{i.}|_1 \leq 4\sum_{j \in \hat{S}_i }|\hat{\bm{\alpha}}_{ij} - \bm{\alpha}_{ij}| \leq 4 \sqrt{| \hat{S}_i|} |\hat{\bm{\alpha}}_{ij} - \bm{\alpha}_{ij}|_2 \leq C|\hat{\bm{\alpha}}_{ij} - \bm{\alpha}_{ij}|_2 
        \end{aligned}
    \end{equation}
and from Assumption B2, for sufficiently large $n$
    \begin{equation*}
    \begin{aligned}
       & (\hat{\bm{\alpha}}_{i.} - \bm{\alpha}_{i.})\hat{\bm{\Sigma}}^{(0)} \hat{\bm{\Sigma}}^{(0)\top} (\hat{\bm{\alpha}}_{i.} - \bm{\alpha}_{i.})^\top \\
       &\geq (\hat{\bm{\alpha}}_{i.} - \bm{\alpha}_{i.}){\bm{\Sigma}}^{(0)} {\bm{\Sigma}}^{(0)\top} (\hat{\bm{\alpha}}_{i.} - \bm{\alpha}_{i.})^\top - |\hat{\bm{\Sigma}}^{(0)\top} - {\bm{\Sigma}}^{(0)}|_{\max}^2 |\hat{\bm{\alpha}}_{i.} - \bm{\alpha}_{i.}|_1^2\\
        &\geq \frac{c}{2}|\hat{\bm{\alpha}}_{i.} - \bm{\alpha}_{i.}|_2^2.
        \end{aligned}
    \end{equation*}

    Thus, we have
    \begin{equation*}
        \frac{c}{2}|\hat{\bm{\alpha}}_{i.} - \bm{\alpha}_{i.}|_2^2 \leq \frac{3\lambda}{2}\sum_{j \in \hat{S}_i} |\hat{\bm{\alpha}}_{ij} - \bm{\alpha}_{ij}|\leq C\lambda |\hat{\bm{\alpha}}_{i.} - \bm{\alpha}_{i.}|_2\Rightarrow |\hat{\bm{\alpha}}_{i.} - \bm{\alpha}_{i.}|_2 \leq C \lambda
    \end{equation*}
    for any $i$ with probability tending to 1. \eqref{eq: wp1 } implies that
    $$
|\hat{\bm{\alpha}}_{i.} - \bm{\alpha}_{i.}|_1 \leq |\hat{\bm{\alpha}}_{i.} - \bm{\alpha}_{i.}|_2 \leq C' \lambda
    $$
    with probability tending to 1. 

    Assume that $\min_{i = 1,\cdots, n, j\in S_i}\vert\bm{\alpha}_{ij}\vert > 2\lambda$. Since $\vert a\vert \geq \vert b\vert - \vert a  - b\vert$, 
we have
    \begin{align*}
    &\mathbb{P}\left(\bigcup_{i = 1}^d\hat{S}_i\neq S_i\right)
    \leq \mathbb{P}\left(\bigcup_{i = 1}^d \exists j\in S_i, j\not \in\hat{S}_i\right) + \mathbb{P}\left(\bigcup_{i = 1}^d \exists j\not\in S_i, j\in \hat{S}_i\right)\\
    \leq& \mathbb{P}\left(\min_{i = 1,\cdots, p,j\in S_i}\vert\hat{\bm{\alpha}}_{ij}\vert\leq \lambda\right) + \mathbb{P}\left(\max_{i = 1,\cdots, p,j\not\in S_i}\vert\hat{\bm{\alpha}}_{ij}\vert > \lambda\right)\\
    \leq &\mathbb{P}\left(\min_{i = 1,\cdots, p,j\in S_i}2\lambda  -\vert\hat{\bm{\alpha}}_{ij} - \bm{\alpha}_{ij}\vert\leq \lambda\right) + \mathbb{P}\left(\max_{i = 1,\cdots, p,j\not\in S_i}\vert\hat{\bm{\alpha}}_{ij} - \bm{\alpha}_{ij}\vert > \lambda\right)\\
    \leq &2\mathbb{P}\left(\max_{i,j}\vert \hat{\bm{\alpha}}_{ij} - \bm{\alpha}_{ij}\vert > \lambda\right) =o(1),
\end{align*}
and we prove \eqref{eq: index set}.
\end{proof}

\subsection{Proof of Lemma \ref{thm: GA}}
\begin{proof}
  Denote $\mathbb{E}_0( \cdot) = \cdot - \mathbb{E}(\cdot)$ and sum of $X_i$ and $\tilde{X}_i$ defined in Lemma \ref{lem: momeng for sum},
  $$
    T_X = \sum_{i=1}^n X_i \quad \text{ and } \quad T_{\tilde{X}}= \sum_{i=1}^n \tilde{X}_i.
    $$
    By Lemma 2 in \cite{zhang2023debiased}, we have a decomposition
    \begin{equation*}
        \begin{aligned}
            &|\mathbb{P}(\sqrt{n}|\bar{X} -\mu|_\infty \leq x)-\mathbb{P}(|Z|_\infty \leq x)| \\
            \leq &\sup_{x\in \mathbb{R}}\left|\mathbb{E}h_{\psi,\psi, x}\left(\frac{1}{\sqrt{n}}\mathbb{E}_0 T_{X,1},\ldots, \frac{1}{\sqrt{n}}\mathbb{E}_0 T_{X,d}\right) -\mathbb{E}h_{\psi,\psi, x}(Z_1, \ldots, Z_d)\right|\\
           & \qquad+ Ct(1+\sqrt{\log d} + \sqrt{|\log t|}),
        \end{aligned}
    \end{equation*}
    where $t= (1+\log (2d))/\psi$ and $h_{\psi,\psi,x}(\cdot)$ is well-defined in \cite{zhang2023debiased}.
    By Lemma \ref{lem: momeng for sum}, we have
    $$
\| T_{X,j} - T_{\tilde{X},j}\|_2 = \|\sum_{i=1}^n \mathbb{E}_0 (X_{ij} - \tilde{X}_{ij})\|_2 \leq  \sqrt{n}\rho^m \|X.\|_2.
    $$
    Let $\ell = 1\vee \log d$. We have
    \begin{equation*}
        \begin{aligned}
            &\left|\mathbb{E} h_{\psi,\psi, x}\left(\frac{1}{\sqrt{n}}\mathbb{E}_0 T_{X,1}, \ldots, \frac{1}{\sqrt{n}}\mathbb{E}_0 T_{X,d}\right)-\mathbb{E}h_{\psi,\psi, x}\left(\frac{1}{\sqrt{n}} \mathbb{E}_0 T_{\tilde{X},1}\ldots, \frac{1}{\sqrt{n}} \mathbb{E}_0T_{\tilde{X},d}\right)\right|\\
            \leq & \frac{C\psi}{\sqrt{n}}\|\max_{j=1,\ldots, d}| T_{X,j} - T_{\tilde{X},j}|\|_2 \\
            \leq& C\psi \sqrt{\ell}\rho^{m} \|X_{\cdot}\|_2.
        \end{aligned}
    \end{equation*}
Suppose $n = (M+m)w$, where $M \gg m$ and $M, m, w \to\infty$ as $n \to \infty.$ We apply the block technique and split the interval $[1,n]$ into alternating large blocks $L_b = [(b-1)(M+m)+1, bM+(b-1)M]$ and small blocks $S_b = [bM+(b-1)m+1, b(M+m)]$, $1\le b \le w$. Let 
$$
Y_b = \sum_{i \in L_b}X_i, \qquad \tilde{Y}_b =\sum_{i\in L_b}\tilde{X}_i, \qquad T_Y = \sum_{b=1}^w Y_b, \qquad T_{\tilde{Y}} = \sum_{b=1}^w \tilde{Y}_b.
$$
By inequality (ii) in Lemma \ref{lem: momeng for sum}, we have
$$
\|\frac{1}{\sqrt{n}}(\mathbb{E}_0T_{\tilde{X}} - \mathbb{E}_0T_{\tilde{Y}})\|_2  = \|\frac{1}{\sqrt{n}}\sum_{b=1}^w \sum_{i \in S_b}\mathbb{E}_0\tilde{X}_i\|_2\leq \frac{\sqrt{wm}}{\sqrt{n}}\|X_.\|_2
$$
which implies
\begin{equation}
    \label{eq: decomp 1}
    \begin{aligned}
   & \sup_{x\in \mathbb{R}}\left|\mathbb{E}h_{\psi,\psi, x}\left(\frac{1}{\sqrt{n}}\mathbb{E}_0 T_{\tilde{X},1},\ldots, \frac{1}{\sqrt{n}}\mathbb{E}_0T_{\tilde{X},d}  \right)- \mathbb{E}h_{\psi,\psi, x}\left(\frac{1}{\sqrt{n}}\mathbb{E}_0 T_{\tilde{Y},1},\ldots,\frac{1}{\sqrt{n}}\mathbb{E}_0 T_{\tilde{Y},d}  \right)\right|\\
        \leq &C_2 \psi \|\max_{j= 1,\ldots, d} |\frac{1}{\sqrt{n}}(\mathbb{E}_0T_{\tilde{X}} - \mathbb{E}_0T_{\tilde{Y}})|\|_2 \leq C_2 \psi\frac{ \sqrt{\ell wm}}{\sqrt{n}} \|X_.\|_2 
        \end{aligned}
\end{equation}
Let $Z_b, 1\leq b\leq w$, be i.i.d $N(0, MB)$ and $\tilde{Z}_b$ be i.i.d $N(0, M\tilde{B})$, where the covariance matrices $B$ and $\tilde{B}$ are respectively given by
$$
B=(b_{ij})_{i,j=1}^d = \mathrm{Cov}(Y_b/\sqrt{M}) \quad \text{and}\quad \tilde{B}=(\tilde{b}_{ij})_{i,j=1}^d = \mathrm{Cov}(\tilde{Y}_b/\sqrt{M}).
$$
Write $T_{\tilde{Z}}=\sum_{b=1}^w \tilde{Z}_b$ and let $Z \sim N(0, \bm{\Sigma})$. Define
$$
H_k = \sum_{b=1}^{k-1}\tilde{Y}_b+\sum_{b=k+1}^w \tilde{Z}_b,
$$
satisfying $H_k + \tilde{Y}_k = H_{k+1}+\tilde{Y}_{k+1}$. Thus,

\begin{equation}
\label{eq: decomp 2}
    \begin{aligned}
  &  \left|\mathbb{E}h_{\psi,\psi,x}\left(\frac{1}{\sqrt{n}}\mathbb{E}_0T_{\tilde{Y},1}, \ldots, \frac{1}{\sqrt{n}}\mathbb{E}_0T_{\tilde{Y},d}\right) - \mathbb{E}h_{\psi,\psi,x}\left(\frac{1}{\sqrt{n}}\mathbb{E}_0T_{\tilde{Z},1}, \ldots, \frac{1}{\sqrt{n}}\mathbb{E}_0T_{\tilde{Z},d}\right)\right|\\
   =& \Bigg|\mathbb{E}h_{\psi,\psi,x}\left(\frac{1}{\sqrt{n}}\mathbb{E}_0 (H_{w,1}+\tilde{Y}_{w,1}), \ldots, \frac{1}{\sqrt{n}}\mathbb{E}_0(H_{w,d}+\tilde{Y}_{w,d})\right)  \\
   &\qquad -\mathbb{E}h_{\psi,\psi,x}\left(\frac{1}{\sqrt{n}}\mathbb{E}_0 (H_{1,1}+\tilde{Z}_{1,1}), \ldots, \frac{1}{\sqrt{n}}\mathbb{E}_0(H_{1,d}+\tilde{Z}_{1,d})\right) \Bigg|\\
   \leq & \sum_{k=1}^w\Bigg|\mathbb{E}h_{\psi,\psi,x}\left(\frac{1}{\sqrt{n}}\mathbb{E}_0 (H_{k,1}+\tilde{Y}_{k,1}), \ldots, \frac{1}{\sqrt{n}}\mathbb{E}_0(H_{k,d}+\tilde{Y}_{k,d})\right)  \\
   &\qquad -\mathbb{E}h_{\psi,\psi,x}\left(\frac{1}{\sqrt{n}}\mathbb{E}_0 (H_{k,1}+\tilde{Z}_{k,1}), \ldots, \frac{1}{\sqrt{n}}\mathbb{E}_0(H_{k,d}+\tilde{Z}_{k,d})\right) \Bigg|\\
   \leq &C w\psi^3 \ell^{3/2}\max_{j} \|\frac{1}{\sqrt{n}}\mathbb{E}_0T_{\tilde{Y},j}\|_2^3\\
   \stackrel{(a)}{\leq} &C  w\psi^3 \ell^{3/2} \left( \frac{\sqrt{M}}{\sqrt{n}}\|X_.\|_2\right)^3 \leq C\psi^3 \ell^{3/2} \frac{\sqrt{M}}{\sqrt{n}}\|X_.\|_2^3
    \end{aligned}
\end{equation}
where the inequality $(a)$ is from
$$
\|\frac{1}{\sqrt{n}}\mathbb{E}_0T_{\tilde{Y},j}\|_2 = \frac{1}{\sqrt{n}}\|\mathbb{E}_0T_{\tilde{Y},j}\|_2 \leq \frac{\sqrt{M}}{\sqrt{n}}\|X_.\|_2
$$
for any $j\in[d]$. By the definition of $T_{\tilde{Z}}$ and its covariance matrices $\tilde{B}$, 
\begin{align*}
\bm{\Sigma}^{\tilde{Z}}:=\mathrm{Cov}(T_{\tilde{Z}}/\sqrt{n}) = \frac{Mw}{n}\tilde{B}.
\end{align*}
Similar to Lemma 7.3 from \cite{zhang2017gaussian} and (B.30) in \cite{zhang2023debiased},
\begin{equation}
\label{eq: decomp 3}
    \begin{aligned}
   & \left|\mathbb{E}h_{\psi,\psi,x}\left(\frac{1}{\sqrt{n}}\mathbb{E}_0T_{\tilde{Z},1}, \ldots, \frac{1}{\sqrt{n}}\mathbb{E}_0T_{\tilde{Z},d}\right)-\mathbb{E}h_{\psi,\psi,x}\left(\frac{1}{\sqrt{n}}\mathbb{E}_0T_{{Z},1}, \ldots, \frac{1}{\sqrt{n}}\mathbb{E}_0T_{{Z},d}\right)\right|\\
    \leq & C\psi \sqrt{\ell}\|\frac{1}{\sqrt{n}}\sum_{b=1}^w \sum_{i \in S_b}\mathbb{E}_0\tilde{Z}_i\|_2 +   C\psi^2  |\bm{\Sigma}^{\tilde{Z}}-\bm{\Sigma}|_{\max}\\
    \leq &  C\psi \sqrt{\ell}\frac{\sqrt{wm}}{\sqrt{n}} +  \frac{C\psi^2Mw}{n} (|\tilde{B}-B|_{\max} +|B-\bm{\Sigma}|_{\max}) +C\psi^2 \left(1-\frac{Mw}{n}\right)|\bm{\Sigma}|_{\max}\\
   \leq & C\psi \sqrt{\ell}\frac{\sqrt{wm}}{\sqrt{n}} + \frac{CMw\psi^2}{n} \|X_.\|_2^2\left(\rho^{2m}+\frac{\log M}{M}\right) + \frac{C_2wm}{n}
        \end{aligned}
\end{equation}

From \eqref{eq: decomp 1}, \eqref{eq: decomp 2} and \eqref{eq: decomp 3}, we prove the Lemma.
\end{proof}

\subsection{Proof of Theorem \ref{thm: GA post selection}}
\begin{proof}
    From Theorem \ref{thm: alpha}, $\mathbb{P}(\cap_{i=1,\ldots, d}\{\hat{S}_i =S_i\})\rightarrow 1$ as the sample size $n \rightarrow \infty$. For $i =1, \ldots, d$, $$\tilde{\bm{\alpha}}_{i.} = \hat{\bm{\Sigma}}_{i.}^{(1)}\hat{\bm{\Sigma}}^{(0)}\mathcal{F}_{S_{i}}(\hat{\bm{\Sigma}}^{(0)\top}\hat{\bm{\Sigma}}^{(0)})$$
From Lemma \ref{lem: Sigma hat}, we can denote $\bm{\Delta}^{(k)} = \hat{\bm{\Sigma}}^{(k)} - \bm{\Sigma}^{(k)}$ with
$$
\bm{\Delta}^{(k) } = \frac{1}{n} \sum_{t = 1}^{n-k} (X_{t+k} - \Bar{X})(X_{t} - \Bar{X})^\top - \bm{\Sigma}^{(k)} := \frac{1}{n}\sum_{t = 1}^{n-k} \bm{\Gamma}_t^{(k)}, \text{ and } \bm{\Gamma}_t^{(k)} = \mathbb{E}_0z_{t+k}z_{t}^\top ,
$$
here $z_t = X_t - \Bar{X}$ and $t\in \mathbb{Z}$ and $\mathbb{E}_0 \cdot = \cdot - \mathbb{E}\cdot$. Therefore, for $i = 1,\ldots, d$,
\begin{equation}
\label{eq: formula of CLT}
    \begin{aligned}
    \tilde{\bm{\alpha}}_{i.} - \bm{\alpha}_{i.} = &\hat{\bm{\Sigma}}_{i.}^{(1)}\hat{\bm{\Sigma}}^{(0)}\mathcal{F}_{S_{i}}(\hat{\bm{\Sigma}}^{(0)\top}\hat{\bm{\Sigma}}^{(0)})-
 \bm{\Sigma}_{i.}^{(1)}\bm{\Sigma}^{(0)}\mathcal{F}_{S_{i}}(\bm{\Sigma}^{(0)\top}\bm{\Sigma}^{(0)})\\
=&(\bm{\Sigma}_{i.}^{(1)} + \bm{\Delta}_{i.}^{(1)})(\bm{\Sigma}^{(0)} + \bm{\Delta}^{(0)})\mathcal{F}_{S_i}(\hat{\bm{\Sigma}}^{(0)\top }\hat{\bm{\Sigma}}^{(0)} )-  \bm{\Sigma}_{i.}^{(1)}\bm{\Sigma}^{(0)}\mathcal{F}_{S_i}(\bm{\Sigma}^{(0)\top } \bm{\Sigma}^{(0)} ) \\
=&\bm{\Sigma}_{i.}^{(1)} \bm{\Sigma}^{(0)} (\mathcal{F}_{S_i}(\hat{\bm{\Sigma}}^{(0)\top }\hat{\bm{\Sigma}}^{(0)} )  - \mathcal{F}_{S_i}(\bm{\Sigma}^{(0)\top } \bm{\Sigma}^{(0)} )) \\
&+ ( \bm{\Delta}_{i.}^{(1)} \bm{\Sigma}^{(0)} + \bm{\Sigma}_{i.}^{(1)} \bm{\Delta}^{(0)} +  \bm{\Delta}_{i.}^{(1)} \bm{\Delta}^{(0)} )\mathcal{F}_{S_i}(\hat{\bm{\Sigma}}^{(0)\top }\hat{\bm{\Sigma}}^{(0)} ) 
    \end{aligned}
\end{equation}
Now, for any $j = 1, \ldots, d$
\begin{align*}
   \bm{\kappa}^{(t)} =
 \begin{bmatrix}
\ve\left( \left[\bm{\Gamma}_{t,1.}^{(1)}\bm{\Sigma}^{(0)} + \bm{\Sigma}_{1.}^{(1)}\bm{\Gamma}_{t}^{(0)} -\bm{\Sigma}_{1.}^{(1)}\bm{\Sigma}^{(0)} \mathcal{F}_{S_1}(\bm{\Sigma}^{(0)} \bm{\Sigma}^{(0)}) (\bm{\Sigma}^{(0)}\bm{\Gamma}_{t}^{(0)} + \bm{\Gamma}_{t}^{(0)}\bm{\Sigma}^{(0)})\right]\mathcal{F}_{S_1}(\bm{\Sigma}^{(0)}\bm{\Sigma}^{(0)})\right)
\\
\vdots  \\
 \ve\left( \left[\bm{\Gamma}_{t,d.}^{(1)}\bm{\Sigma}^{(0)} + \bm{\Sigma}_{d.}^{(1)}\bm{\Gamma}_{t}^{(0)} -\bm{\Sigma}_{d.}^{(1)}\bm{\Sigma}^{(0)} \mathcal{F}_{S_d}(\bm{\Sigma}^{(0)} \bm{\Sigma}^{(0)}) (\bm{\Sigma}^{(0)}\bm{\Gamma}_{t}^{(0)} + \bm{\Gamma}_{t}^{(0)}\bm{\Sigma}^{(0)})\right]\mathcal{F}_{S_d}(\bm{\Sigma}^{(0)}\bm{\Sigma}^{(0)})\right)
\end{bmatrix}
.
\end{align*}
Then,
$$
\begin{aligned}
    &\mathbb{E}\left( \bm{\Gamma}^{(1)}_{t,i.} \bm{\Sigma}^{(0)} + \bm{\Sigma}^{(1)}_{i.} \bm{\Gamma}^{(0)}_{t}- \bm{\Sigma}^{(1)}_{i.} \bm{\Sigma}^{(0)}\mathcal{F}_{S_i}(\bm{\Sigma}^{(0)\top } \bm{\Sigma}^{(0)} )  (\bm{\Gamma}^{(0)}_{t}\bm{\Sigma}^{(0)} + \bm{\Sigma}^{(0)}\bm{\Gamma}^{(0)}_{t}) \right) \mathcal{F}_{S_i}(\bm{\Sigma}^{(0)\top } \bm{\Sigma}^{(0)} ) \\
= &\left(\mathbb{E} \bm{\Gamma}^{(1)}_{t,i.} \bm{\Sigma}^{(0)} + \bm{\Sigma}^{(1)} \mathbb{E}\bm{\Gamma}^{(0)}_{t} - \bm{\Sigma}^{(1)}_{i.} \bm{\Sigma}^{(0)} \mathcal{F}_{S_i}(\bm{\Sigma}^{(0)\top } \bm{\Sigma}^{(0)} ) (\mathbb{E}\bm{\Gamma}^{(0)}_{t}\bm{\Sigma}^{(0)} + \bm{\Sigma}^{(0)} \mathbb{E}\bm{\Gamma}^{(0)}_{t}) \right) \mathcal{F}_{S_i}(\bm{\Sigma}^{(0)\top } \bm{\Sigma}^{(0)} )
=\bm{0}.
\end{aligned}
$$
Besides, for any $i, j = 1, \ldots, d$,
\begin{align*}
   & \left\|\left( \left[\bm{\Gamma}_{t,i.}^{(1)}\bm{\Sigma}^{(0)} + \bm{\Sigma}_{i.}^{(1)}\bm{\Gamma}_{t}^{(0)} -\bm{\Sigma}_{i.}^{(1)}\bm{\Sigma}^{(0)} \mathcal{F}_{S_i}(\bm{\Sigma}^{(0)} \bm{\Sigma}^{(0)}) (\bm{\Sigma}^{(0)}\bm{\Gamma}_{t}^{(0)} + \bm{\Gamma}_{t}^{(0)}\bm{\Sigma}^{(0)})\right]\mathcal{F}_{S_i}(\bm{\Sigma}^{(0)}\bm{\Sigma}^{(0)})\right)_j\right\|_{q/2}\\
\leq &|\mathcal{F}_{S_i}(\bm{\Sigma}^{(0)}\bm{\Sigma}^{(0)})|_1 \left(\max_{p=1,\ldots,d} \|\bm{\Gamma}_{t,ip}^{(1)}\|_{q/2}|\bm{\Sigma}^{(0)}|_1+|\bm{\Sigma}_{i.}^{(1)}|_1 \|\bm{\Gamma}_{t}^{(0)}\|_{q/2} \right.\\
&+\left.2|\bm{\Sigma}_{i.}^{(1)}\bm{\Sigma}^{(0)} \mathcal{F}_{S_i}(\bm{\Sigma}^{(0)} \bm{\Sigma}^{(0)})|_1 |\bm{\Sigma}^{(0)}|_1\|\bm{\Gamma}_{t}^{(0)} \|_{q/2}\right)\leq C_0
\end{align*}
for constant $C_0$ since $\|\bm{\Gamma}_{t}^{(0)} \|_{q/2} \leq C \|z_{t}\|_q^2 \leq C_1$.
Therefore, by Proposition \ref{lem: gmc(q)}, $\bm{\kappa}^{(t)}$ satisfies GMC($q/2$). For
$$
\max_{i,j = 1,\ldots, d}\frac{1}{\sqrt{n}}\left|\sum_{t=1}^n \bm{\kappa}^{(t)}_{(i-1)\times d +j}\right| = \max_{i = 1,\ldots, d,j \in \mathcal{S}_i}\frac{1}{\sqrt{n}}\left|\sum_{t=1}^n \bm{\kappa}^{(t)}_{(i-1)\times d +j}\right|, 
$$
by Lemma \ref{thm: GA}, we have
\begin{align*}
    \sup_{x\in \mathbb{R}}\left|\mathbb{P}\left(\max_{i,j = 1,\ldots, d}\frac{1}{\sqrt{n}}|\sum_{t=1}^n \bm{\kappa}^{(t)}_{(i-1)\times d +j}| \leq x\right)-\mathbb{P}\left(\max_{i=1,\ldots, d, j\in S_i}|{Z}_{ij}^*|\leq x\right)\right|=o(1)
\end{align*}
where ${Z}_{ij}^*$ is joint normal random variable with mean 0 and 
\begin{align*}
    \cov({Z}^*_{i_1, j_1}, {Z}^*_{i_2, j_2})= \frac{1}{n}\sum_{i_1 = 1}^n \sum_{i_2=1}^n \mathbb{E} \bm{\kappa}_{j_1}^{(i_1)} \bm{\kappa}_{j_2}^{(i_2)}.
\end{align*}

On the other hand, for any vector $\bm{c} \in \mathbb{R}^{|S_i|}$, 
\begin{align*}
        \bm{c}^\top (\hat{\bm{\Sigma}}^{(0)\top } \hat{\bm{\Sigma}}^{(0) } )_{S_i}\bm{c} \geq &\bm{c}^\top (\bm{\Sigma}^{(0)\top } \bm{\Sigma}^{(0)} )_{S_i}\bm{c} - |\bm{c}^\top ((\bm{\Sigma}^{(0)\top } \bm{\Sigma}^{(0)} )_{S_i} - (\hat{\bm{\Sigma}}^{(0)\top } \hat{\bm{\Sigma}}^{(0) } )_{S_i})\bm{c} | \\
        \geq &c|\bm{c}|_2^2 - |(\bm{\Sigma}^{(0)\top } \bm{\Sigma}^{(0)} )_{S_i} - (\hat{\bm{\Sigma}}^{(0)\top } \hat{\bm{\Sigma}}^{(0) } )_{S_i})|_{\max} |\bm{c}|_1^2\\
        \geq &c|\bm{c}|_2^2 - |\bm{\Sigma}^{(0)\top } \bm{\Sigma}^{(0)} - \hat{\bm{\Sigma}}^{(0)\top } \hat{\bm{\Sigma}}^{(0) }  |_{\max}\cdot  \max_{i = 1, 2,\ldots, d} |S_i|\cdot |\bm{c}|_2^2  \geq \frac{c}{2}|\bm{c}|_2^2
        \end{align*}
        and
        \begin{align*}
        |\bm{\Sigma}^{(0)\top } \bm{\Sigma}^{(0)} - \hat{\bm{\Sigma}}^{(0)\top } \hat{\bm{\Sigma}}^{(0) }  |_2&\leq |\bm{\Sigma}^{(0)\top } \bm{\Sigma}^{(0)} - \hat{\bm{\Sigma}}^{(0)\top } \hat{\bm{\Sigma}}^{(0) }  |_{\max} \cdot  \max_{i = 1, 2,\ldots, d} |S_i|
    \end{align*}
    for sufficiently large $n$ with high probability. Therefore, all $(\hat{\bm{\Sigma}}^{(0)\top } \hat{\bm{\Sigma}}^{(0) } )_{S_i}$ are positive definite. 

Therefore, if $\hat{S}_i = S_i$ for $i = 1, \ldots,d$, then
\begin{equation*}
    \begin{aligned}
       & \max_{i,j=1, \ldots, d}|\sqrt{n} (\tilde{\bm{\alpha}}_{ij}-\bm{\alpha}_{ij}) - \frac{1}{\sqrt{n}}\sum_{t=1}^n \bm{\kappa}^{(t)}_{(i-1)\times d +j}| \\
        \leq & \sqrt{n} \max_{i=1,\ldots, d}  |       \left[ \bm{\Delta}_{i.}^{(1)} \bm{\Delta}^{(0)}  - \bm{\Sigma}_{i.}^{(1)} \bm{\Sigma}^{(0)} \mathcal{F}_{S_i}(\bm{\Sigma}^{(0)\top } \bm{\Sigma}^{(0)}) \bm{\Delta}^{(0)}\bm{\Delta}^{(0)}\right. \\
& \left.- ( \bm{\Delta}_{i.}^{(1)} \bm{\Sigma}^{(0)} + \bm{\Sigma}_{i.}^{(1)} \bm{\Delta}^{(0)} )\mathcal{F}_{S_i}(\bm{\Sigma}^{(0)\top } \bm{\Sigma}^{(0)} ) ( \bm{\Delta}^{(0)} \bm{\Sigma}^{(0)} + \bm{\Sigma}^{(0)} \bm{\Delta}^{(0)} )   \right]\mathcal{F}_{S_i}(\bm{\Sigma}^{(0)\top }\bm{\Sigma}^{(0)})|_\infty\\
 \leq &\sqrt{n} \max_{i=1,\ldots, d}  |  \bm{\Delta}_{i.}^{(1)} \bm{\Delta}^{(0)}  - \bm{\Sigma}_{i.}^{(1)} \bm{\Sigma}^{(0)} \mathcal{F}_{S_i}(\bm{\Sigma}^{(0)\top } \bm{\Sigma}^{(0)}) \bm{\Delta}^{(0)}\bm{\Delta}^{(0)} \\
& - ( \bm{\Delta}_{i.}^{(1)} \bm{\Sigma}^{(0)} + \bm{\Sigma}_{i.}^{(1)} \bm{\Delta}^{(0)} )\mathcal{F}_{S_i}(\bm{\Sigma}^{(0)\top } \bm{\Sigma}^{(0)} ) ( \bm{\Delta}^{(0)} \bm{\Sigma}^{(0)} + \bm{\Sigma}^{(0)} \bm{\Delta}^{(0)} ) |_\infty  \max_{i=1,\ldots, d}| \mathcal{F}_{S_i}(\bm{\Sigma}^{(0)\top }\bm{\Sigma}^{(0)})|_1\\
    \end{aligned}
\end{equation*}
Besides, 
$$
\begin{aligned}
    |\bm{\Delta}^{(1)} \bm{\Delta}^{(0)}|_{\max}  \leq \max_{k, j \in S_i} \left|\sum_{l=1}^d \bm{\Delta}_{kl}^{(1)} \bm{\Delta}_{lj}^{(0)} \right| \leq |S_i|\max_{k \in S_i, l = 1, \ldots, d} |\bm{\Delta}_{kl}^{(1)}|\max_{j \in S_i, l = 1, \ldots, d} |\bm{\Delta}_{lj}^{(0)}|\leq C
\end{aligned}
$$
and by Lemma \ref{lem: Sigma hat}, 
$$
\begin{aligned}
\||\bm{\Delta}^{(1)}\bm{\Delta}^{(0)}|_{\max}\|_{q/2} = O\left(\sqrt{\frac{\log d}{n}} \right), \text{ and } \||\bm{\Delta}^{(0)}\bm{\Delta}^{(0)}|_{\max}\|_{q/2} = O\left(\sqrt{\frac{\log d}{n}}  \right).
\end{aligned}
$$
Therefore, it implies
\begin{equation*}
    \max_{i,j = 1, \ldots, d}|\sqrt{n}(\tilde{\bm{\alpha}}_{ij}-\bm{\alpha}_{ij}) - \frac{1}{\sqrt{n}}\sum_{t=1}^n \bm{\kappa}^{(t)}_{(i-1)\times d +j}| =O_{\mathbb{P}}\left(\sqrt{\frac{\log d}{n}} \right).
\end{equation*}
Thus, for any $w>0$ and sufficiently large $n$, we have
\begin{align*}
&\mathbb{P}\left(\sqrt{n}|\tilde{\mathcal{A}} - \mathcal{A} |_{\max}\leq x \right) \\
\leq &w +\mathbb{P}\left(\max_{i,j=1,\ldots, d}|\frac{1}{\sqrt{n}}\sum_{t=1}^n \bm{\kappa}^{(t)}_{(i-1)\times d +j}| \leq x + C_w \sqrt{\log d /{n}} \right)\\
\leq &\mathbb{P}(\max_{i = 1,\ldots, d, j \in S_i}|Z_{ij}|\leq x)+w+ C \log^{3/2} (d \vee n)/ \sqrt{n}
\end{align*}
and
\begin{align*}
    &\mathbb{P}\left(\sqrt{n}|\hat{\mathcal{A}} - \mathcal{A} |_{\max} \leq x \right)\\
    \geq &-w+\mathbb{P}\left(\max_{i,j=1,\ldots, d}|\frac{1}{\sqrt{n}}\sum_{t=1}^n \bm{\kappa}^{(t)}_{(i-1)\times d +j}| \leq x - C_w \sqrt{\log d /{n}} \right)\\
\geq &\mathbb{P}(\max_{i = 1,\ldots, d, j \in S_i}|Z_{ij}|\leq x)-w- C \log^{3/2} (d \vee n)/ \sqrt{n}\\
\end{align*}
which implies \eqref{eq: GA post selection}.
\end{proof}

\subsection{Proof of Theorem \ref{thm: swb}}
\begin{proof}
    From Theorem \ref{thm: alpha}, we have $\hat{S}_i = S_i$ for $i = 1, \ldots , d$ with probability tending to 1. If $\hat{S}_i = S_i$ , then
    $$
\bm{\Delta}_{ij}^* = \sqrt{n}(\tilde{\bm{\alpha}}_{ij}^* - \tilde{\bm{\alpha}}_{ij}) = \frac{1}{\sqrt{n}} \sum_{t = 1}^{n-1} \sum_{k, l = 1}^d e_t\hat{\bm{\Theta}}_{t,i k} \hat{\bm{\Sigma}}^{(0)}_{kl}\mathcal{F}_{S_i}(\hat{\bm{\Sigma}}^{(0)\top}\hat{\bm{\Sigma}}^{(0)})_{lj},
    $$
    if $j \not\in S_i$, then $\bm{\Delta}_{ij}^* = 0$. If $j_1 \in S_{i_1}$ and $j_2 \in S_{i_2}$, then for sufficiently large $n$ we have
    \begin{equation*}
    \begin{aligned}
        \mathbb{E}^*\bm{\Delta}_{i_1,j_1}^* \bm{\Delta}_{i_2,j_2}^* = \frac{1}{n}\sum_{t_1,t_2=1}^{n-1}\sum_{k_1,l_1=1}^d \sum_{k_2,l_2=1}^d &\hat{\bm{\Theta}}_{t_1, i_1, k_1} \hat{\bm{\Theta}}_{t_2, i_2, k_2}\hat{\bm{\Sigma}}_{k_1,l_1}^{(0)} \hat{\bm{\Sigma}}_{k_2,l_2}^{(0)} \\
       & \mathcal{F}_{S_i}(\hat{\bm{\Sigma}}^{(0)}\hat{\bm{\Sigma}}^{(0)})_{l_1,j_1}\mathcal{F}_{S_i}(\hat{\bm{\Sigma}}^{(0)}\hat{\bm{\Sigma}}^{(0)})_{l_2,j_2} K\left(\frac{t_1 -t_2}{h_n}\right)
    \end{aligned}
\end{equation*}
Notice that for $i = 1, \ldots, d$, denote $z_{t} = X_t - \bar{X}$ and $\mathbb{E}_0(\cdot) = \cdot - \mathbb{E}\cdot$ we have
\begin{equation*}
    \begin{aligned}
        \hat{\bm{\Theta}}_{t, ij}& = (X_{t+1,i} - \bar{X}_{i})(X_{tj} - \bar{X}_j) - \sum_{k = 1}^d \tilde{\bm{\alpha}}_{ik}(X_{tk}-\bar{X}_k)(X_{tj} -\bar{X}_j)\\
        &= z_{t+1,i}z_{t j} - \sum_{k=1}^d  \tilde{\bm{\alpha}}_{ik} z_{tk}z_{tj}\\
        &= z_{t+1,i}z_{t j} - \sum_{k=1}^d \bm{\alpha}_{ik}  z_{tk}z_{t j} - \sum_{k=1}^d (\tilde{\bm{\alpha}}_{ik} - \bm{\alpha}_{ik})  z_{tk}z_{t j}\\
        &= \mathbb{E}_0  z_{t+1,i}z_{t j} - \sum_{k=1}^d \bm{\alpha}_{ik} \mathbb{E}_0 z_{tk}z_{t j} - \sum_{k=1}^d (\tilde{\bm{\alpha}}_{ik} - \bm{\alpha}_{ik})  z_{tk}z_{t j}.
    \end{aligned}
\end{equation*}
Define 
$$
\bm{\omega}_{kl}^{(t,i)}:= 
\sum_{j = 1}^d \left(\mathcal{F}_{S_i}(\bm{\Sigma}^{(0)\top}\bm{\Sigma}^{(0)}) \bm{\Sigma}^{(0)} \right)_{jl}z_{tk}z_{tj}
$$
then
$$
\sum_{j=1}^d \left(\mathcal{F}_{S_i}(\bm{\Sigma}^{(0)\top}\bm{\Sigma}^{(0)}) \bm{\Sigma}^{(0)} \right)_{jl}\sum_{k=1}^d (\tilde{\bm{\alpha}}_{ik} - \bm{\alpha}_{ik})  z_{tk}z_{tj} = \sum_{k=1}^d \bm{\omega}_{kl}^{(t,i)}(\tilde{\bm{\alpha}}_{ik} - \bm{\alpha}_{ik}).  
$$
And there exists a positive constant $C_q >0$ for $q \geq 4$ such that

$$
\|\bm{\omega}_{kl}^{(t,i)}\|_{q/2} \leq \|z_{tk}\|_{q} \sum_{j=1}^d \left( \left|\left(\mathcal{F}_{S_i}(\bm{\Sigma}^{(0)\top}\bm{\Sigma}^{(0)}) \bm{\Sigma}^{(0)} \right)_{jl}\right| \cdot\|z_{tj}\|_q  \right) \leq C_q
$$
For $t = 2, \ldots, n$, define a $\mathbb{R}^{d^2}$ vector
\begin{align*}
\bm{\psi}_{t+1} = \begin{bmatrix}
   \mathcal{F}_{S_1}(\bm{\Sigma}^{(0)\top}\bm{\Sigma}^{(0)})\bm{\Sigma}^{(0)}(\mathbb{E}_0 z_{t+1,1}z_t - \sum_{k=1}^d \bm{\alpha}_{1k} \mathbb{E}_0 z_{tk}z_{t} )\\
    \vdots\\
   \mathcal{F}_{S_d}(\bm{\Sigma}^{(0)\top}\bm{\Sigma}^{(0)})\bm{\Sigma}^{(0)} (\mathbb{E}_0 z_{t+1,d}z_t - \sum_{k=1}^d \bm{\alpha}_{dk} \mathbb{E}_0 z_{tk}z_{t} )
\end{bmatrix}
\end{align*}
which satisfies $\mathbb{E}\bm{\psi}_{t+1} =0$. Since $\bm{\psi}_{t+1}$ is a function of $X_{t+1}$ with its causal form \eqref{eq: causal}, it is also a function of $\mathcal{F}_{-\infty}^{t+1}=(\bm{\varepsilon}_{t+1},  \bm{\varepsilon}_t, \ldots )$. Denote $\psi_{t+1,i} = g_i(\ldots, \bm{\varepsilon}_{t+1},  \bm{\varepsilon}_t)$. Then, we have for each $i,l = 1, \ldots, d$, $q \geq 4$,

\begin{equation}
\label{eq: E_0 zt+1zt}
    \begin{aligned}
      & \quad  \left \| \sum_{j = 1}^d \left( \mathcal{F}_{S_i}(\bm{\Sigma}^{(0)\top}\bm{\Sigma}^{(0)}) \bm{\Sigma}^{(0)} \right)_{jl}\mathbb{E}_0 z_{t+1,i}z_{tj}\right\|_{q/2} \\
      &\leq \sum_{j = 1}^d \left|\left( \mathcal{F}_{S_i}(\bm{\Sigma}^{(0)\top}\bm{\Sigma}^{(0)}) \bm{\Sigma}^{(0)} \right)_{jl} \right|  \cdot \| \mathbb{E}_0 z_{t+1,i}z_{tj}\|_{q/2}\\
     &   \leq C_q \sum_{j=1}^d \left|\left( \mathcal{F}_{S_i}(\bm{\Sigma}^{(0)\top}\bm{\Sigma}^{(0)}) \bm{\Sigma}^{(0)} \right)_{jl} \right| \cdot 2 \kappa_q (\delta_{q,t+1,i}+ \delta_{q, t, j})  \\
       & \leq C_q \max_{i =1,\ldots, d }|\mathcal{F}_{S_i}(\bm{\Sigma}^{(0)\top}\bm{\Sigma}^{(0)}) \bm{\Sigma}^{(0)}|_1 \\
     &\leq C_q \max_{i =1,\ldots, d }|\mathcal{F}_{S_i}(\bm{\Sigma}^{(0)\top}\bm{\Sigma}^{(0)}) |_1\cdot |\bm{\Sigma}^{(0)}|_1 \leq C
    \end{aligned}
\end{equation}
and
\begin{equation}
\label{eq: alpha E_0 ztzt}
    \begin{aligned}
  \quad &  \left \| \sum_{j = 1}^d \left( \mathcal{F}_{S_i}(\bm{\Sigma}^{(0)\top}\bm{\Sigma}^{(0)}) \bm{\Sigma}^{(0)} \right)_{jl}\sum_{k=1}^d \alpha_{ik}\mathbb{E}_0 z_{tk}z_{tj}\right\|_{q/2} \\
  \leq & \sum_{j = 1}^d  \left|\left( \mathcal{F}_{S_i}(\bm{\Sigma}^{(0)\top}\bm{\Sigma}^{(0)}) \bm{\Sigma}^{(0)} \right)_{jl} \right|\cdot \left\| \sum_{k=1}^d \bm{\alpha}_{ik} \mathbb{E}_0 z_{tk}z_{tj}\right\|_{q/2}\\
        \leq &C_q \sum_{j = 1}^d\left|\left( \mathcal{F}_{S_i}(\bm{\Sigma}^{(0)\top}\bm{\Sigma}^{(0)}) \bm{\Sigma}^{(0)} \right)_{jl} \right|\cdot 2 \kappa_q \sum_{k=1}^d \bm{\alpha}_{ik}\left(\delta_{q,t,k}+ \delta_{q, t, j}\right)\\
        \leq &C_q \rho_A^q\max_{i =1,\ldots, d }|\mathcal{F}_{S_i}(\bm{\Sigma}^{(0)\top}\bm{\Sigma}^{(0)})|_1  \cdot |\bm{\Sigma}^{(0)}|_1\leq C
    \end{aligned}
\end{equation}
for a positive constant $C>0$ since $(z_{t+1}z_{t})$ is also a stationary process of the form \eqref{eq: causal} and $\|z_{t+1}z_t - z_{t+1, \{0\}}z_{t,\{0\}}\|_{q/2} \leq 2 \kappa_q (\delta_{t,q}+\delta_{t+1,q})$ by H\"older inequality. Similarly, for $q \geq 4$, each $i, l = 1,\ldots, d$, we have
\begin{equation*}
    \begin{aligned}
        &\|\bm{\psi}_{t+1, (i-1)d +l} - \bm{\psi}_{t+1, \{0\}, (i-1)d+l}\|_{q/2} \\
        \leq& \left\| \sum_{j = 1}^d \left(\mathcal{F}_{S_i}(\bm{\Sigma}^{(0)\top}\bm{\Sigma}^{(0)}) \bm{\Sigma}^{(0)}\right)_{jl}(\mathbb{E}_0 z_{t+1,i} z_{tj} - \mathbb{E}_0z_{t+1,\{0\},i} z_{t,\{0\},j})\right\|_{q/2}\\
      &+\left\|\sum_{j = 1}^d  \left(\mathcal{F}_{S_i}(\bm{\Sigma}^{(0)\top}\bm{\Sigma}^{(0)}) \bm{\Sigma}^{(0)}\right)_{jl}\left (\sum_{k=1}^d \alpha_{ik}\mathbb{E}_0 z_{tk}z_{tj} - \sum_{k=1}^d \alpha_{ik}\mathbb{E}_0 z_{t,\{0\},k}z_{t,\{0\},j}\right)\right\|_{q/2} \\
         \leq & C_q (|\mathcal{A}|_1 + 2)\max_{i =1,\ldots, d}|\mathcal{F}_{S_i}(\bm{\Sigma}^{(0)\top}\bm{\Sigma}^{(0)})|_1 |\bm{\Sigma}^{(0)}|_1 \sup_{i \in [d],t \in \mathbb{Z}}\|z_{t, i} - z_{t, \{0\}, i}\|_q .
    \end{aligned}
\end{equation*}
For $z_{t}$ satisfies GMC$(q)$, $\bm{\psi}_{t}$ satisfies GMC$(q/2)$. Now, we figure out covariance matrix between $\sum_{t=1}^{n-1}\bm{\psi}_{t+1}$, 
\begin{equation*}
    \begin{aligned}
    \cov\left(\sum_{t=1}^{n-1}\bm{\psi}_{t+1, (i_1-1)d+l_1}, \sum_{t=1}^{n-1}\bm{\psi}_{t+1, (i_2-1)d+l_2}\right) = n \sum_{i_1=1}^{n-1}\sum_{i_2 =1}^{n-1} \mathbb{E}\bm{\kappa}_{l_1}^{(i_1)}\bm{\kappa}_{l_2}^{(i_2)}    =: n \bm{\Sigma}_{l_1, l_2}^{(i_1, i_2)}
    \end{aligned}
\end{equation*}
where $\bm{\Sigma}_{l_1, l_2}^{(i_1, i_2)} = \cov(Z_{i_1,l_1}, Z_{i_2,l_2})$ is defined in Theorem \ref{thm: GA post selection}. Moreover, $\bm{\psi}_{t, (i-1)d+l} = 0$ if $l \not\in S_i$ for $i = 1, \ldots, d$, which implies $\bm{\psi}_{t}$ only has $\sum_{i =1}^d|S_i| =O(d)$ non-zero elements, and by Lemma \ref{lem: kernel + cov}
\begin{align*}
        &\max_{i_1, i_2, l_1, l_2 \in [d]} \left|\frac{1}{n}\sum_{t_1=1}^{n-1}\sum_{t_2 =1}^{n-1}\bm{\psi}_{t_1, (i_1-1)d+l_1}\bm{\psi}_{t_2, (i_2-1)d+l_2}K\left(\frac{t_1 -t_2}{h_n}\right) - n \bm{\Sigma}_{l_1, l_2}^{(i_1, i_2)}\right|\\
        =&\max_{i_1, i_2 \in [d], l_1 \in S_{i_1},l_2 \in S_{i_2} } \left|\frac{1}{n}\sum_{t_1=1}^{n-1}\sum_{t_2 =1}^{n-1}\bm{\psi}_{t_1, (i_1-1)d+l_1}\bm{\psi}_{t_2, (i_2-1)d+l_2}K\left(\frac{t_1 -t_2}{h_n}\right) - n \bm{\Sigma}_{l_1, l_2}^{(i_1, i_2)}\right|\\
        =& O_{\mathbb{P}} \left(h_n^{-1}+h_n\sqrt{\frac{\log d}{n}}\right)
\end{align*}
Since 
\begin{align*}
        &\sum_{j = 1}^d \left(\mathcal{F}_{S_i}(\bm{\Sigma}^{(0)\top}\bm{\Sigma}^{(0)})\bm{\Sigma}^{(0)}\right)_{jl} \hat{\bm{\Theta}}_{t,ij} \\
        = &\sum_{j=1}^d \left(\mathcal{F}_{S_i}(\bm{\Sigma}^{(0)\top}\bm{\Sigma}^{(0)})\bm{\Sigma}^{(0)}\right)_{jl} (\mathbb{E}_0 z_{t+1,i}z_{tj} -  \sum_{k=1}^d \alpha_{ik}\mathbb{E}_0z_{tk}z_{tj}) -\sum_{k \in S_i}\bm{\omega}_{kl}^{(t,i)}(\tilde{\bm{\alpha}}_{ik} - \bm{\alpha}_{ik})\\
        =& \bm{\psi}_{t+1, (i-1)d+l}-\sum_{k \in S_i}\bm{\omega}_{kl}^{(t,i)}(\tilde{\bm{\alpha}}_{ik} - \bm{\alpha}_{ij}),
\end{align*}
we have
\begin{align*}
        &\Bigg|\frac{1}{n}\sum_{t_1=1}^{n-1} \sum_{t_2=1}^{n-1}K\left(\frac{t_1-t_2}{h_n}\right) \left(\sum_{j=1}^d \left(\mathcal{F}_{S_{i_1}}(\bm{\Sigma}^{(0)\top}\bm{\Sigma}^{(0)} )\bm{\Sigma}^{(0)}\right)_{j l_1} \hat{\bm{\Theta}}_{t,i_1 j}\right) \\
       &\qquad \cdot\left(\sum_{j=1}^d \left(\mathcal{F}_{S_{i_2}}(\bm{\Sigma}^{(0)\top}\bm{\Sigma}^{(0)})\bm{\Sigma}{(0)}\right)_{jl_2} \hat{\bm{\Theta}}_{t,i_2 j}\right)\\
      & \quad - \frac{1}{n}\sum_{t_1=1}^{n-1} \sum_{t_2=1}^{n-1}K\left(\frac{t_1-t_2}{h_n}\right) \bm{\psi}_{t_1+1, (i_1-1)d+l_1} \bm{\psi}_{t_2+1, (i_2-1)d+l_2}\Bigg|\\
       \leq & \left|\frac{1}{n}\sum_{k \in S_{i_1}}(\tilde{\bm{\alpha}}_{i_1 k} - \bm{\alpha}_{i_1 k})\sum_{t_1=1}^{n-1} \sum_{t_2=1}^{n-1}K\left(\frac{t_1-t_2}{h_n}\right) \bm{\psi}_{t_2+1, (i_2-1)d+l_2} \bm{\omega}_{l_1 k}^{(t_1,i_1)}\right|\\
       &+\left|\frac{1}{n}\sum_{k \in S_{i_2}}(\tilde{\bm{\alpha}}_{i_2 k} - \bm{\alpha}_{i_2 k})\sum_{t_1=1}^{n-1} \sum_{t_2=1}^{n-1}K\left(\frac{t_1-t_2}{h_n}\right) \bm{\psi}_{t_1+1, (i_1-1)d+l_1} \bm{\omega}_{l_2 k}^{(t_2,i_2)}\right|\\
       &+\left|\frac{1}{n}\sum_{k_1 \in S_{i_1}}\sum_{k_2 \in S_{i_2}}(\tilde{\bm{\alpha}}_{i_1 k_1} - \bm{\alpha}_{i_1 k_1})(\tilde{\bm{\alpha}}_{i_2 k_2} - \bm{\alpha}_{i_2 k_2})\sum_{t_1=1}^{n-1} \sum_{t_2=1}^{n-1}K\left(\frac{t_1-t_2}{h_n}\right)\bm{\omega}_{l_1 k_1}^{(t_1,i_1)} \bm{\omega}_{l_2 k_2}^{(t_2,i_2)}\right|\\
       \leq & |S_{i_1}| \max_{k \in S_{i_1}} |\tilde{\bm{\alpha}}_{i_1 k} - \bm{\alpha}_{i_1 k}| \max_{k \in S_{i_1}}\left|\frac{1}{n}K\left(\frac{t_1 -t_2}{h_n}\right)  \bm{\psi}_{t_2+1, (i_2-1)d+l_2} \bm{\omega}_{l_1 k}^{(t_1,i_1)}\right| \\
       &+|S_{i_2}| \max_{k \in S_{i_2}} |\tilde{\bm{\alpha}}_{i_2 k} - \bm{\alpha}_{i_2 k}| \max_{k \in S_{i_2}}\left|\frac{1}{n}K\left(\frac{t_1 -t_2}{h_n}\right)  \bm{\psi}_{t_1+1, (i_1-1)d+l_1} \bm{\omega}_{l_2 k}^{(t_2,i_2)}\right|\\
       &+|S_{i_1}||S_{i_2}|\max_{k \in S_{i_1}} |\tilde{\bm{\alpha}}_{i_1 k} - \bm{\alpha}_{i_1 k}|\max_{k \in S_{i_2}} |\tilde{\bm{\alpha}}_{i_2 k} - \bm{\alpha}_{i_2 k}| \\
       &\qquad \cdot\max_{k_1 \in S_{i_1},k_2 \in S_{i_1}} \left|\frac{1}{n}\sum_{t_1 =1}^{n-1}\sum_{t_2 =1}^{n-1}K\left(\frac{t_1-t_2}{h_n}\right)\bm{\omega}_{l_1 k_1}^{(t_1,i_1)} \bm{\omega}_{l_2 k_2}^{(t_2,i_2)} \right|.
\end{align*}
Since the matrix $\left(K\left(\frac{t_1-t_2}{h_n}\right)\right)_{t_1,t_2 =1}^{n-1}$ is a Toeplitz matrix, for any fixed $i_1, i_2, l_1, l_2,k$ and $q \geq 4$, we have
\begin{equation*}
    \begin{aligned}
        &\left\|\left|\frac{1}{n}\sum_{t_1 =1}^{n-1}\sum_{t_2 =1}^{n-1}K\left(\frac{t_1-t_2}{h_n}\right)\bm{\psi}_{t_2+1, (i_2-1)d+l_2}\bm{\omega}_{l_1 k}^{(t_1,i_1)}\right|\right\|_{q/4}\\
       & \leq \frac{C h_n}{n}\sqrt{\sum_{t =1}^{n-1}\left\|\bm{\psi}_{t+1, (i_2-1)d+l_2}\right\|_{q/2}^2}\sqrt{\sum_{t =1}^{n-1}\left\|\bm{\omega}_{l_1 k}^{(t,i_1)}\right\|_{q/2}^2}\leq C_1 h_n
    \end{aligned}
\end{equation*}
and 
\begin{equation*}
    \begin{aligned}
      &  \left\|\left|\frac{1}{n}\sum_{t_1 =1}^{n-1}\sum_{t_2 =1}^{n-1}K\left(\frac{t_1-t_2}{h_n}\right)\bm{\omega}_{l_1 k}^{(t_1,i_1)}\bm{\omega}_{l_2 k}^{(t_2,i_2)}\right|\right\|_{q/4}\\
      &  \leq \frac{C h_n}{n}\sqrt{\sum_{t =1}^{n-1}\left\|\bm{\omega}_{l_1 k}^{(t,i_1)}\right\|_{q/2}^2}\sqrt{\sum_{t =1}^{n-1}\left\|\bm{\omega}_{l_2 k}^{(t,i_2)}\right\|_{q/2}^2}\leq C_1 h_n
    \end{aligned}
\end{equation*}
for a positive constant $C_1$, which implies 
\begin{equation*}
   \max_{i_1,i_2 = 1,\ldots, d ,k \in S_{i_1}, l_2 \in S_{i_2}} \left|\frac{1}{n}\sum_{t_1 =1}^{n-1}\sum_{t_2 =1}^{n-1}K\left(\frac{t_1-t_2}{h_n}\right)\bm{\psi}_{t_2+1, (i_2-1)d+l_2}\bm{\omega}_{l_1 k}^{(t_1,i_1)}\right|=O_\mathbb{P}\left(h_n \sqrt{\log d}\right),
\end{equation*}
\begin{equation*}
   \max_{i_1,i_2 = 1,\ldots, d ,k \in S_{i_2}, l_1 \in S_{i_1}} \left|\frac{1}{n}\sum_{t_1 =1}^{n-1}\sum_{t_2 =1}^{n-1}K\left(\frac{t_1-t_2}{h_n}\right)\bm{\psi}_{t_1+1, (i_1-1)d+l_1}\bm{\omega}_{l_2 k}^{(t_2,i_2)}\right|=O_\mathbb{P}\left(h_n \sqrt{\log d}\right),
\end{equation*}
and
\begin{equation*}
   \max_{i_1,i_2 = 1,\ldots, d ,k \in S_{i_2}, l_1 \in S_{i_1}} \left|\frac{1}{n}\sum_{t_1 =1}^{n-1}\sum_{t_2 =1}^{n-1}K\left(\frac{t_1-t_2}{h_n}\right)\bm{\omega}_{l_1 k}^{(t_1,i_1)}\bm{\omega}_{l_2 k}^{(t_2,i_2)}\right|=O_\mathbb{P}\left(h_n \sqrt{\log d}\right),
\end{equation*}
Therefore, combining with two inequalities, we have
\begin{equation}
    \begin{aligned}
        &\Bigg|\frac{1}{n}\sum_{t_1=1}^{n-1} \sum_{t_2=1}^{n-1}K\left(\frac{t_1-t_2}{h_n}\right) \left(\sum_{j=1}^d \left(\mathcal{F}_{S_{i_1}}(\bm{\Sigma}^{(0)\top}\bm{\Sigma}^{(0)})\bm{\Sigma}^{(0)}\right)_{jl_1} \hat{\bm{\Theta}}_{t,i_1 j}\right) \\
       &\qquad \cdot\left(\sum_{j=1}^d \left(\mathcal{F}_{S_{i_2}}(\bm{\Sigma}^{(0)\top}\bm{\Sigma}^{(0)})\bm{\Sigma}^{(0)}\right)_{jl_2} \hat{\bm{\Theta}}_{t,i_2 j}\right)\\
      &\quad - \frac{1}{n}\sum_{t_1=1}^{n-1} \sum_{t_2=1}^{n-1}K\left(\frac{t_1-t_2}{h_n}\right) \bm{\psi}_{t_1+1, (i_1-1)d+l_1} \bm{\psi}_{t_2+1, (i_2-1)d+l_2}\Bigg|=O_\mathbb{P}\left(h_n \sqrt{\frac{\log^4 (d \vee n) } {n}}\right) 
      \end{aligned}
      \end{equation}
     Now, we figure out $\max_{i =1,\ldots, d}|\mathcal{F}_{S_i}(\hat{\bm{\Sigma}}^{(0)\top}\hat{\bm{\Sigma}}^{(0)}) - \mathcal{F}_{S_i}({\bm{\Sigma}}^{(0)\top}{\bm{\Sigma}^{(0)}})|_1$.
      By \cite{horn2012matrix}, we have 
      \begin{equation*}
          \begin{aligned}
              &|\mathcal{F}_{S_i}(\hat{\bm{\Sigma}}^{(0)\top}\hat{\bm{\Sigma}}^{(0)}) \hat{\bm{\Sigma}}^{(0)} - \mathcal{F}_{S_i}({\bm{\Sigma}}^{(0)\top}{\bm{\Sigma}^{(0)}})\bm{\Sigma}^{(0)}|_1 \\
              =& |(\hat{\bm{\Sigma}}^{(0)\top}\hat{\bm{\Sigma}}^{(0)})_{S_i}^{-1} \hat{\bm{\Sigma}}^{(0)}-({\bm{\Sigma}}^{(0)\top}{\bm{\Sigma}^{(0)}})_{S_i}^{-1} \bm{\Sigma}^{(0)}|_1\\
              \leq &|(\hat{\bm{\Sigma}}^{(0)\top}\hat{\bm{\Sigma}}^{(0)})_{S_i}^{-1} -({\bm{\Sigma}}^{(0)\top}{\bm{\Sigma}^{(0)}})_{S_i}^{-1} |_1 |\hat{\bm{\Sigma}}^{(0)}|_1 + |({\bm{\Sigma}}^{(0)\top}{\bm{\Sigma}^{(0)}})_{S_i}^{-1}|_1 |\hat{\bm{\Sigma}}^{(0)} - \bm{\Sigma}^{(0)}|_1\\
              \leq & C |(\hat{\bm{\Sigma}}^{(0)\top}\hat{\bm{\Sigma}}^{(0)})_{S_i}^{-1} -({\bm{\Sigma}}^{(0)\top}{\bm{\Sigma}^{(0)}})_{S_i}^{-1}|_2 |\hat{\bm{\Sigma}}^{(0)}|_{\max} + C|({\bm{\Sigma}}^{(0)\top}{\bm{\Sigma}^{(0)}})_{S_i}^{-1}|_2|\hat{\bm{\Sigma}}^{(0)} - \bm{\Sigma}^{(0)}|_{\max} \\
              \leq & C_1|\hat{\bm{\Sigma}}^{(0)\top}\hat{\bm{\Sigma}}^{(0)} -{\bm{\Sigma}}^{(0)\top}{\bm{\Sigma}^{(0)}}|_{\max}|\hat{\bm{\Sigma}}^{(0)}|_{\max}+ C_1|{\bm{\Sigma}}^{(0)}|_{\max}^2 |\hat{\bm{\Sigma}}^{(0)} - \bm{\Sigma}^{(0)}|_{\max} \\
              \leq  &C_1 |\hat{\bm{\Sigma}}^{(0)\top}(\hat{\bm{\Sigma}}^{(0)} - \bm{\Sigma}^{(0)})|_{\max} |\hat{\bm{\Sigma}}^{(0)}|_{\max}+ C_1|(\hat{\bm{\Sigma}}^{(0)\top} - \bm{\Sigma}^{(0)\top})\bm{\Sigma}^{(0)}|_{\max}|\hat{\bm{\Sigma}}^{(0)}|_{\max} \\
              &\qquad+C_1|{\bm{\Sigma}}^{(0)}|_{\max}^2 |\hat{\bm{\Sigma}}^{(0)} - \bm{\Sigma}^{(0)}|_{\max}  \\
              \leq & C_1|\hat{\bm{\Sigma}}^{(0)}|_{\max} ^2|\hat{\bm{\Sigma}}^{(0)} - \bm{\Sigma}^{(0)}|_{\max} + C_1 |\hat{\bm{\Sigma}}^{(0)} - \bm{\Sigma}^{(0)}|_{\max} |\bm{\Sigma}^{(0)}|_{\max}  ^2\leq \lambda_0,
          \end{aligned}
      \end{equation*}
which implies that $$\max_{i = 1, \ldots, d}|\mathcal{F}_{S_i}(\hat{\bm{\Sigma}}^{(0)\top}\hat{\bm{\Sigma}}^{(0)}) \hat{\bm{\Sigma}}^{(0)}- \mathcal{F}_{S_i}({\bm{\Sigma}}^{(0)\top}{\bm{\Sigma}^{(0)}})\bm{\Sigma}^{(0)}|_1 = O_\mathbb{P}(\log n \sqrt{\log d/n}).$$ 
Therefore, define the term $$\xi^{(i)}_{jl} = \left(\mathcal{F}_{S_i}(\hat{\bm{\Sigma}}^{(0)\top}\hat{\bm{\Sigma}}^{(0)})\hat{\bm{\Sigma}}^{(0)}\right)_{jl} - \left(\mathcal{F}_{S_i}({\bm{\Sigma}}^{(0)\top}{\bm{\Sigma}^{(0)}})\bm{\Sigma}^{(0)}\right)_{jl},$$
we have
$$
\sum_{j = 1}^d \left(\mathcal{F}_{S_i}(\hat{\bm{\Sigma}}^{(0)\top}\hat{\bm{\Sigma}}^{(0)} )\hat{\bm{\Sigma}}^{(0)}\right)_{jl} \hat{\bm{\Theta}}_{t, ij} = \sum_{j \in S_{i}} \left(\mathcal{F}_{S_i}({\bm{\Sigma}}^{(0)\top}{\bm{\Sigma}^{(0)}})\bm{\Sigma}^{(0)}\right)_{jl}\hat{\bm{\Theta}}_{t, ij}  + \sum_{j \in S_i}\xi_{jl}^{(i)}\hat{\bm{\Theta}}_{t, ij} 
$$
and
\begin{equation*}
    \begin{aligned}
      &\Bigg|\frac{1}{n}\sum_{t_1=1}^{n-1} \sum_{t_2=1}^{n-1}K\left(\frac{t_1-t_2}{h_n}\right) \left(\sum_{j=1}^d \left(\mathcal{F}_{S_{i_1}}(\hat{\bm{\Sigma}}^{(0)\top}\hat{\bm{\Sigma}}^{(0)})\hat{\bm{\Sigma}}^{(0)}\right)_{jl_1} \hat{\bm{\Theta}}_{t_1,i_1 j}\right) \\
     & \qquad \cdot\left(\sum_{j=1}^d \left(\mathcal{F}_{S_{i_2}}(\hat{\bm{\Sigma}}^{(0)\top}\hat{\bm{\Sigma}}^{(0)})\hat{\bm{\Sigma}}^{(0)}\right)_{jl_2} \hat{\bm{\Theta}}_{t_2,i_2 j}\right)\\
      &\quad-\frac{1}{n}\sum_{t_1=1}^{n-1} \sum_{t_2=1}^{n-1}K\left(\frac{t_1-t_2}{h_n}\right) \left(\sum_{j=1}^d \left(\mathcal{F}_{S_{i_1}}(\bm{\Sigma}^{(0)\top}{\bm{\Sigma}}^{(0)}){\bm{\Sigma}}^{(0)}\right)_{jl_1} \hat{\bm{\Theta}}_{t_1,i_1 j}\right) \\
      &\qquad \cdot\left(\sum_{j=1}^d \left(\mathcal{F}_{S_{i_2}}({\bm{\Sigma}}^{(0)\top}{\bm{\Sigma}}^{(0)}){\bm{\Sigma}}^{(0)}\right)_{jl_2} \hat{\bm{\Theta}}_{t_2,i_2 j_2}\right)\Bigg|\\
      \leq &\left|\sum_{j_1 \in S_{i_1}}\sum_{j_2 \in S_{i_2}}\xi^{(i_1)}_{j_1 l_1} \left(\mathcal{F}_{S_{i_2}}(\bm{\Sigma}^{(0)\top}\bm{\Sigma}^{(0)})\bm{\Sigma}^{(0)}\right)_{j_2 l_2} \frac{1}{n}\sum_{t_1=1}^{n-1} \sum_{t_2=1}^{n-1}K\left(\frac{t_1-t_2}{h_n}\right)\hat{\bm{\Theta}}_{t_1,i_1 j_1}\hat{\bm{\Theta}}_{t_2,i_1 j_2}\right| \\
      &+\left|\sum_{j_1 \in S_{i_1}}\sum_{j_2 \in S_{i_2}}\xi^{(i_2)}_{j_2 l_2} \left(\mathcal{F}_{S_{i_1}}(\bm{\Sigma}^{(0)\top}\bm{\Sigma}^{(0)})\bm{\Sigma}^{(0)}\right)_{j_1 l_1} \frac{1}{n}\sum_{t_1=1}^{n-1} \sum_{t_2=1}^{n-1}K\left(\frac{t_1-t_2}{h_n}\right)\hat{\bm{\Theta}}_{t_1,i_1 j_1}\hat{\bm{\Theta}}_{t_2,i_1 j_2}\right| \\
      &+\left|\sum_{j_1 \in S_{i_1}}\sum_{j_2 \in S_{i_2}}\xi^{(i_1)}_{j_1 l_1} \xi^{(i_2)}_{j_2 l_2}  \frac{1}{n}\sum_{t_1=1}^{n-1} \sum_{t_2=1}^{n-1}K\left(\frac{t_1-t_2}{h_n}\right)\hat{\bm{\Theta}}_{t_1,i_1 j_1}\hat{\bm{\Theta}}_{t_2,i_1 j_2}\right|.
    \end{aligned}
\end{equation*}
By Cauchy-Schwartz inequality, we have
\begin{equation*}
    \begin{aligned}
        \sum_{t=1}^{n-1} \hat{\bm{\Theta}}_{t, ij}^2 &= \sum_{t=1}^{n-1} \left(\mathbb{E}_0 z_{t+1,i}z_{tj} - \sum_{k=1}^d \bm{\alpha}_{ik} \mathbb{E}_0 z_{tk}z_{tj} - \sum_{k=1}^d (\tilde{\bm{\alpha}}_{ik} - \bm{\alpha}_{ik})z_{tk}z_{tj}\right)^2\\
        & \leq 2\sum_{t=1}^{n-1} \left(\mathbb{E}_0 z_{t+1,i}z_{tj} - \sum_{k=1}^d \bm{\alpha}_{ik} \mathbb{E}_0 z_{tk}z_{tj}\right)^2 + 2\sum_{t=1}^{n-1} \left(\sum_{k=1}^d (\tilde{\bm{\alpha}}_{ik} - \bm{\alpha}_{ik})z_{tk}z_{tj}\right)^2\\
        & \leq 2\sum_{t=1}^{n-1} \left(\mathbb{E}_0 z_{t+1,i}z_{tj} - \sum_{k=1}^d \bm{\alpha}_{ik} \mathbb{E}_0 z_{tk}z_{tj}\right)^2 + 2 \sum_{t=1}^{n-1} \sum_{k=1}^d (\tilde{\bm{\alpha}}_{ik} - \bm{\alpha}_{ik})^2 \sum_{k=1}^d z_{tk}^2 z_{tj}^2\\
        & \leq 4 \sum_{t=1}^{n-1}\left(\mathbb{E}_0 z_{t+1,i}z_{tj} \right)^2 + 4 \sum_{t=1}^{n-1} \left(\sum_{k \in S_i}\bm{\alpha}_{ik} \mathbb{E}_0 z_{tk}z_{tj}\right)^2 + 2 \sum_{t=1}^{n-1} \sum_{k \in S_i} (\tilde{\bm{\alpha}}_{ik} - \bm{\alpha}_{ik})^2 \sum_{k\in S_i} z_{tk}^2 z_{tj}^2
    \end{aligned}
\end{equation*}
For upper bounds \eqref{eq: E_0 zt+1zt} and \eqref{eq: alpha E_0 ztzt} of $\mathbb{E}_0 z_{t+1,i}z_{tj}$ and $\sum_{k=1}^d \bm{\alpha}_{ik} \mathbb{E}_0 z_{tk}z_{tj}$, we have
 $\left(\mathbb{E}_0 z_{t+1,i}z_{tj} \right)^2 $ and $\left(\sum_{k=1}^d \bm{\alpha}_{ik} \mathbb{E}_0 z_{tk}z_{tj}\right)^2$ satisfy GMC$(q/4)$ and since
\begin{equation*}
    \left\|\sum_{t=1}^{n-1} z_{tj}^2\sum_{k \in S_i} z_{tk}^2 \right\|_{q/4} \leq \sum_{t=1}^{n-1}\sum_{k \in S_i} \|z_{tj}\|_q^2 \|z_{tk}\|_q^2 \leq C n,
\end{equation*}
for some positive constant $C$, $z_{tj}^2\sum_{k \in S_i} z_{tk}^2$ also satisfies GMC$(q/4)$. It implies that from \eqref{eq: convergence rate for post selection}, 
\begin{equation*}
    \begin{aligned}
        \max_{i = 1, \ldots, d, j \in S_i } \sum_{t=1}^{n-1} \hat{\bm{\Theta}}_{t, ij}^2 \leq & 4\max_{i = 1, \ldots, d, j \in S_i } \sum_{t=1}^{n-1}\left(\mathbb{E}_0 z_{t+1,i}z_{tj} \right)^2 + 4 \max_{i = 1, \ldots, d, j \in S_i }\sum_{t=1}^{n-1} \left(\sum_{k \in S_i}\bm{\alpha}_{ik} \mathbb{E}_0 z_{tk}z_{tj}\right)^2 \\
        &+2 \sum_{t=1}^{n-1} \max_{i = 1, \ldots, d, j \in S_i }\sum_{k \in S_i} (\tilde{\bm{\alpha}}_{ik} - \bm{\alpha}_{ik})^2 \max_{i = 1, \ldots, d, j \in S_i }\sum_{k\in S_i} z_{tk}^2 z_{tj}^2 =O_{\mathbb{P}}\left( n\sqrt{ \log d}\right).
    \end{aligned}
\end{equation*}
Therefore, we have 
\begin{equation*}
    \begin{aligned}
        &\max_{i_1 ,i_2 =1, \ldots, d, j_1 \in S_{i_1}, j_2 \in S_{i_2}} \left|\frac{1}{n}\sum_{t_1 =1}^{n-1}\sum_{t_2 =1}^{n-1}K\left(\frac{t_1 -t_2}{h_n}\right) \hat{\bm{\Theta}}_{t_1, i_1j_1}\hat{\bm{\Theta}}_{t_2, i_2j_2}\right|\\
        \leq & \frac{Ch_n}{n}\sqrt{\sum_{t=1}^{n-1} \hat{\bm{\Theta}}_{t, i_1j_1}^2}\sqrt{\sum_{t=1}^{n-1} \hat{\bm{\Theta}}_{t, i_2j_2}^2} = O_{\mathbb{P}}\left(h_n \sqrt{\log d}\right),
    \end{aligned}
\end{equation*}
and
\begin{equation*}
    \begin{aligned}
        &\max_{i_1, i_2 \in [d], l_1 , l_2 \in [d]}\left|\sum_{j_1 \in S_{i_1}}\sum_{j_2 \in S_{i_2}}\xi^{(i_2)}_{j_2 l_2} \left(\mathcal{F}_{S_{i_1}}(\bm{\Sigma}^{(0)\top}\bm{\Sigma}^{(0)})\bm{\Sigma}^{(0)}\right)_{j_1 l_1} \frac{1}{n}\sum_{t_1=1}^{n-1} \sum_{t_2=1}^{n-1}K\left(\frac{t_1-t_2}{h_n}\right)\hat{\bm{\Theta}}_{t_1,i_1 j_1}\hat{\bm{\Theta}}_{t_2,i_1 j_2}\right|\\
        \leq &C \max_{i_2 \in [d], j_2 \in [d], l_2 \in S_{i_2}} |\xi_{j_2,l_2}^{(l_2)}|\max_{i_1 ,i_2 =1, \ldots, d, j_1 \in S_{i_1}, j_2 \in S_{i_2}} \left|\frac{1}{n}\sum_{t_1 =1}^{n-1}\sum_{t_2 =1}^{n-1}K\left(\frac{t_1 -t_2}{h_n}\right) \hat{\bm{\Theta}}_{t_1, i_1j_1}\hat{\bm{\Theta}}_{t_2, i_2j_2}\right|\\
        =&O_{\mathbb{P}}\left(h_n\sqrt{\frac{\log ^4 (d \vee n)}{n}}\right).
    \end{aligned}
\end{equation*}
Similarly, 
\begin{equation*}
    \begin{aligned}
  & \max_{i_1, i_2 \in [d], l_1 , l_2 \in [d]} \left|\sum_{j_1 \in S_{i_1}}\sum_{j_2 \in S_{i_2}}\xi^{(i_1)}_{j_1 l_1} \xi^{(i_2)}_{j_2 l_2}  \frac{1}{n}\sum_{t_1=1}^{n-1} \sum_{t_2=1}^{n-1}K\left(\frac{t_1-t_2}{h_n}\right)\hat{\bm{\Theta}}_{t_1,i_1 j_1}\hat{\bm{\Theta}}_{t_2,i_1 j_2}\right|\\
   \leq &C \max_{i \in [d], j \in [d], l \in S_{i}} |\xi_{j,l}^{(l)}|^2\max_{i_1 ,i_2 =1, \ldots, d, j_1 \in S_{i_1}, j_2 \in S_{i_2}} \left|\frac{1}{n}\sum_{t_1 =1}^{n-1}\sum_{t_2 =1}^{n-1}K\left(\frac{t_1 -t_2}{h_n}\right) \hat{\bm{\Theta}}_{t_1, i_1j_1}\hat{\bm{\Theta}}_{t_2, i_2j_2}\right|\\
   =&O_{\mathbb{P}}\left(\frac{h_n\log ^{7/2} (d \vee n)}{n}\right).
    \end{aligned}
\end{equation*}
Therefore, 
\begin{equation*}
    \begin{aligned}
      &\Bigg|\frac{1}{n}\sum_{t_1=1}^{n-1} \sum_{t_2=1}^{n-1}K\left(\frac{t_1-t_2}{h_n}\right) \left(\sum_{j=1}^d \left(\mathcal{F}_{S_{i_1}}(\hat{\bm{\Sigma}}^{(0)\top}\hat{\bm{\Sigma}}^{(0)})\hat{\bm{\Sigma}}^{(0)}\right)_{jl_1} \hat{\bm{\Theta}}_{t_1,i_1 j}\right) \\
      & \qquad\cdot\left(\sum_{j=1}^d \left(\mathcal{F}_{S_{i_2}}(\hat{\bm{\Sigma}}^{(0)\top}\hat{\bm{\Sigma}}^{(0)})\hat{\bm{\Sigma}}^{(0)}\right)_{jl_2} \hat{\bm{\Theta}}_{t_2,i_2 j}\right)\\
      &\quad-\frac{1}{n}\sum_{t_1=1}^{n-1} \sum_{t_2=1}^{n-1}K\left(\frac{t_1-t_2}{h_n}\right) \left(\sum_{j=1}^d \left(\mathcal{F}_{S_{i_1}}(\bm{\Sigma}^{(0)\top}{\bm{\Sigma}}^{(0)}){\bm{\Sigma}}^{(0)}\right)_{jl_1} \hat{\bm{\Theta}}_{t_1,i_1 j}\right) \\
     &\qquad \cdot\left(\sum_{j=1}^d \left(\mathcal{F}_{S_{i_2}}({\bm{\Sigma}}^{(0)\top}{\bm{\Sigma}}^{(0)}){\bm{\Sigma}}^{(0)}\right)_{jl_2} \hat{\bm{\Theta}}_{t_2,i_2 j_2}\right)\Bigg|=O_{\mathbb{P}}\left(\frac{h_n \log^{2} (d \vee n)}{\sqrt{n}}\right).
      \end{aligned}
      \end{equation*}
 \end{proof}
 \subsection{Proof of Corollary \ref{cor: bootstrap}}
      \begin{proof}
       Algorithm \ref{algo: swb} implies that $\bm{\Delta}^*_{ij} $ has joint normal distribution conditional on the observations $\{X_t\}$. From Lemma B.1 in  \cite{zhang2023debiased}, we have for sufficiently large $n$, 
       \begin{align*}
           &\sup_{x \in \mathbb{R}}\left|\mathbb{P}^*\left(\max_{i,j = 1, \ldots, d} |\mathbf{\Delta}_{ij}^*| \leq x \right) - \mathbb{P}\left(\max_{i = 1, \ldots, d, j\in S_i}|Z_{ij}|\leq x\right)\right|\\
           =&  \sup_{x \in \mathbb{R}}\left|\mathbb{P}^*\left(\max_{i = 1, \ldots, d, j\in S_i} |\mathbf{\Delta}_{ij}^*| \leq x \right) - \mathbb{P}\left(\max_{i = 1, \ldots, d, j\in S_i}|Z_{ij}|\leq x\right)\right|\\
           \leq& C\left(\delta^{1/3}(1+\log^3(n))+\frac{\delta^{1/6}}{1+\log^{1/4}(n)}\right),
       \end{align*}
       where $\delta =   \max_{i_1, i_2, j_1, j_2 = 1,\ldots, d} |\mathbb{E}^*\bm{\Delta}_{i_1,j_1}\bm{\Delta}_{i_2,j_2} - \bm{\Sigma}_{j_1,j_2}^{(i_1,i_2)} |$, which has order $O_\mathbb{P}\left(h_n^{-1} + \frac{h_n \log^2 (d \vee n)}{\sqrt{n}}\right)$
\end{proof}

\section{Some Useful Lemmas}
\subsection{Proof of Lemma \ref{lem: momeng for sum}}
\begin{lemma}
\label{lem: momeng for sum} 
There exists $\rho \in (0,1)$ such that $\|X_{\cdot}\|_2 < \infty$. For $m \geq 0$, define
    \begin{equation*}
        \tilde{X}_{i} = (\tilde{X}_{i1}, \ldots, \tilde{X}_{id})^\top =\mathbb{E}(X_i \mid \varepsilon_{i-m}, \varepsilon_{i-m+1}, \ldots, \varepsilon_{i}).
    \end{equation*}
    Denote $\mathbb{E}_0 (\cdot) = \cdot - \mathbb{E}(\cdot)$. Then the followings hold for any $1\leq j\leq d$.
\begin{enumerate}[label=(\roman*)]
    \item $\|\sum_{i=1}^n \mathbb{E}_0X_{ij}\|_2 \leq  \sqrt{n} \|X_{\cdot}\|_2$.
    \item $\|\sum_{i=1}^n \mathbb{E}_0\tilde{X}_{ij}\|_2 \leq    \sqrt{n} \|X_{\cdot}\|_2$.
    \item $\|\sum_{i=1}^n  \mathbb{E}_0(X_{ij} -\tilde{X}_{ij}) \|_2 \leq  \sqrt{n} \rho^m \|X_{\cdot}\|_2$.
\end{enumerate}
\end{lemma}

\begin{proof}
(i) Note that we have the decomposition $\mathbb{E}_0 X_{ij}=\sum_{k=0}^\infty \mathcal{P}_{i-k}(X_{ij})$, where $(\mathcal{P}_{k}(X_{ij}))_k$ is a martingale difference sequence with respect to the filtration $\{\mathcal{F}_{-\infty}^k
\}$. 
Then we have 
\begin{equation}
\label{eq: rosenthal}
\|\sum_{i=1}^n \mathcal{P}_{i-k}(X_{ij})\|_2^2  =\sum_{i=1}^n\|\mathcal{P}_{i-k}(X_{ij})\|_2^2.
\end{equation}
By Jensen’s inequality, the stationarity of the process $(X_{ij})_i$ and the definition of the dependence measure, it follows that
\begin{equation}
\label{eq: MG diff}
\|\mathcal{P}_{i-k}(X_{ij})\|_2 = \|\mathbb{E}(X_{ij} - X_{ij,\{i-k\}})|\mathcal{F}_{i-k})\|_2 \leq \|X_{ij} - X_{ij,\{i-k\}}\|_2 = \delta_{k,2,j}.
\end{equation}
By \eqref{eq: rosenthal} and \eqref{eq: MG diff}, $\|\sum_{i=1}^n \mathcal{P}_{i-k}(X_{ij})\|_2 \leq \sqrt{n} \delta_{k,2,j}$. 
By Minkowski’s inequality and the fact that $\sum_{k=0}^\infty 
\delta_{k,2,j} \leq \|X_{\cdot j}\|_2 $, we have 
\begin{equation*}
\|\sum_{i=1}^n \mathbb{E}_0X_{ij}\|_2 \leq \sum_{k=0}^\infty \|\sum_{i=1}^n \mathcal{P}_{i-k}(X_{ij})\|_2 \leq \sqrt{n}\sum_{k=0}^\infty 
\delta_{k,2,j} \leq   \sqrt{n} \|X_{\cdot}\|_2.
\end{equation*}

(ii) We can adopt the decomposition $\mathbb{E}_0 \tilde{X}_{ij} = \sum_{k=0}^{m-1} \mathcal{P}^{i-k}(X_{ij})$. Note that $(\mathcal{P}^{k}(\cdot))_k$ is a backward martingale difference sequence with respect to $\{\mathcal{F}_k^\infty\}$. By similar arguments of deriving (i), we can obtain
\begin{equation*}
\|\mathcal{P}^{i-k}(X_{ij})\|_2 = \|\mathbb{E}(X_{ij}-X_{ij,\{i-k\}})|\mathcal{F}^{i-k})\|_2
\leq \|X_{ij} - X_{ij,\{i-k\}}\|_2 = \delta_{k,2,j}
\end{equation*}
and
\begin{equation}
\label{eq: ii}
\|\sum_{i=1}^n \mathbb{E}_0\tilde{X}_{ij}\|_2 \leq \sum_{k=0}^{m-1} \|\sum_{i=1}^n\mathcal{P}^{i-k}(X_{ij})\|_2 
\leq   \sqrt{n} \sum_{k=0}^{m-1}\delta_{k,2,j}\leq \sqrt{n} \|X_{\cdot}\|_2.
\end{equation}

(iii) Since $\mathbb{E}_0 (X_{ij}-\tilde{X}_{ij}) = \sum_{k=m}^{\infty} \mathcal{P}^{i-k}(X_{ij})$, similarly as \eqref{eq: ii}, we can obtain
\begin{equation*}
\|\sum_{i=1}^n \mathbb{E}_0(X_{ij}-\tilde{X}_{ij})\|_2
\leq   \sqrt{n} \sum_{k=m}^{\infty}\delta_{k,2,j}
\end{equation*}
Result (iii) follows in view of $\sum_{k=m}^{\infty}\delta_{k,2,j} \leq \rho^m \|X_{\cdot j}\|_2 \leq \rho^m \|X_{\cdot}\|_2.$
\end{proof}

 \subsection{Proof of Lemma \ref{lem: consistency of sigma}}
\begin{corollary}
\label{thm: mean}
Let conditions A1 to A3 hold. Let $\widehat{\mu}_j$ be the mean estimator of $\mu_j = \mathbb{E}(X_{ij})$ for $\sigma^* \geq \sigma_2$ satisfies
\begin{equation}
    \label{eq: condition for x}
    \mathcal{C}_1\log \left(\frac{1}{x}\right) \left(\mathcal{C}_1(\log n)^2+\mathcal{C}_2 \log n \frac{\|X_.\|_2}{\sigma_2}\right)\leq 4n,
\end{equation}
where $ 0<x\leq 1/e$. Then for $1\leq j\leq d$, we have
    \begin{equation}
    \label{eq: mean}
        \mathbb{P}(|\widehat{\mu}_j - \mu_j|\geq t) \leq 2\exp \left\{-\frac{n^2t^2}{4 C_1(n\|X.\|_2^2+M^2)+2 C_2 M(\log n)^2 nt}\right\}.
    \end{equation}
Let $$
 t =  \sqrt{ \frac{C_1   \|X.\|_2^2 \log (1/x)}{n}}+\frac{(\sqrt{C_1}+C_2)M(\log n)^2  \log (1/x)}{n} , 
$$
which implies that $\mathbb{P}(\widehat{\mu}_j -\mu_j\geq t)\leq e^{-1/4}x$. In particular, letting $x = d^{-c-1}$, if \eqref{eq: condition for x} is satisfied, for some $c>0$, it follows 
\begin{equation}
\label{eq: probability inequality for mean}
    \mathbb{P}\left(|\widehat{\mu}-\mu|_\infty\geq \sqrt{c+1}(\mathcal{C}_1 \sigma^*\log n +\mathcal{C}_2\|X.\|_2)\sqrt{\frac{\log d}{n}}\right)\leq 2e^{-1/4}d^{-c}.
\end{equation}
\end{corollary}
 
\begin{lemma}
\label{lem: consistency of sigma}
    Assume that $\|X_{\cdot}\|_2<\infty$. Let $\lambda \asymp \log n\sqrt{\log d /n}$. Then with probability at least $1- \frac{20 e^{-1/4}}{3}d^{-c}$ for some constant $c > 0$, it holds that
    \begin{equation*}
        |\hat{\bm{\Sigma}}^{(0)} - \bm{\Sigma}^{(0)}|_{\max} \leq \lambda_0 \text{  and  }|\hat{\bm{\Sigma}}^{(1)} - \bm{\Sigma}^{(1)}|_{\max} \leq \lambda_0 
    \end{equation*}
 \end{lemma}

\begin{proof}
Denote $|\mu|_\infty = \max_{1\leq j\leq d} |\mu_j|$. Let $x = d^{-c-2}$. Then
$$
\bm{\Delta}_n(d^{-c-2}) = \sqrt{c+2}\left(\mathcal{C}_1 \sigma^*\log n +\mathcal{C}_2\|X.\|_2\right)\sqrt{\frac{\log d}{n}}
$$
By the Bonferroni inequality and Corollary \ref{thm: mean},
\begin{equation}
    \mathbb{P}(|\hat{\mu}-\mu|_\infty \geq \bm{\Delta}_n(d^{-c-2}))\leq 2e^{-1/4} d^{-c-1}
\end{equation}
    Since 
    $$
\max_{1\leq j,k \leq d}|\hat{\mu}_j \hat{\mu}_k - \mu_j \mu_k| \leq 2 |\mu|_\infty |\hat{\mu} - \mu|_\infty + |\hat{\mu}-\mu|_\infty^2,
    $$
    it follows that
\begin{equation}
\label{eq: mujmuk}
    \mathbb{P}\left(\max_{1\leq j,k \leq p}|\hat{\mu}_j \hat{\mu}_k - \mu_j \mu_k| 
 \geq 2|\mu|_\infty \bm{\Delta}_n(d^{-c-2})+\bm{\Delta}^2_n(d^{-c-2})\right) \leq 2e^{-1/4}d^{-c-1}.
\end{equation}
By the triangle inequality and H\"{o}lder inequality, we can compute the dependence measure of the process $(X_{ij}X_{ik})$, $i \in \mathbb{Z}$, as
\begin{equation}
    \begin{aligned}
        &\|X_{ij}X_{ik} - X_{ij,\{0\}}X_{ik,\{0\}}\|_2 \\
        \leq & \|X_{ij}\|_4 \|X_{ik} - X_{ik,\{0\}}\|_4 +\|X_{ik,\{0\}}\|_4 \|X_{ij} - X_{ij,\{0\}}\|_4\\
        = & \omega_4 (\bm{\Delta}_{i,4,j}+\bm{\Delta}_{i,4,k}),
    \end{aligned}
\end{equation}
which implies 
$$
\sup_{m\leq 0}\rho^{-m} \sum_{i=m}^\infty\|X_{ij}X_{ik} - X_{ij,\{0\}}X_{ik,\{0\}}\|_2 \leq 2 \omega_4 \|X.\|_4.
$$
Thus, 
\begin{equation}
\label{eq: mujk}
    \begin{aligned}
        \mathbb{P}\left(\frac{1}{n}\left|\sum_{i = 1}^n X_{ij}X_{ik} - \mathbb{E}X_{ij}X_{ik}\right|\geq \lambda_0 \right) \leq 2 e^{-1/4}d^{-c}
    \end{aligned}
\end{equation}
with 
$$
\begin{aligned}
    \lambda_0 = & \sqrt{c+2}[\mathcal{C}_1(2 |\mu|_\infty \sigma^*+1) \log n+\mathcal{C}_2 (2|\mu|_\infty \|X.\|_2+\omega_4\|X.\|_4 )]\sqrt{\frac{\log d}{n}} \\
    &+ (c+2)(\mathcal{C}_1  \sigma^* \log n  + \mathcal{C}_2 \|X.\|_2)^2 \frac{\log d}{n}
\end{aligned}
$$
Then, 
$$
\mathbb{P}(|\hat{\bm{\Sigma}}^{(0)} - \bm{\Sigma}^{(0)}|_{\max} \geq \lambda_0) \leq \frac{8e^{-1/4}}{3}d^{-c}
$$
follows from the Bonferroni inequality, \eqref{eq: mujmuk} and \eqref{eq: mujk} in view of $p\geq 3$ and 
$$
|\hat{\bm{\Sigma}}^{(0)} -\bm{\Sigma}^{(0)}|_{\max} \leq \max_{1\leq j,k\leq d}|\hat{\mu}_j\hat{\mu}_k -\mu_j\mu_k|+\max_{1\leq j,k \leq d} |\hat{\mu}_{jk} -\mu_{jk}|.
$$

Similarly, we can obtain $|\hat{\bm{\Sigma}}^{(1)} - \bm{\Sigma}^{(1)}|_{\max}\leq \lambda_0$ with at least $1- 4e^{-1/4}d^{-c}$. In conclusion, $
        |\hat{\bm{\Sigma}}^{(0)} - \bm{\Sigma}^{(0)}|_{\max} \leq \lambda_0 $ and $|\hat{\bm{\Sigma}}^{(1)} - \bm{\Sigma}^{(1)}|_{\max} \leq \lambda_0 $ holds simultaneously at least $1-\frac{20e^{-1/4}}{3}d^{-c}$.
\end{proof}

\subsection{Proof of Lemma \ref{lem: martingale difference}}
\begin{lemma}
\label{lem: martingale difference}
    Let $Z_t, 1 \leq t \leq n$, be $d$-dimensional martingale difference vectors, $p \geq 1$. Let $s>1$ and $q \geq 2$. Then
\begin{equation}
\label{eq: D1}
\left\|\left|Z_1+\ldots+Z_n\right|_s\right\|_q^2 \leq C_q \cdot s \sum_{t=1}^n\left\|\left|Z_t\right|_s\right\|_q^2
\end{equation}
where $C_q$ is a positive constant depending on $q$ only.
\end{lemma}

\begin{proof}
    Suppose $Z_t, 1 \leq t \leq n$, are martingale difference vectors with respect to the $\sigma$-field $\mathcal{G}_t$. By Theorem 4.1 of \cite{pinelis1994optimum}, we have
$$
\left\|\left|Z_1+\ldots+Z_n\right|_s\right\|_q \leq C_q\left\{\left\|\sup _t\left|Z_t\right|_s\right\|_q+\sqrt{(s-1)}\left\|\left[\sum_{t=1}^n \mathbb{E}\left(\left|Z_t\right|_s^2 \mid \mathcal{G}_{t-1}\right)\right]^{1 / 2}\right\|_q\right\} .
$$
Hence, \eqref{eq: D1} follows in view of the fact that
$$
\left\|\sup _t\left|Z_t\right|_s\right\|_q \leq\left(\sum_{t=1}^n\left\|\left|Z_t\right|_s\right\|_q^q\right)^{1 / q} \leq\left(\sum_{t=1}^n\left\|\left|Z_t\right|_s\right\|_q^2\right)^{1 / 2},
$$
and the inequality
$$
\begin{aligned}
\left\|\left[\sum_{t=1}^n \mathbb{E}\left(\left|Z_t\right|_s^2 \mid \mathcal{G}_{t-1}\right)\right]^{1 / 2}\right\|_q & \leq\left(\sum_{t=1}^n\left\|\mathbb{E}\left(\left|Z_t\right|_s^2 \mid \mathcal{G}_{t-1}\right)\right\|_{q / 2}\right)^{1 / 2} \\
& \leq\left(\sum_{t=1}^n\left\|\left.\left|Z_t\right|\right|_s\right\|_q^2\right)^{1 / 2}
\end{aligned}
$$
which is implied by the triangle inequality and Jensen's inequality.
\end{proof} 

\subsection{Proof of Lemma \ref{lem: Lemma D3}}
\begin{lemma}
\label{lem: Lemma D3}
Let $Z_t, 1 \leq t \leq n$, be $d$-dimensional martingale difference vectors with $d \geq 1$. Let $\ell=1 \vee \log d$ and $q \geq 2$. Then
\begin{equation}
\label{eq: D2}
\left\|\left|Z_1+\ldots+Z_n\right|_{\infty}\right\|_q^2 \leq C_q \cdot \ell \sum_{t=1}^n\left\|\left|Z_t\right|_{\infty}\right\|_q^2,
\end{equation}
where $C_q$ is a positive constant depending on $q$ only.
\end{lemma}
\begin{proof}
If $d=1,2$, \eqref{eq: D2} can be obtained by Theorem 3.2 in \cite{burkholder1973distribution} without much difficulty. If $d \geq 3$, for any vector $v= \left(v_1, \ldots, v_d\right)^{\top}$, we have $|v|_{\infty} \leq|v|_{\log d} \leq d^{1 / \log d}|v|_{\infty}=e|v|_{\infty}$. Hence, \eqref{eq: D2} is an immediate result of Lemma \ref{lem: martingale difference} by letting $s=\log d$.
\end{proof}

\subsection{Proof of Lemma \ref{lem: Sigma hat}}
\begin{lemma}
\label{lem: Sigma hat}
   Under Assumptions in Proposition \ref{lem: gmc(q)}, we have
    \begin{equation}
        \begin{aligned}
            \| |\hat{\bm{\Sigma}}^{(0)} - \bm{\Sigma}^{(0)}|_{\max} \|_{q/2} =O\left( \sqrt{\frac{\log d}{n}} \right)\\
            \text{ and }\| |\hat{\bm{\Sigma}}^{(1)} - \bm{\Sigma}^{(1)}|_{\max}\|_{q/2} =O\left(\sqrt{\frac{\log d }{n}}\right)
        \end{aligned}
    \end{equation}
\end{lemma}

\begin{proof}
By Lemma \ref{lem: Lemma D3}, we have 
$$
\||\hat{\bm{\Sigma}}^{(0)}- {\bm{\Sigma}}^{(0)}|_{\max}\|_{q/2} \leq C_q \sqrt{\frac{\log d}{n}}\sum_{t=1}^n \| |\mathbb{E}_0(X_t - \Bar{X})(X_t - \Bar{X})^\top |_{\max} \|_{q/2} \leq C_{\rho,q} \sqrt{\frac{\log d}{n}} \|X_{\cdot}\|_q^2.
$$
    For any $1\leq i,j \leq d$
    \begin{align*}
        |\bm{\Sigma}^{(1)}_{ij}| = |\mathbb{E}(X_{t+1,i} - \Bar{X}_i)(X_{tj} - \Bar{X}_j )| \leq \|X_{t+1,i} - \Bar{X}_i\|_{q}  \|X_{tj} - \Bar{X}_j\|_{q}   \leq C,
    \end{align*}
we have
\begin{align*}
\||\hat{\bm{\Sigma}}^{(1)} - \bm{\Sigma}^{(1)}|_{\max}\|_{q/2} \leq \||n^{-1} \sum_{t=1}^{n-1} (X_{t} - \Bar{X})(X_{t+1} -\Bar{X})^\top - \bm{\Sigma}^{(1)}|_{\max}\|_{q/2} + d/n|\bm{\Sigma}^{(1)}|_{\max}\\
\leq C_{\rho,q} \sqrt{\frac{\ell}{n-1}} \|X_{\cdot}\|_q^2 + \frac{Cd}{n} \leq C_{\rho,q}'\sqrt{\frac{\log d}{n}}
\end{align*}
\end{proof}

\subsection{Proof of Lemma \ref{lem: kernel + cov}}
\begin{lemma}
\label{lem: kernel + cov}
    Suppose $K(\cdot)$ is a kernel function as in definition \ref{definition: Kernel function}, and the bandwidth $h_n$ satisfies $h_n \rightarrow \infty $ and $h_n \sqrt{\frac{\log d}{n}} \rightarrow 0$ as $n \rightarrow \infty $.  Then, 
    \begin{equation}
        \label{eq: kernel + cov}
\max_{i, j=1, \cdots, d} \left|\frac{1}{n} \sum_{t_1=1}^n \sum_{t_2=1}^n K\left(\frac{t_1-t_2}{h_n}\right){X}_{t_1,i} {X}_{t_2,j} -n \cov (\bar{X}_i, \bar{X}_j)\right| =O_p\left(h_n^{-1}+ h_n \sqrt{\frac{\log d}{n}}\right),
    \end{equation}
\end{lemma}

\begin{proof}
Given the kernel function $K(\cdot)$, by triangle inequality, we have
\begin{equation*}
    \begin{aligned}
        &\left|\frac{1}{n}\sum_{t_1 = 1 }^n \sum_{t_2 =1}^n K\left(\frac{t_1 - t_2}{h_n}\right) X_{t_1,i }X_{t_2,j} - n\cov(\bar{X}_i, \bar{X}_j)\right| \\
        \leq &\left|\frac{1}{n}\sum_{t_1 = 1 }^n \sum_{t_2 =1}^n \left(1 - K\left(\frac{t_1 - t_2}{h_n}\right) \right)\mathbb{E}X_{t_1,i }X_{t_2,j} \right|\\
        &+\left|\frac{1}{n}\sum_{t_1 = 1 }^n \sum_{t_2 =1}^n  K\left(\frac{t_1 - t_2}{h_n}\right) (X_{t_1,i }X_{t_2,j} - \mathbb{E}X_{t_1,i }X_{t_2,j} )\right|
    \end{aligned}
\end{equation*}
On one hand, by \eqref{eq: cov bound}, we have
\begin{equation*}
    \begin{aligned}
       & \left|\frac{1}{n}\sum_{t_1 = 1 }^n \sum_{t_2 =1}^n \left(1 - K\left(\frac{t_1 - t_2}{h_n}\right) \right)\mathbb{E}X_{t_1,i }X_{t_2,j} \right| \\
        \leq & \frac{C}{n}\sum_{t_1 = 1 }^n \sum_{t_2 =1}^n \left(1 - K\left(\frac{t_1 - t_2}{h_n}\right) \right) \sum_{k=0}^\infty \bm{\Delta}_{t_1 -t_2+k,2,i} \bm{\Delta}_{k,2,j}\\
        \leq & C_1 \sum_{s=0}^\infty \left(1 - K\left(\frac{s}{h_n}\right) \right) \sum_{k=0}^\infty \bm{\Delta}_{s+k,2,i}\bm{\Delta}_{k,2,j}\\
        \leq &\frac{C_1 \max_{x \in [0,1]}|K'(x)|}{h_n}\sum_{k=0}^\infty \sum_{s=0}^{h_n}s \bm{\Delta}_{s+k,2,i}\bm{\Delta}_{k,2,j} + \sum_{k=0}^\infty \sum_{s=h_n+1}^\infty \bm{\Delta}_{s+k,2,i}\bm{\Delta}_{k,2,j}= O\left( \frac{1}{h_n}\right)
    \end{aligned}
\end{equation*}
On the other hand,
\begin{equation*}
    \begin{aligned}
       & \left\|\frac{1}{n}\sum_{t_1=1}^n \sum_{t_2 =1}^n K\left(\frac{t_1-t_2}{h_n}\right) X_{t_1, i}X_{t_2,j} -\mathbb{E}X_{t_1,i}X_{t_2,j}\right\|_{q/2}\\
        \leq & \frac{1}{n}\sum_{s=0}^{n-1}K\left(\frac{s}{h_n}\right)\left\|\sum_{t_2=1}^{n-s}X_{t_2 + s,i}X_{t_2,j} - \mathbb{E}X_{t_2 + s,i}X_{t_2,j}\right\|_{q/2} \\
     &   + \frac{1}{n} \sum_{s=0}^{n-1} K\left(\frac{s}{h_n}\right)\left\|\sum_{t_1=1}^{n-s}X_{t_1,i}X_{t_1+s,j} - \mathbb{E}X_{t_1,i}X_{t_1+s,j}\right\|_{q/2}\leq  \frac{C}{n} \sqrt{n} \sum_{s=0}^\infty K\left(\frac{s}{h_n}\right) \leq \frac{C_1 h_n}{\sqrt{n}}
    \end{aligned}
\end{equation*}

    Therefore,
    \begin{equation*}
        \begin{aligned}
            &\left\|\max_{i,j=1,\ldots, d}\left|\frac{1}{n}\sum_{t_1 = 1 }^n \sum_{t_2 =1}^n K\left(\frac{t_1 - t_2}{h_n}\right) X_{t_1,i }X_{t_2,j} - n\cov(\bar{X}_i, \bar{X}_j)\right|\right\|_{q/2}\\
            \leq&\max_{i,j=1,\ldots, d}  \left|\frac{1}{n}\sum_{t_1 = 1 }^n \sum_{t_2 =1}^n \left(1 - K\left(\frac{t_1 - t_2}{h_n}\right) \right)\mathbb{E}X_{t_1,i }X_{t_2,j} \right|\\
        &+\left\| \max_{i,j=1,\ldots, d}\left|\frac{1}{n}\sum_{t_1 = 1 }^n \sum_{t_2 =1}^n  K\left(\frac{t_1 - t_2}{h_n}\right) (X_{t_1,i }X_{t_2,j} - \mathbb{E}X_{t_1,i }X_{t_2,j} )\right|\right\|_{q/2 }=O \left(\frac{1}{h_n} + h_n \sqrt{\frac{\log d}{n}} \right)
        \end{aligned}
    \end{equation*}
\end{proof}

\end{document}